\documentclass[aps,pra,twocolumn,nofootinbib,longbibliography]{revtex4-1} 
\usepackage[T1]{fontenc}
\usepackage[utf8]{inputenc}
\usepackage{verbatim}
\usepackage{dsfont} 
\usepackage{amsmath}
\usepackage{amsthm} 
\usepackage{enumerate}
\usepackage{graphicx}
\usepackage{mathtools}
\usepackage{bm} 
\usepackage{amssymb} 
\usepackage{blindtext}
\usepackage{tcolorbox}
\usepackage{appendix} 
\usepackage{calc}
\usepackage{accents}
\usepackage[normalem]{ulem}
\usepackage{hyperref} 
\hypersetup{
    colorlinks=true,
    linkcolor=green,
    citecolor=magenta,
    unicode=true,  
}
\usepackage[capitalize,nameinlink]{cleveref}

\newtheorem{result}{Result}
\newtheorem{lemma}{Lemma}
\newtheorem{definition}{Definition}
\newtheorem{remark}{Remark}

\newtheorem{theorem}{Theorem}
\newtheorem{corollary}{Corollary}

\newtheorem*{summary*}{Summary of results for measurement informativeness}

\newcommand{\ket}[1]{\left| #1 \right\rangle}

\newcommand{\ketbra}[2]{\left|#1 \rangle \langle #2 \right|}

\DeclareMathOperator{\tr}{Tr}
\DeclareMathOperator{\sgn}{sgn}

\begin{document}

\title{Characterisation of quantum betting tasks in terms of Arimoto mutual information}

\author{Andr\'es F. Ducuara$^{1,2,3,4}$} 
\email[]{andres.ducuara@bristol.ac.uk}

\author{Paul Skrzypczyk$^{3}$}
\email[]{paul.skrzypczyk@bristol.ac.uk}

\affiliation{$^{1}$ Quantum Engineering Centre for Doctoral Training, University of Bristol, Bristol BS8 1FD, United Kingdom \looseness=-1} 

\affiliation{$^{2}$Quantum Engineering Technology Labs, University of Bristol, Bristol BS8 1FD, United Kingdom \looseness=-1}

\affiliation{$^{3}$H.H. Wills Physics Laboratory, University of Bristol, Tyndall Avenue, Bristol, BS8 1TL, United Kingdom \looseness=-1}

\affiliation{$^{4}$Department of Electrical and Electronic Engineering, University of Bristol, Bristol BS8 1FD, United Kingdom \looseness=-1}

\date{\today}

\begin{abstract}
    We introduce operational quantum tasks based on betting with risk-aversion -- or quantum betting tasks for short --  inspired by standard quantum state discrimination and classical horse betting with risk-aversion and side information. In particular, we introduce the operational tasks of quantum state betting (QSB), noisy quantum state betting (nQSB), and quantum channel betting (QCB) played by gamblers with different risk tendencies.  We prove that the advantage that informative measurements (non-constant channels) provide in QSB (nQSB) is exactly characterised by Arimoto's $\alpha$-mutual information, with the order $\alpha$ determining the risk aversion of the gambler. More generally, we show that Arimoto-type information-theoretic quantities characterise the advantage that resourceful objects offer at playing quantum betting tasks when compared to resourceless objects, for general quantum resource theories (QRTs) of measurements, channels, states, and state-measurement pairs, with arbitrary resources. In limiting cases, we show that QSB (QCB) recovers the known tasks of quantum state (channel) discrimination when $\alpha \rightarrow \infty$, and quantum state (channel) exclusion when $\alpha \rightarrow -\infty$.  Inspired by these connections, we also introduce new quantum R\'enyi divergences for measurements, and derive a new family of resource monotones for the QRT of measurement informativeness. This family of resource monotones recovers in the same limiting cases as above, the generalised robustness and the weight of informativeness. Altogether, these results establish a broad and continuous family of four-way correspondences  between operational tasks, mutual information measures, quantum R\'enyi divergences, and resource monotones, that can be seen to generalise two limiting correspondences that were recently discovered for the QRT of measurement informativeness.
\end{abstract} 
\maketitle

\section{Introduction}

The field of quantum information theory (QIT) was born out of the union of the theory of quantum mechanics and the classical theory of information \cite{NC}. This union also happened to kickstart what it is nowadays known as the (ongoing) second quantum revolution which, roughly speaking, aims at the development of quantum technologies \cite{SQR1, SQR2}. Compared with its direct predecessors however, QIT is still a relatively young field and therefore, it is important to keep unveiling, exploiting, and strengthening the links between quantum theory and classical information theory.

In this direction, the framework of quantum resource theories (QRTs) has emerged as a fruitful approach to quantum theory \cite{QRT_QP, RT_review}. A central subject of study within QRTs is that of \emph{resource quantifiers} \cite{QRT_QP, RT_review}. Two well-known families of these measures are the so-called \emph{robustness-based} \cite{RoE, GRoE, RoNL_RoS_RoI, RoS, RoA, RoC, SL, RoT, RoT2, RT_magic, citeme1, citeme2, LBDPS, TR3, FL2020} and \emph{weight-based} \cite{EPR2, WoE, WoS, WoAC} resource quantifiers. Importantly, these quantities have been shown to be linked to operational tasks and therefore, this establishes a type of quantifier-task correspondence. Explicitly, robustness-based quantifiers are linked to discrimination-based operational tasks \cite{RoS, TR1, RoI_task, RoI_Channels, RoA, SL, TR2, RT1}, whilst weight-based resource quantifiers are linked to exclusion-based operational tasks \cite{DS, uola2020}. A resource quantifier is a particular case of a more general quantity known as a \emph{resource monotone} \cite{TG} and therefore, this correspondence can alternatively be addressed as a \emph{monotone-task} correspondence.

From a different direction, in classical information theory, the \emph{Kullback-Leibler (KL) divergence} (also known as the Kullback-Leibler relative entropy) emerges as a central object of study \cite{KL1951}. The importance of this quantity is in part due to the fact that it acts as a parent quantity for many other quantities, such as the Shannon entropy, conditional entropy, conditional divergence, mutual information, and the channel capacity \cite{CT}. Within this classical framework, it has also proven fruitful to consider \emph{R\'enyi-extensions} of these quantities \cite{renyi}. In particular, there is a clear procedure for how to define the R\'enyi-extensions of both Shannon entropy and KL-divergence, which are known as the R\'enyi entropy and the R\'enyi divergence, respectively \cite{renyi, RD}. Interestingly however, there is yet no consensus within the community as to what is the ``proper" way to R\'enyi-extend other quantities. As a consequence of this, there are several different candidates for R\'enyi conditional entropies \cite{review_RCE}, R\'enyi conditional divergences \cite{BLP1}, and R\'enyi mutual information measures \cite{review_RMI}. The latter quantities are also known as measures of dependence \cite{BLP1} or $\alpha$-mutual information measures \cite{review_RMI}, and we address them here as (R\'enyi) \emph{dependence} measures or \emph{mutual informations}. In particular, we highlight the mutual information measures proposed by Sibson \cite{sibson}, Arimoto \cite{arimoto},  Csiszár \cite{csiszar}, as well as one recent proposal, independently derived by Lapidoth-Pfister \cite{LP}, and Tomamichel-Hayashi \cite{TH}. It is known that these mutual information measures (with the exception of Arimoto's) can be derived from their respective conditional R\'enyi divergence \cite{BLP1} and therefore, we address this relationship as a \emph{mutual information-divergence} correspondence. 

The links between the two worlds of QRTs and classical information theory are now beginning to be understood to run much deeper than just the \emph{monotone-task} and \emph{mutual information-divergence} correspondences from above. In fact, they are intimately connected via a more general \emph{four-way} monotone-task-mutual information-divergence correspondence, which holds true in particular for the QRT of measurement informativeness (a QRT where the resource is a measurement's ability to extract information encoded in a state) \cite{SL}. Explicitly, the  robustness-discrimination correspondence \cite{SL, TR2} is furthermore connected to the information-theoretic quantity known as the \emph{accessible information} \cite{Wilde_book} which can, in turn, be written in terms of mutual information measures. In a similar manner the weight-exclusion correspondence \cite{DS, uola2020} is linked to the \emph{excludible information} \cite{DS, MO}, which can also be written in terms of mutual information measures. Even though it was not explicitly stated in any of the above references the fourth corner in terms of ``R\'enyi divergences", it is nowadays a well known fact within the community, first noted by Datta, that the generalised robustness is related to the R\'enyi divergence of order $\infty$ (also called the max quantum divergence) \cite{datta}, with a similar case happening for the  weight and the divergence of order $-\infty$ \cite{DS}. These two apparently ``minor" remarks raise the following fascinating question: Could there exist a whole spectrum of connections between mutual information measures, R\'enyi divergences, resource monotones, and operational tasks, with only the two extreme ends at $\pm \infty$ currently being uncovered? \cite{DS}.

In this work we start by providing a positive answer to this question, by implementing insights from the theory of games and economic behaviour \cite{risk_vNM}. This latter theory, in short, encompasses many of the theoretical tools currently used in the economic sciences. In particular, we invoke here the so-called \emph{expected utility theory} \cite{risk_vNM} and more specifically, we borrow the concept of \emph{risk-aversion}; the behavioural tendency of rational agents to have a preference one way or another for guaranteed outcomes versus uncertain outcomes. This concept remains of great research interest in the economic sciences, with various Nobel prices having been awarded to its understanding \cite{nobel}.

In general, the concept of \emph{risk aversion} is a ubiquitous characteristic of rational agents and, as such, it naturally emerges as a subject of study in various different areas of knowledge such as: the economic sciences \cite{risk_EGS}, biology and behavioral ecology \cite{risk_biology1, risk_biology2}, and neuroscience \cite{NS1, NS2, NS3}. In short, it addresses the behavioural tendencies of rational agents when faced with uncertain events. Intuitively, a gambler spending money on bets with the hope of winning big, can be seen as an individual taking (potentially unnecessary) risks, in the eyes of a more conservative gambler. One of the challenges that economists have tackled, since roughly the second half of the previous century, is the incorporation of the concept of risk aversion into theoretical models describing the behaviour of rational agents, as well as its quantification, and exploitation of its descriptive power \cite{risk_EGS}.

The concept risk was first addressed within theoretical models by Bernoulli in 1738 (translated into English by Sommer in 1954) \cite{risk_bernoulli}. Later on, the theory of expected utility, formalised by von Neumann and Morgenstern in 1944 \cite{risk_vNM}, provided a framework within which to address and incorporate behavioural tendencies like risk aversion. It was then further formalised, independently and within the theory of expected utility, by Arrow, Pratt, and Finetti in the 1950's and 60's \cite{risk_arrow, risk_pratt, risk_finetti} who, in particular, introduced measures for its quantification. The quest for further understanding and exploiting this concept has since remained of active research interest in the economic sciences \cite{risk_EGS}. Recently, an important step was taken in the work of Bleuler, Lapidoth and Pfister (BLP) in 2020 \cite{BLP1}, where the concept of risk aversion was implemented within the realm of classical information theory, as part of the operational tasks of horse betting games with risk and side information.

In this work, inspired by the concepts of \emph{betting}, \emph{risk aversion}, the tasks introduced by BLP \cite{BLP1}, as well as by standard quantum state discrimination, we introduce operational \emph{quantum betting tasks}. Surprisingly, we find that these tasks  turn out to provide the correct approach for solving the conundrum regarding the four-way correspondence for QRTs described above. Specifically, we find that the concept of risk aversion allows us to define operational quantum tasks which can be viewed as a generalisation of discrimination and exclusion.

We start by exploring the QRT of measurement informativeness, and find that Arimoto's $\alpha$-mutual information exactly quantifies the advantage provided by informative measurements when playing one of these quantum betting tasks which we call quantum state betting (QSB). We then explore general QRTs of measurements with arbitrary resources, and similarly derive Arimoto-type information-theoretic measures which quantify the advantage provided by resourceful measurements. Specifically, we find that the concept of \emph{Arimoto's gap}, an information-theoretic quantity which generalises Arimoto's mutual information, characterises QSB games when comparing a resourceful gambler with gamblers with access only to free resources.

In addition to QRTs of measurements, we also explore QRTs of other objects. First, we explore the QRT of non-constant channels. In this scenario we introduce the tasks of noisy quantum state betting (nQSB), and find appropriate Arimoto-type quantities which characterise the performance gain of resourceful objects over resourceless objects in these tasks. Furthermore, we extend these results to QRTs of channels with arbitrary resources, and similarly characterise the advantage provided by resourceful channels in comparison to the best resourceless alternatives.

We also explore the concept of \emph{betting} and \emph{risk-aversion} for tasks beyond QSB and nQSB games, by introducing quantum channel betting (QCB) tasks. We first address these tasks for general single-object QRTs of states with arbitrary resources. In this regime we find that, similarly to the case of QSB and nQSB, there exist Arimoto-type information-theoretic quantities which characterise the performance of resourceful gamblers over resourceless gamblers. We further extend these results to multi-object QRTs of state-measurement pairs. These results therefore altogether highlight that  \emph{betting} and \emph{risk-aversion} are powerful and useful concepts that naturally emerge in general QRTs with arbitrary resources, objects, as well as different tasks.

Finally, we report additional results for the QRT of measurement informativeness, by deriving a continuous four-way correspondence between operational tasks, mutual information measures, R\'enyi divergences, and resource monotones, which generalise correspondences recently found in the literature \cite{SL, DS}.

We believe that the concepts of \emph{betting} and \emph{risk-aversion} have the potential to positively impact  our understanding of the framework of resource theories as well as our understanding of the theory of quantum information more generally. 

This work is organised as follows. In Sec.~\ref{s:risk} we start by describing the concept of  risk aversion in the theory of games and economic behaviour. In Sec.~\ref{s:arimoto} and \ref{s:arimoto2} we address Arimoto's mutual information measure and the R\'enyi capacity both in classical and quantum domains. In Sec.~\ref{s:RoMI} we describe the QRT of measurement informativeness and the QRT of non-constant channels. In Sec.~\ref{s:gap} we address further Arimoto-type information-theoretic quantities for general QRTs of measurements, channels, states, and state-measurement pairs with arbitrary resources. Our main results sections start in Sec.~\ref{s:QBT}, where we introduce operational quantum tasks based on betting with risk-aversion, or quantum betting tasks for short, and introduce various tasks as follows: quantum state betting (QSB) in Sec.~\ref{ss:QSB}, \ref{ss:QSB1}, \ref{ss:QSB2}, noisy quantum state betting (nQSB) in Sec.~\ref{ss:nQSB}, and quantum channel betting (QCB) in Sec.~\ref{ss:QCB}. In Sec.~\ref{s:results} we address the characterisation of quantum betting tasks in terms of Arimoto-type information-theoretic quantities. In  Sec.~\ref{ss:result1} we relate QSB games to  Arimoto's mutual information, for the QRT of measurement informativeness. In Sec. \ref{ss:result2} we characterise noisy QSB (nQSB) games in terms of a noisy Arimoto mutual information, for the QRT of non-constant channels. In Sec. \ref{ss:result3} we characterise QSB and nQSB games in terms of Arimoto-type quantities, for general QRTs of measurements and channels with general resources. In Sec. \ref{ss:result4} we characterise QCB games in terms of Arimoto-type measures for single-object QRTs of states with arbitrary resources as well as multi-object QRTs of state-measurement pairs with arbitrary resources.  In Sec. \ref{ss:result5} we characterise horse betting games in terms of the Arimoto's mutual information in the classical regime, without invoking quantum theory. In Sec.~\ref{e:divergences} and \ref{e:monotones} we address quantum R\'enyi divergences and resource monotones, and derive a four-way correspondence for the QRT of measurement informativeness. We finish in Sec.~\ref{s:conclusions} with conclusions, open questions, perspectives, and avenues for future research.

\section{Background theory}

In this section we address the preliminary theoretical tools necessary to establish our main results. We start with the concept of risk in the theory of games and economic behaviour. We then introduce a pair of games involving risk. After this, we introduce Arimoto's $\alpha$-mutual information measure and the R\'enyi capacity in both classical and quantum information theory. Next, we review the QRTs of measurement informativeness and non-constant channels and, finally, Arimoto-type information-theoretic measures for general QRTs of measurements, channels, states, and state-measurement pairs.

\subsection{The concept of risk in the theory of games and economic behaviour}
\label{s:risk}

In expected utility theory \cite{risk_vNM}, the level of `satisfaction' of a rational agent, when \emph{receiving} (obtaining, being awarded) a certain amount of wealth, or goods or services, is described by a utility function \cite{risk_vNM}. The utility function of a rational agent is a function $u:A\rightarrow \mathds{R}$, with $A=\{a_i\}$ a the set of \emph{alternatives} from which the rational agent can choose from. The set $A$ is endowed with a binary relation $\preceq$. The utility function is asked to be a monotone for such a binary relation; if $a_1 \preceq a_2$ then $u(a_1) \leq u(a_2)$. In this work we address the set of alternatives as representing \emph{wealth} and therefore, it is enough to consider an interval of the real numbers. 

We are going to consider two different types of situation in this work. In the first case, the wealth will always be non-negative, and so we consider the interval being $A=\mathcal{I}=[0,w^M]\subseteq \mathds{R}$, with $w^M>0$ a maximal amount of wealth, and the standard binary relation $\leq$. Similarly, we also will also consider a situation where the wealth is  \emph{non-positive}, meaning we address a utility function taking negative arguments $w<0$, with $\mathcal{I}=[-w^M,0]\subseteq \mathds{R}$, as the level of (dis)satisfaction when the rational agent has to \emph{pay} an amount of money $|w|$ (or when the amount $|w|$ is taken away from him). 

We note here that the utility function does not necessarily need to be positive (or negative), because it is only used to \emph{compare} alternatives. The condition that the utility function is monotonic is the equivalent to it being an increasing function for both positive and negative wealth. Intuitively, this represents that the rational agent is interested in acquiring as much wealth as possible (for positive wealth), and losing the least amount of wealth as possible (for negative wealth). Additionally, the utility function is asked to be twice-differentiable, both for mathematical convenience and, because it is natural to assume that smooth changes in wealth imply smooth changes in the rational agent's satisfaction.

In order to address the concept of risk, we first need to introduce two games (or operational tasks), which involves a player Bob (the \emph{Better} or \emph{Gambler}, who we take to be a rational agent with a utility function $u$) and a referee Alice, who is in charge of the game. We are going to address two different games which we call here: i) \emph{gain games} and ii) \emph{loss games}.

\subsubsection{A gain game and utility theory}

In a \emph{gain game}, Alice (Referee) offers Bob (Gambler) the choice between two options: i) a fixed \emph{guaranteed} amount of wealth $w^G \in [0,w^M]$ or ii) a \emph{bet}. The bet consists of the following: Alice uses a random event distributed according to a probability mass function (PMF) $p_W$, (i.e.~$\sum_{w\in \mathcal{I}} p_W(w)=1$, $p_W(w)\geq 0$, $\forall w\in \mathcal{I}$,  with $W$ a random variable in the alphabet $\mathcal{I}$), in order to give Bob a reward. Specifically, Alice will reward Bob with an amount of wealth $w^B=w$, whenever the random event happens to be $w$, which happens with probability $p(w)$ (we drop the label $W$ on $p_W(w)$ from now on). The choice facing Bob is therefore between a fixed guaranteed amount of wealth $w^G \in [0,w^M]$, or taking the bet and potentially earning more $w^B>w^G$, at the risk of earning less $w^B<w^G$. 

Since the utility function $u(w)$ determines Bob's satisfaction when acquiring the amount wealth $w$, we will see below that it can be used to model his behaviour in this game, i.e. whether he chooses the first or second option.  First, considering the bet (option ii) we can consider the \emph{expected gain} of Bob at the end, 
\begin{align}
    \mathbb{E}[W]
    =
    \sum_{w\in \mathcal{I}}
    p(w)
    w.
\end{align} 
How satisfied Bob is with this expected amount of wealth is given by the utility of this value, i.e.
\begin{align}
    u
    \left(
    \mathbb{E}[W]
    \right)
    =
    u
    \left(
    \sum_{w\in \mathcal{I}}
    p(w)
    w
    \right)
    .
\end{align}
Now, Bob's wealth at the end of the bet is a random variable, this means that his satisfaction will also be a random variable, with some uncertainty. We can also ask what Bob's expected satisfaction, i.e.~\emph{expected utility} will be at the end of the bet,
\begin{align}
    \mathbb{E}[u(W)]
    =
    \sum_{w\in\mathcal{I}}
    p(w)
    u(w)
    .
\end{align}
This represents how satisfied Bob will be with the bet on average. 

We can now introduce the first key concept, that of the \emph{Certainty Equivalent (CE)}: it is the amount of (certain) wealth $w^{ICE}$ which Bob is as satisfied with as the average wealth he would gain from the bet. In other words, the amount of wealth which is as desirable as the bet itself. That is, it is the amount of wealth $w^{ICE}$ that satisfies
\begin{align}
    \mathbb{E}[u(W)] =u(w^{ICE}). \label{e:wCE}
\end{align}
It is crucial to note that the certainty equivalent wealth depends upon the utility function $u$ and the PMF $p_W$, and therefore we interchangeably write it as $w^{ICE}(u,p_W)$. We can now return to the original game, i.e.~the choice between a fixed return $w^G$, or the average return $\mathbb{E}[W]$. The rational decision for Bob is to pick which of the two he is most satisfied with. We now see that if we set $w^G > w^{ICE}$ then he will choose to take the guaranteed amount, if $w^G < w^{ICE}$ he will choose the bet, and if $w^G = w^{ICE}$ then in fact the two options are equivalent to Bob, and he can rationally pick either. That is, we see that the certainty equivalent $w^{ICE}$ sets the boundary between which option Bob will pick. 

Introducing the certainty equivalent moreover allows us to introduce the concept of Bob's \emph{risk-aversion}. To do so, we will compare Bob's expected wealth, in relation to the certainty equivalent of the bet. There are only three possible scenarios,
\begin{align}
     w^{ICE} &< \mathbb{E}[W], \label{eq:R1} \\
	 w^{ICE} &> \mathbb{E}[W], \label{eq:R2} \\
     w^{ICE} &= \mathbb{E}[W]. \label{eq:R3}
\end{align}
In the first case \eqref{eq:R1}, Alice can offer Bob an amount of wealth $w^G$ that is larger than $w^{ICE}$ but less than $\mathbb{E}[W]$, $w^{ICE} < w^G < \mathbb{E}[W]$ and Bob will rationally take this amount over accepting the bet, even though he will walk away with less wealth on average than if he took the bet. In other words, Bob is \emph{reluctant} to take the bet, and so we say that he is \emph{risk-averse}. 

In the second case \eqref{eq:R2}, on the other hand, if Alice wants to make Bob walk away from the bet, and accept a fixed amount of wealth instead, she will have to offer him \emph{more} than the expected gain. That is, Bob will only choose an amount $w^G$ if $w^G > w^{ICE} > \mathbb{E}[W]$. Here Bob is \emph{risk-seeking}. 

Finally, in the third case \eqref{eq:R3}, Bob will take the bet if Alice offers him any $w^G$ less than the expected gains from the bet, and will take the guaranteed amount $w^G$ if it is larger. In this case, we say that Bob is \emph{risk-neutral}, as Bob is essentially indifferent between the uncertain gains of the bet and the certain gains of the guaranteed return.  

If we recall that by definition the utility function $u$ is strictly increasing in the interval $\mathcal{I}$ (more wealth is also more satisfactory to Bob), then by applying  $u$ to the previous three equations, and using the definition of $w^{ICE}$ \eqref{e:wCE}, we get
\begin{align}
	\mathbb{E}[u(W)] &< u(\mathbb{E}[W]), \label{eq:Rp1} \\
	\mathbb{E}[u(W)] &> u(\mathbb{E}[W]), \label{eq:Rp2} \\
	\mathbb{E}[u(W)] &= u(\mathbb{E}[W]). \label{eq:Rp3}
\end{align}
This is an important result, which shows that Bob's risk-aversion is characterised by the curvature of his utility function: Bob is risk-averse when his utility function is concave \eqref{eq:Rp1}, risk-seeking when his utility function is convex \eqref{eq:Rp2}, and risk-neutral when it is linear \eqref{eq:Rp3}. This intuitively makes sense, since roughly speaking this corresponds to his satisfaction growing more slowly than wealth when he is risk-averse and his satisfaction growing faster than wealth when he is risk-seeking. We now move on to analyse the concept of risk in our second game.

\subsubsection{A loss game and utility theory}

Let us now analyse a game which we call here a \emph{loss game}. Similarly to the gain game from the previous section, in an loss game we have two agents, a Referee (Alice) and a Gambler (Bob), who has to make a payment to the Referee. In an loss game Bob is now asked to choose between two options: i) paying a fixed amount of wealth $|w^{F}|$, $w^F\in[-w^M,0]$ or ii) a bet. Choosing the bet means Bob has to pay an amount of wealth according to the outcome of a PMF $p_W$. Similarly to the gain game, we address some quantities of interest: \emph{expected debt} ($\mathds{E}(W)$), \emph{expected utility} ($\mathds{E}[u(W)]$), and the \emph{certainty equivalent (CE)} $w^{ICE}(u,p_W)$, as the amount of wealth $w^{ICE}$ such that $u(w^{ICE}) = \mathds{E}[u(W)]$. We note the CE depends on the utility function $u$ representing the Player, and the PMF $p_W$ representing the bet. The CE is the amount of wealth that Bob pays to Alice, which generates the same level of (dis)satisfaction, had Bob opted for the bet instead. We also note here that both the expected debt and the certainty equivalent are now negative quantities.

We now analyse the meaning of the certainty equivalent in loss games, i.e., where Bob (the Gambler) has to choose between having to pay a certain fixed amount of wealth (fixed debt) $|w^F|$, or paying an average amount (average debt) $|\mathbb{E}[W]|$. The rational decision for Bob is to pick which of the two options he is \emph{more satisfied} (equivalently, we could say least dissatisfied) with. We then see that if we set $w^F < w^{ICE}$ he then will choose to take the bet, if $w^F > w^{ICE}$ he will choose to pay the fixed amount, and if $w^F = w^{ICE}$ he can rationally pick either. That is, we see that the certainty equivalent $w^{ICE}$ again sets here the boundary between which option Bob will pick in an loss game. 

We now compare Bob's expected debt $\mathbb{E}[W]$ and the certainty equivalent of the bet $w^{ICE}$. We have the three possible scenarios, 
\begin{align}
     w^{ICE} & < \mathbb{E}[W], 
     \hspace{0.3cm}
     \longleftrightarrow
     \hspace{0.3cm}
     |w^{ICE}|  > |\mathbb{E}[W]|,
     \label{eq:D1} \\
	 w^{ICE} & > \mathbb{E}[W], 
	 \hspace{0.3cm}
     \longleftrightarrow
     \hspace{0.3cm}
	 |w^{ICE}|  < |\mathbb{E}[W]|,
	 \label{eq:D2} \\
     w^{ICE} & = \mathbb{E}[W],
     \hspace{0.3cm}
     \longleftrightarrow
     \hspace{0.3cm}
     |w^{ICE}|  = |\mathbb{E}[W]|.
     \label{eq:D3}
\end{align}
In the first case \eqref{eq:D1}, Alice can \emph{request} from Bob a fixed amount of wealth $|w^F|$
as $w^{ICE} < w^{F} < \mathbb{E}[W]$, which is equivalent to $|w^{ICE}| > |w^{F}| > |\mathbb{E}[W]|$ and Bob will still \emph{prefer} to pay this amount over opting for the bet, even though he will \emph{potentially have to pay less} $|\mathds{E}(W)|$, on average, had he opted for the bet. In other words, Bob is \emph{reluctant} to take the bet, and so we see that he is \emph{risk-averse}.

In the second case \eqref{eq:D2}, if Alice wants to make Bob walk away from choosing the bet, and accept \emph{paying} a fixed amount of wealth instead, she will have to offer him a deal where he has to pay \emph{less} than the CE (and in turn less than the expected debt). In other words, in this case Bob is confident that the bet will allow him to pay less than the expected debt. That is, Bob will choose paying a fixed amount $|w^F|$ only if $w^{F} > w^{ICE} > \mathbb{E}[W]$, which is equivalent to $|w^{F}| < |w^{ICE}| < |\mathbb{E}[W]|$. Here Bob can then be considered as \emph{risk-seeking}, because he is hopeful/optimistic about having the chance of paying less than the expected debt.

Taking into account the utility function is still an strictly increasing function for negative wealth, together with the definition of the certainty equivalent we get:
\begin{align}
	\mathbb{E}[u(W)] &< u(\mathbb{E}[W]), \label{eq:Dp1} \\
	\mathbb{E}[u(W)] &> u(\mathbb{E}[W]), \label{eq:Dp2} \\
	\mathbb{E}[u(W)] &= u(\mathbb{E}[W]). \label{eq:Dp3}
\end{align}
This means that in an loss game we can also characterise the risk tendencies of a Gambler in terms of the concavity/convexity/linearity of his utility function as: risk-averse (concavity \eqref{eq:Dp1}), risk-seeking (convexity \eqref{eq:Dp2}), risk-neutral (linear \eqref{eq:Dp3}). This characterisation of risk tendencies and the types of games are going to be useful later on when introducing more elaborate games involving the discrimination or exclusion of quantum states. We now move on to the quantification of risk.

\subsubsection{Quantifying risk tendencies}

We can go one step further, and not only classify whether Bob (the Gambler) is risk-averse, risk-seeking, or risk-neutral, but moreover \emph{quantify} how risk-averse he is. Let us start by addressing a \emph{gain game}, which means we are interested in analysing Bob being represented by an utility function on positive wealth. Since Bob's attitude toward risk relates to the concavity/convexity/linearity of the utility function $u$, it is natural that the second derivative of the function is going to play a role. This, because $u$ is concave on an interval if and only if its second derivative is non-positive on that interval. However, it is also desirable for measures representing risk to be invariant under \emph{affine transformations} of the utility function, which in this context means that they are invariant under transformations of the form $u \rightarrow a+b u$, with $a, b \in \mathds{R}$. This is because the actual values of utility aren't themselves physical, but only the \emph{comparison} between values, and therefore rescaling or displacing the utility should not alter how risk-averse we quantify Bob to be. Given these requirements, a natural measure that emerges is the so-called Relative Risk Aversion (RRA) measure\footnote{An additional benefit of this quantifier is that it is dimensionless, which is not satisfied by all quantifiers of risk-aversion}:
\begin{align}
    RRA(
    w
    )
    \coloneqq
    -
    w
    \frac{
    u^{''}
    (w)
    }{
    u'(w)
    }.
    \label{e:RRA}
\end{align}
This measure assigns \emph{positive} values for risk-averse players in a gain game (\emph{concave} utility functions of positive wealth) because we have: i) $w>0$, because we are considering the player \emph{receiving} money ii) $u^{''}(w)<0$, $\forall w$, because a risk-averse player in a gain game is represented by a concave function, and iii) $u^{'}(w)>0$, because the utility function is a strictly increasing function. An analysis of signs then yields $RRA(w)>0$.

Similarly, we now also analyse this measure of risk-aversion when Bob plays a \emph{loss game}. A loss game is characterised by \emph{negative wealth}, and we have already derived the fact that that a risk-averse Gambler is also characterised by a \emph{concave} utility function. We now want to quantify the degree of risk-aversion of a Gambler playing the loss game, and therefore we then can proceed in a similar fashion as before, and define the risk-aversion measure RRA. 

We now check that this measure assigns \emph{negative} values for risk-averse players in a loss game (\emph{concave} utility functions of negative wealth) because we have: i) $w<0$ because we are considering the player \emph{paying} money ii) $u^{''}(w)<0$, $\forall w$, because a risk-averse player in a loss game is represented by a concave function, and iii) $u^{'}(w)>0$, because the utility function is a strictly increasing function. An analysis of signs yields $RRA(w)<0$. We can see that this is the opposite to what happens in gain games, where $RRA(w)>0$ represents risk-averse players. We highlight this fact in Table~\ref{tab:tab1}, and present an analysis of the sign of the RRA measure for the two types of players (risk-averse or risk-seeking) and the two types of games (gain game or loss game).
\begin{table}[h!]
    \centering
    \begin{tabular}{|c||c|c|}
    \hline
        & 
        Risk-averse player 
        & 
        Risk-seeking player
        \\
        & $u^{''}(w)<0$ & $u^{''}(w)>0$
        \\
        \hline \hline
        $w>0$ &  $RRA(w)>0$ & $RRA(w)<0$
        \\
        \hline
        $w<0$ &  $RRA(w)<0$ & $RRA(w)>0$
        \\
        \hline
    \end{tabular}
    \caption{Analysis of the sign of the quantity $RRA(w)$ for the different regimes being considered. We have that the utility function is always strictly increasing, meaning that $u^{'}(w)>0$, and therefore we then only need to analyse the signs of $w$ and $u^{''}(w)$. In particular, we have that risk-averse players are represented by positive RRA when dealing with positive wealth, and by negative RRA when dealing with negative wealth.}
    \label{tab:tab1}
\end{table}

\subsubsection{The isoelastic certainty equivalent}

We now note that the RRA measure does not assign a global value for how risk averse Bob is, but allows this to depend upon the wealth $w$, i.e. Bob may be more or less risk averse depending on the wealth that is at stake. In order to remove this, it is usual to consider those utility functions where Bob's relative risk aversion is \emph{constant}, independent of wealth. In this case, \eqref{e:RRA} can be solved assuming $RRA(w) = R$, which leads to the so-called \emph{isoelastic utility function} for positive and negative wealth as:
\begin{align}
    u_R(w)
    \coloneqq
    \begin{cases}
        \sgn(w)
        \frac{|w|^{1-R}-1}{1-R}, 
        & \text{if}\ R \neq 1 \\
        \sgn(w)
        \ln(|w|), 
        & \text{if}\ R = 1
    \end{cases}
    ,
    \label{eq:isoelastic}
\end{align}  
with the auxiliary ``sign" function:
\begin{align}
    \sgn(
    w
    )
    \coloneqq
    \begin{cases}
        1, & 
        w
        \geq 0;
        \\
        -1,& 
        w
        < 0.
    \end{cases}
    \label{eq:sgn}
\end{align}
The parameter $R$ varies from minus to plus infinity, describing all possible risk tendencies of Bob, for either positive or negative wealth. For positive wealth for instance, $R$ goes from maximally risk-seeking at $R=-\infty$, passing through risk-neutral at $R=0$, to maximally risk-averse at $R=\infty$. In Fig.~\ref{fig:fig0} we can see the behaviour of the isoelastic function for positive wealth and different values of $R$.
\begin{figure}[h!]
    \centering
    \includegraphics[scale=0.57]{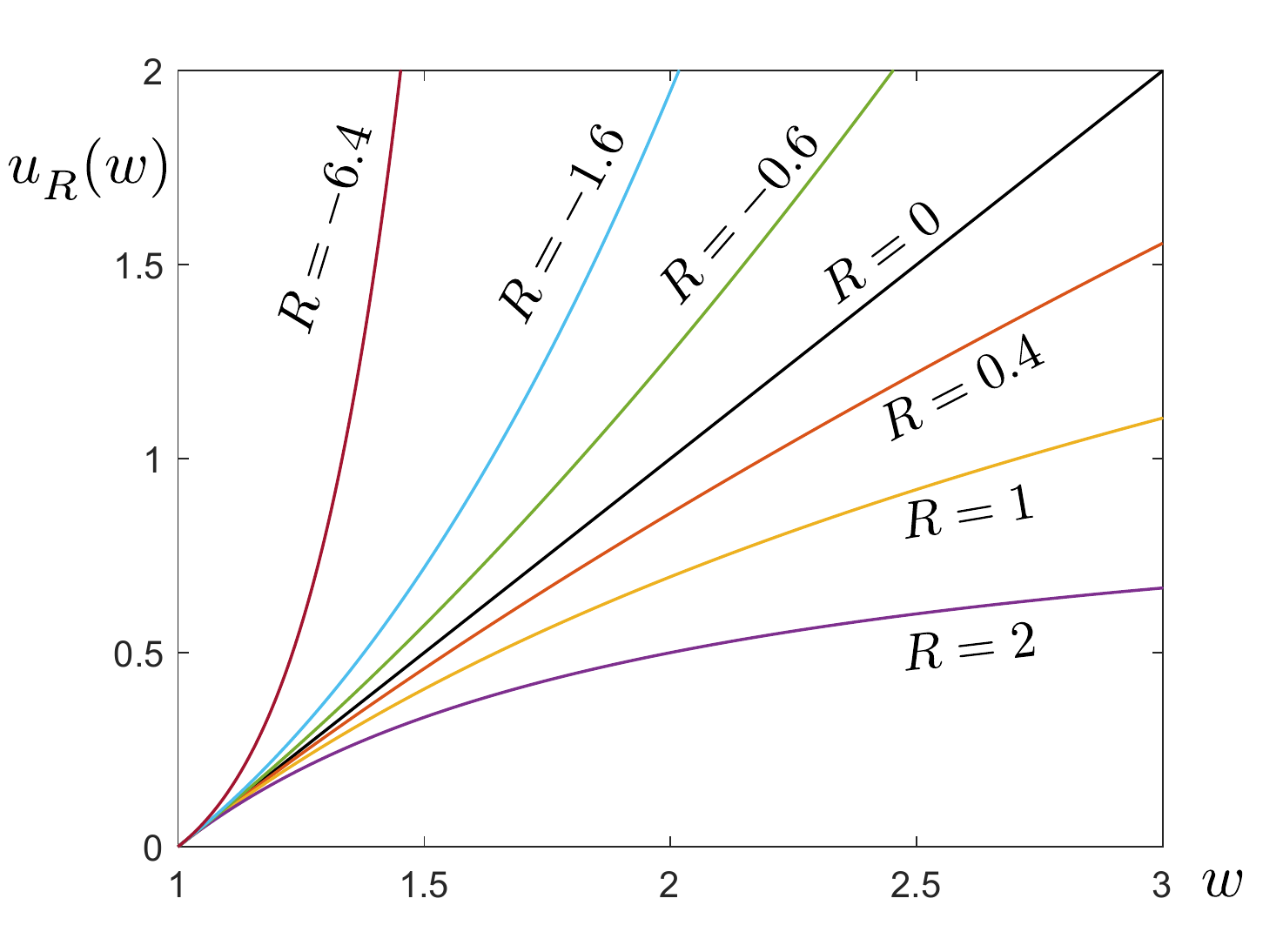}
    \vspace{-0.5cm}
    \caption{
    Isoelastic utility function $u_R(w)$ \eqref{eq:isoelastic} as a function of positive wealth ($1\leq w \leq 3$) for players with different risk tendencies (different values of $R$). The risk parameter $R$ quantifies different types of risk tendencies: i) $R<0$ risk-seeking players (convex) ii) $R=0$ risk-neutral players (linear), and iii) $R>0$ risk-averse players (concave). Risk-aversion for positive wealth then increases from $-\infty$ to $\infty$.
    }
    \label{fig:fig0}
\end{figure}

The certainty equivalent \eqref{e:wCE} for this setup can be calculated for either positive or negative wealth as:
\begin{align}
    w^{ICE}_R
    =
    u^{-1}_R
    \left(
    \mathbb{E}[u_R(W)]
    \right)
    =
    \left(
    \sum_{w\in\mathcal{I}}
    w^{1-R}
    \,
    p(w)
    \right)^{\frac{1}{1-R}}
    .
    \label{eq:wCEF}
\end{align}
The certainty equivalent (CE) of the isoelastic function, or \emph{isoelastic certainty equivalent} (ICE), is going to be the figure of merit in the next section, and it is going to play an important role in this paper. As we have already seen, the CE stands out as an important quantity because it: i) determines the choice of a Gambler when playing either a gain or loss game, helping to establish the characterisation of risk tendencies of said Gambler and ii) optimising the CE is equivalent to optimising the \emph{expected utility}, given that the utility function is a strictly increasing function and that $u(w^{ICE}) = \mathds{E}[u(W)]$. One may be tempted here to propose the \emph{expected utility} function $\mathds{E}[u(W)]$ as the figure of merit instead of the CE, but the expected utility unfortunately suffers from having the rather awkward set of units $[w]^{1-R}$, whilst the certainty equivalent on the other hand has simply units of wealth $[w]$ (\$, \pounds, ...).

\subsection{Arimoto's $\alpha$-mutual information and R\'enyi channel capacity}\label{s:arimoto}

We start this subsection by introducing the \emph{$\alpha$-mutual information measures} of interest (also known as dependence measures \cite{BLP1,review_RMI}) and particularly, Arimoto's $\alpha$-mutual information \cite{arimoto}. A reminder note on notation before we start: we consider random variables (RVs) ($X, G,...$) on a finite alphabet $\mathcal{X}$, and the probability mass function (PMF) of $X$ represented as $p_X$ satisfying: $p_X(x)\geq 0$, $\forall x\in \mathcal{X}$, and $\sum_{x\in \mathcal{X}}p_X(x)=1$. For simplicity, we omit the alphabet when summing, and write $p_X(x)$ as $p(x)$ when evaluating. The support of $p_X$ as ${\rm supp} (p_X) \coloneqq \{x\,|\,p(x)>0\}$, the cardinality of the support as $|{\rm supp}(p_X)|$, and the extended line of real numbers as $ \mathds{ \overline R}\coloneqq \mathds{R} \cup \{\infty,-\infty\}$. We now start by considering the R\'enyi entropy. 

\begin{definition} (R\'enyi entropy \cite{renyi})
	The R\'enyi entropy of order $\alpha \in \mathds{\overline R}$ of a PMF $p_X$ is denoted as $H_{\alpha}(X)$. The orders $\alpha \in(-\infty,0)\cup(0,1)\cup(1,\infty)$ are defined as:
	\begin{align}
		H_{\alpha}(X)
		& \coloneqq
		\frac{1}{1-\alpha}
		\log
		\left(
		\sum_x
		p(x)^{\alpha}
		\right)
		.
		\label{eq:RE}
	\end{align}
	The orders $\alpha \in\{0,1,\infty,-\infty\}$ are defined by continuous extension of \eqref{eq:RE} as: $H_0(X)\coloneqq \log |{\rm supp}(p_X)|$, $H_1(X)\coloneqq H(X)$, with $H(X) \coloneqq -\sum_x p(x)\log p(x)$ the Shannon entropy \cite{CT}, $H_{\infty}(X)\coloneqq-\log \max_x p(x)=-\log p_{\rm max}$, and $H_{-\infty}(X)\coloneqq -\log \min_x p(x) = - \log p_{\rm min}$. The R\'enyi entropy is a function of the PMF $p_X$ and therefore, one can alternatively write $H_\alpha(p_X)$. However, we keep the convention of writing $H_{\alpha}(X)$.
\end{definition}
The R\'enyi entropy is mostly considered for positive orders, but it is also sometimes explored for negative values \cite{NV1, NV2, SR1, SR2}. In this work we use the whole spectrum $\alpha \in \mathds{\overline R}$. We now consider the Arimoto-R\'enyi extension of the conditional entropy.
\begin{definition} (Arimoto-R\'enyi conditional entropy \cite{arimoto}) The Arimoto-R\'enyi conditional entropy of order $\alpha \in \mathds{\overline R}$ of a joint PMF $p_{XG}$ is denoted as $H_{\alpha}(X|G)$. The orders $\alpha \in (-\infty,0) \cup (0,1) \cup (1,\infty)$ are defined as:
\begin{align}
	H_{\alpha}(X|G)
	& \coloneqq
	\frac{\alpha}{(1-\alpha)}
	\log
	\left[
	\sum_g
	\left(
	\sum_x
	p(x,g)^\alpha
	\right)^\frac{1}{\alpha}	
	\right]
	.
	\label{eq:ARCE}
\end{align}
	The orders $\alpha \in\{0,1,\infty,-\infty\}$ are defined by continuous extension of \eqref{eq:ARCE} as: $H_{0}(X|G) \coloneqq \log \max_g |{\rm supp}(p_{X|G=g})|$, $H_{1}(X|G) \coloneqq H(X|G)$, with $H(X|G) \coloneqq -\sum_{x,g}p(x,g)\log p(x|g)$ the conditional entropy \cite{CT}, $H_{\infty}(X|G)
	\coloneqq
	-\log 
	\sum_g
	\max_x
	p(x,g)$, and $H_{-\infty}
	(X|G)
	\coloneqq
	-
	\log 
	\sum_g
	\min_x
	p(x,g)$. Arimoto-R\'enyi conditional entropy is a function of the joint PMF $p_{XG}$ and therefore, one can alternatively write $H_\alpha(p_{XG})$. However, we keep the convention of writing $H_{\alpha}(X|G)$.
\end{definition}
We remark that there are alternative ways to R\'enyi-extend the conditional entropy \cite{review_RCE}. The Arimoto-R\'enyi conditional entropy is however, the only one (amongst five alternatives \cite{review_RCE}) that simultaneously satisfy the following desirable properties for a conditional entropy \cite{review_RCE}: i) monotonicity, ii) chain rule, iii) consistency with the Shannon entropy, and iv) consistency with the $\infty$ conditional entropy (also known as min-entropy). Consistency with the conditional entropy means that $\lim_{\alpha \rightarrow 1}H_{\alpha}(X|G) = H(X|G)$, and similarly for property iv). In this sense, one can think about the Arimoto-R\'enyi conditional entropy as the ``most appropriate" R\'enyi-extension (if not the outright ``proper" R\'enyi extension) of the conditional entropy. We now consider Arimoto's mutual information, and its associated R\'enyi channel capacity 
\begin{definition} (Arimoto's $\alpha$-mutual information \cite{arimoto})
	Arimoto's mutual information of order $\alpha \in\mathds{\overline R}$ of a joint PMF $p_{XG}$ is given by:
	\begin{align}
		I
		_\alpha(X;G)
		&\coloneqq
		\sgn(\alpha)
		\left[
		H_{\alpha}(X)
		-
		H
		_{\alpha}(X|G)
		\right]
		,
		\label{eq:DMA}
	\end{align}	
	with the R\'enyi entropy \eqref{eq:RE} and the Arimoto-R\'enyi conditional entropy \eqref{eq:ARCE}. The case $\alpha=1$ reduces to the standard mutual information \cite{CT} $I_1(X;G) = I(X;G)$, with $I(X;G) \coloneqq H(X)-H(X|G)$. Arimoto's $\alpha$-mutual-information is a function of the joint PMF $p_{XG}$ and therefore, one can alternatively write $I_\alpha(p_{XG})$ or $I_\alpha(p_{G|X}p_X)$, the latter taking into account that $p_{XG} = p_{G|X}p_X$. We use these three different notations interchangeably. 
\end{definition}
\begin{definition} (R\'enyi channel capacity \cite{arimoto, csiszar, remarks, Nakiboglu}) 
The R\'enyi channel capacity of order $\alpha \in \mathds{\overline R}$,  of a conditional PMF $p_{G|X}$ is given by:
\begin{align}
    C_{\alpha}
		(p_{G|X})
    \coloneqq
      \max_{p_X}
		I_{\alpha}
		(p_{G|X}p_X)
    \label{eq:iso}
\end{align}
The case $\alpha=1$ reduces to the standard channel capacity \cite{CT} $C_1(p_{G|X})=C(p_{G|X})=\max_{p_X}I(X;G)$.
\end{definition}
We remark that there are alternative candidates as R\'enyi-extensions of the mutual information \cite{review_RCE, review_RMI}. In particular, we highlight the $\alpha$-mutual information measures of Sibson \cite{sibson}, Csisz\'ar \cite{csiszar}, and Bleuler-Lapidoth-Pfister \cite{BLP1}, which we address in the appendices as $I^{\rm V}_\alpha(X;G)$ with the label $\rm V\in\{S,C,BLP\}$ representing each case. These $\alpha$-mutual informations are going to be useful, in particular, due to their connection to conditional R\'enyi divergences. We address these information-theoretic quantities in \cref{AA}. We now extend these information-theoretic quantities to the quantum domain.

\subsection{Arimoto's $\alpha$-mutual information and R\'enyi channel capacity in a quantum setting}
\label{s:arimoto2}

We now move on to describe Arimoto's $\alpha$-mutual information in a quantum setting, as well as the R\'enyi channel capacity.
\begin{remark}(Arimoto's $\alpha$-mutual information in a quantum setting)
	We address Arimoto's $\alpha$-mutual information between two classical random variables encoded into quantum objects. Explicitly, the random variable $X$ is  encoded in an ensemble of states $\mathcal{E}=\{\rho_x,p(x)\}$ and therefore, we address it as $X_\mathcal{E}$. On the other hand, $G$ is considered as the random variable obtained from a decoding measurement $\mathds{D}=\{D_g = \ketbra{g}{g}\}$ and therefore, we address it as $G_\mathbb{D}$. We consider a conditional PMF as $p_{G|X}^{(\mathbb{M},\mathcal{S})}$, given by $p(g|x)\coloneqq \tr [D_g \Lambda_\mathbb{M}(\rho_x)]$, $\mathcal{S} \coloneqq \{\rho_x\}$ a set of states, and the quantum-to-classical (measure-prepare) channel associated to the measurement $\mathbb{M}$ given by:
	\begin{align}
    	\Lambda_\mathbb{M}(\sigma)
    	\coloneqq
    	\sum_a
    	\tr
    	[M_a\sigma]
    	\ketbra{a}{a},
    	\label{eq:qc}
	\end{align} 
	with $\{\ket{a}\}$ an orthonormal basis. We effectively have $p(g|x)\coloneqq \tr [M_g \rho_x]$ and therefore we can think about the decoding variable $G_{\mathbb{D}, \mathcal{E}}$ as $G_{\mathbb{M}, \mathcal{E}}$. We are now interested in the $\alpha$-mutual information quantifying the dependence between variables $X_\mathcal{E}$ and $G_{\mathbb{M}, \mathcal{E}}$, when encoded and decoded in the quantum setting described previously. We then consider Arimoto's $\alpha$-mutual information:
	\begin{align}
		I_\alpha
		(X;G)_{\mathcal{E}, \mathbb{M}}
		&\coloneqq 
		\sgn(\alpha)
		\left[
		H_{\alpha}
		(X)_\mathcal{E}
		-
		H_{\alpha}
		(X|G)_{\mathcal{E}, \mathbb{M}}
		\right]
		,
	\end{align}
	with the standard R\'enyi entropy \eqref{eq:RE} and the Arimoto-R\'enyi conditional entropy \eqref{eq:ARCE} for the quantum conditional PMF described above. In similar manner, we are also interested in a noisy Arimoto's $\alpha$-mutual information for any quantum channel $\mathcal{N}(\cdot)$, which we write as 
	$I_\alpha (X;G)_{\mathcal{E}, \mathbb{M}, \mathcal{N}}$, where the conditional PMF is now given by $p(g|x) = \tr[ M_g \mathcal{N} (\rho_x) ]$. In particular, we are going to be interested in the quantity
	\begin{align}
	    I_\alpha (X;G)_{\mathcal{E}, \mathcal{N}}
	    \coloneqq
	    \max_{\mathbb{M}}
	    I_\alpha (X;G)_{\mathcal{E}, \mathbb{M}, \mathcal{N}}
	    .
	\end{align}
\end{remark}
We now consider the R\'enyi capacity in this quantum setting.
\begin{remark} (R\'enyi capacity of a quantum conditional PMF)
	The R\'enyi capacity of order $\alpha \in \mathds{\overline R}$ of a quantum conditional PMF $p_{G|X}^{(\mathbb{M},\mathcal{S})}$ is given by:
	\begin{align}
        C_{\alpha}
    		\left(
    		p_{G|X}^{(\mathbb{M},\mathcal{S})}
    		\right)
        \coloneqq
          \max_{p_X}
    		I_{\alpha }
    		\left(
    		p_{G|X}^{(\mathbb{M},\mathcal{S})}
    		p_X
    		\right).
        \label{eq:iso}
    \end{align}
\end{remark}
The quantity we are interested in the quantum domain is the R\'enyi capacity of order $\alpha$ of a quantum-classical channel.

\begin{definition}(R\'enyi capacity of a quantum-classical channel) The R\'enyi capacity of order $\alpha \in \mathbb{\overline R}$ of a quantum-classical channel $\Lambda_\mathbb{M}$ associated to the measurement $\mathbb{M}$ is given by:
\begin{align}
    C_{\alpha}
	(\Lambda_\mathbb{M})
    \coloneqq
        \max_{\mathcal{S}}
        C_{\alpha}
        \left(
        p_{G|X}^{(\mathbb{M},\mathcal{S})}
        \right)
        =
        \max_\mathcal{E}
        	I_{\alpha}
        	\left(
        	p_{G|X}^{(\mathbb{M},\mathcal{S})}
        	p_X
        	\right)
        ,
\end{align}
with the maximisation over all sets of states $\mathcal{S}=\{\rho_x\}$ or over all ensembles $\mathcal{E}=\{\rho_x, p(x)\}$.
\end{definition}
We now address a resource-theoretic approach for measurement informativeness and non-constant channels.

\subsection{The quantum resource theories of measurement informativeness and non-constant channels}
\label{s:RoMI}

The framework of quantum resource theories (QRTs) has proven a fruitful approach towards quantum theory \cite{QRT_QP, RT_review}. In this work we particularly deal with convex QRTs of measurements, channels. We start with the QRT of measurement informativeness \cite{SL}.
\begin{definition} (QRT of measurement informativeness \cite{SL}) Consider the set of Positive-Operator Valued Measures (POVMs) acting on a Hilbert space of dimension $d$. A POVM $\mathbb{M}$ is a collection of POVM elements  $\mathbb{M}=\{M_a\}$ with $a\in \{1,...,o\}$ satisfying $M_a\geq 0$ $\forall a$ and $\sum_a M_a=\mathds{1}$.  We now consider the resource of informativeness \cite{SL}. We say a measurement $\mathbb{N}$ is uninformative when there exists a PMF $q_A$ such that $N_a=q(a)\mathds{1}$, $\forall a$. We say that the measurement is informative otherwise, and denote the set of all uninformative measurements as ${\rm UI}$. 
\end{definition}
The set of uninformative measurements forms a convex set and therefore, defines a convex QRT of measurements. We now introduce the notion of simulability of measurements, which is also called classical post-processing (CPP).
\begin{definition}(Simulability of measurements \cite{simulability, SL})
A measurement $\mathbb{N}=\{N_x\}$, $x\in \{1,...,k\}$ is simulable by the measurement $\mathbb{M}=\{M_a\}$, $a\in \{1,...,o\}$ when there exists a conditional PMF $q_{X|A}$ such that:
$N_x=\sum_a q(x|a)M_a$, $\forall x$. The simulability of measurements defines a partial order for the set of measurements which we denote as $\mathbb{N} \preceq \mathbb{M}$, meaning that $\mathbb{N}$ is simulable by $\mathbb{M}$. Simulability of the measurement $\mathbb{N}$ can alternatively be understood as a classical post-processing of the measurement $\mathbb{M}$.
\end{definition} 
Two quantifiers for informativeness are the following.
\begin{definition} (Generalised robustness and weight of informativeness)
The generalised robustness \cite{GRoE, SL} and the weight \cite{EPR2, DS} of informativeness of a measurement $\mathbb{M}$ are given by:
\begin{align}
{\rm R}\left(\mathbb{M}\right)
&\coloneqq
{\scriptsize
	\begin{matrix}
	\text{\small \rm min}\\
	r \geq 0\\
	\mathbb{N} \in {\rm UI} \\
	\mathbb{M}^G \\
	\end{matrix}
}
\left\{ 
\rule{0cm}{0.6cm} r\,\bigg| \, M_a+rM^G_a=(1+r)N_a
\right\},
\label{eq:RoI}\\
{\rm W}\left(\mathbb{M}\right)
&\coloneqq
{\scriptsize
	\begin{matrix}
	\text{\small \rm min}\\
	w \geq 0\\
	\mathbb{N} \in {\rm UI} \\
	\mathbb{M}^G \\
	\end{matrix}
}
\left\{ 
\rule{0cm}{0.6cm} w\,\bigg| \, M_a=wM^G_a+(1-w)N_a
\right\}.
\label{eq:WoI}
\end{align}
The generalised robustness quantifies the minimum amount of a general measurement $\mathbb{M}^G$ that has to be added to $\mathbb{M}$ such that we get an uninformative measurement $\mathbb{N}$. The weight on the other hand, quantifies the minimum amount of a general measurement $\mathbb{M}^G$ that has to be used for recovering the measurement $\mathbb{M}$.
\end{definition}
These resource quantifiers are going to be useful later on. We now introduce the QRT of non-constant channels.

\begin{definition}
    (QRT of non-constant channels) Consider the set of completely-positive trace-preserving (CPTP) maps acting on a Hilbert space of dimension $d$. We now consider the resource of non-constant channels. We say that a channel $\mathcal{N}(\cdot)$ is constant, when there exist a state $\rho_\mathcal{N}$ such that $\mathcal{N}(\rho) = \rho_\mathcal{N}$, $\forall \rho \in D(\mathds{H})$. We say that a channel is non-constant otherwise, and denote the set of all constant channels as $\mathcal{C}$.
\end{definition}
We now consider information-theoretic quantities for various general QRTs.

\subsection{Arimoto-type information-theoretic quantities for general QRTs of measurements, channels, states, and state-measurement pairs}
\label{s:gap}

We now address a generalisation of Arimoto's $\alpha$-mutual information to the concept of Arimoto's gap for general resources of measurements, channels, states, and state-measurement pairs. In order to introduce the concept of Arimoto's gap, let us first fix some notation.  In this subsection we consider general QRTs with arbitrary resources, meaning that we address a set of free measurements as $\mathbb{F}$, and a set of free channels as $\mathcal{F}$, which are usually assumed to be convex and closed sets \cite{TR1, TR2, DS}. We now introduce the concept of \emph{Arimoto's gap}, which is defined in terms of the standard Arimoto's $\alpha$-mutual information, and for which we introduce here two variants as follows.

\begin{definition} (Arimoto's gap for measurements and channels \cite{TR2,MO})
Consider a set of free measurements as $\mathbb{F}$, and a pair $(\mathcal{E}, \mathbb{M})$, \emph{Arimoto's gap on POVMs} of order $\alpha \in \mathds{\overline{R}}$ for such a pair is given by:
\begin{align}
    G_{\alpha}^{\mathbb{F}}
	(X;G)_{\mathcal{E},\mathbb{M}}
	\coloneqq
	I_{\alpha}
	(X;G)_{\mathcal{E},\mathbb{M}}
	-
	\max_{\mathbb{N}\in \mathds{F}}
	I_{\alpha}
	(X;G)_{\mathcal{E},\mathbb{N}}
	.
\end{align}
Similarly, consider a set of free channels $\mathcal{F}$ and a triple $(\mathcal{E}, \mathbb{M}, \mathcal{N})$, \emph{Arimoto's gap on channels} of order $\alpha \in \mathds{\overline{R}}$ for such a triple is given by:
\begin{multline}
    G_{\alpha}^{\mathcal{F}}
	(X;G)_{\mathcal{E},\mathbb{M},\mathcal{N}}
	\\ 	
	\coloneqq
	I_{\alpha}
	(X;G)_{\mathcal{E},\mathbb{M},\mathcal{N}}
	-
	\max_{\mathcal{\widetilde{N}} \in \mathcal{F}}
	\max_{\mathbb{N}}
	I_{\alpha}
	(X;G)_{\mathcal{E}, \mathbb{N}, \mathcal{\widetilde{N}}}	.
\end{multline}
Similarly to the previous section, we also address a more refined quantity as:
\begin{align}
    G_{\alpha}^{\mathcal{F}}
	(X;G)_{\mathcal{E},\mathcal{N}}
	\coloneqq
	\max_{\mathbb{M}}
	G_{\alpha}^{\mathcal{F}}
	(X;G)_{\mathcal{E},\mathbb{M},\mathcal{N}}
	.
\end{align}
\end{definition}

These quantities are information-theoretic in nature, being defined in terms of Arimoto's $\alpha$-mutual information. We can think about them as the maximum \emph{gap}, in terms of the Arimoto's $\alpha$-mutual information, between the free set $\mathbb{F}$ $(\mathcal{F})$ and the fixed object of interest $\mathbb{M}$ $(\mathcal{N})$. These two measures can be thought of as generalisations of Arimoto's noisy $\alpha$-mutual information and Arimoto's $\alpha$-mutual information, respectively. This can be checked by setting $(\mathbb{F} = \mathbb{UI})$ and $(\mathcal{F} = \mathcal{C})$, for which we get:
\begin{align}
    G_{\alpha}^{\mathbb{UI}}
	(X;G)_{\mathcal{E},\mathbb{M}}
	&=
	I_{\alpha}
	(X;G)_{\mathcal{E},\mathbb{M}}
	,
    \\
    G_{\alpha}^{\mathcal{C}}
	(X;G)_{\mathcal{E},\mathbb{M},\mathcal{N}}
	&=
	I_{\alpha}
	(X;G)_{\mathcal{E},\mathbb{M},\mathcal{N}}
	.
\end{align}
This is because uninformative measurements achieve $p(g|x) = \tr[M_g\rho_x] = p(g)\tr[\rho_x] = p(g)$, and similarly for constant channels $p(g|x) = \tr[M_g\mathcal{\widetilde{N}}(\rho)] = \tr[M_g\rho_\mathcal{\widetilde{N}}] = p(g)$, meaning that random variables $G$ and $X$ are independent from each other in both cases and therefore
\begin{align}
    \max_{\mathbb{N}\in \mathds{UI}}
	I_{\alpha}
	(X;G)_{\mathcal{E},\mathbb{N}}
    =
    \max_{\mathcal{\widetilde{N}} \in \mathcal{C}}
	\max_{\mathbb{N}}
	I_{\alpha}
	(X;G)_{\mathcal{E},\mathbb{N},	\mathcal{\widetilde{N}}}
	=
	0
	.
\end{align}

Inspired by these information-theoretic quantities for measurements and channels, we now also consider Arimoto-type gaps for states as well as for a hybrid scenario with state-measurements pairs. Similarly for the case of measurements and channels, we address a set of free states as $\rm F$, which is usually assumed to be convex and closed \cite{TR1, TR2}. We now define two variants of the concept of Arimoto's gap for QRTs of states as well as for QRTs of state-measurement pairs.

\begin{definition}(Arimoto's gap for states and for state-measurement pairs)
Consider a set of free states ${\rm F}$, and a triple $(\Lambda, \mathbb{M}, \rho)$, then, \emph{Arimoto's gap on states} of order $\alpha \in \mathds{\overline{R}}$ for such a triple
is given by:
\begin{align}
    G_{\alpha}^{{\rm F}}
	(X;G)_{\Lambda,\mathbb{M},\rho}
	\coloneqq
	I_{\alpha}
	(X;G)_{\Lambda,\mathbb{M},\rho}
	-
	\max_{\sigma \in {\rm F}}
	I_{\alpha}
	(X;G)_{\Lambda,\mathbb{M},\sigma}
	.
\end{align}
Similarly, consider a set of free states ${\rm F}$, a set of free measurements $\mathbb{F}$, and a triple $(\Lambda, \mathbb{M}, \rho)$, then, \emph{Arimoto's gap on state-measurement pairs} of order $\alpha \in \mathds{\overline{R}}$ for such a triple is given by:
\begin{multline}
    G_{\alpha}^{{\rm F},\mathbb{F}}
	(X;G)_{\Lambda,\mathbb{M},\rho}
		\\	
	\coloneqq
	I_{\alpha}
	(X;G)_{\Lambda,\mathbb{M},\rho}
	-
	\displaystyle
	\max_{\substack{
	\sigma \in {\rm F}
	\\
	\mathbb{N}\in \mathds{F}
	}
	}
	I_{\alpha}
	(X;G)_{\Lambda,\mathbb{N},\sigma}
	.\end{multline}
\end{definition}

Similarly to the previous variants on Arimoto's gaps, we have that these information-theoretic measures can be understood as quantifying the maximum \emph{gap}, in terms of the standard Arimoto's $\alpha$-mutual information, between the set of free objects and a fixed triple $(\Lambda, \mathbb{M}, \rho)$. The first variant was first introduced in \cite{TR2} whilst the second multi-object variant was first introduced in \cite{MO}.

Here we finish with the preliminary concepts and theoretical tools needed to describe our main results which we do next.

\section{Quantum betting tasks with risk aversion}
\label{s:QBT}

We now introduce the main new operational tasks that we consider in this work. We start by describing quantum betting tasks being played by gamblers with different risk tendencies. This is inspired by both standard quantum state discrimination and horse betting games in classical information theory.

Horse betting (HB) games were first introduced by Kelly in 1956 \cite{kelly}, a modern introduction can be found, for instance, in Cover \& Thomas \cite{CT}, as well as in the lectures notes by Moser \cite{LN_moser}. Recently, Bleuler, Lapidoth, and Pfister generalised HB games in order to include a factor $\beta = 1-R$ \cite{BLP1}, representing the risk-aversion of the Gambler (Bob) playing these games, with standard HB games being recovered by setting $\beta=0$, corresponding to $R = 1$, i.e. a risk-averse Bob.

Inspired by this, here we introduce three types of quantum betting tasks. First, we introduce quantum state betting (QSB) games. Specifically, we will introduce two variants of QSB games in the form of quantum state discrimination (QSD) with risk, and quantum state exclusion (QSE) with risk.  We will then introduce the central figure of merit for QSB games -- the isoelastic certainty equivalent (ICE), and show how it generalises the quantification of standard quantum state discrimination and exclusion. We then introduce important variants of this first game. In particular, we introduce noisy quantum state betting (nQSB) games and quantum channel betting (QCB) games, which generalises both quantum channel discrimination and exclusion.  The tasks considered in this section, and the way they relate to each other is depicted in \autoref{fig:tasks}.

\begin{figure}[h!]
    \centering
    \includegraphics[scale=0.67]{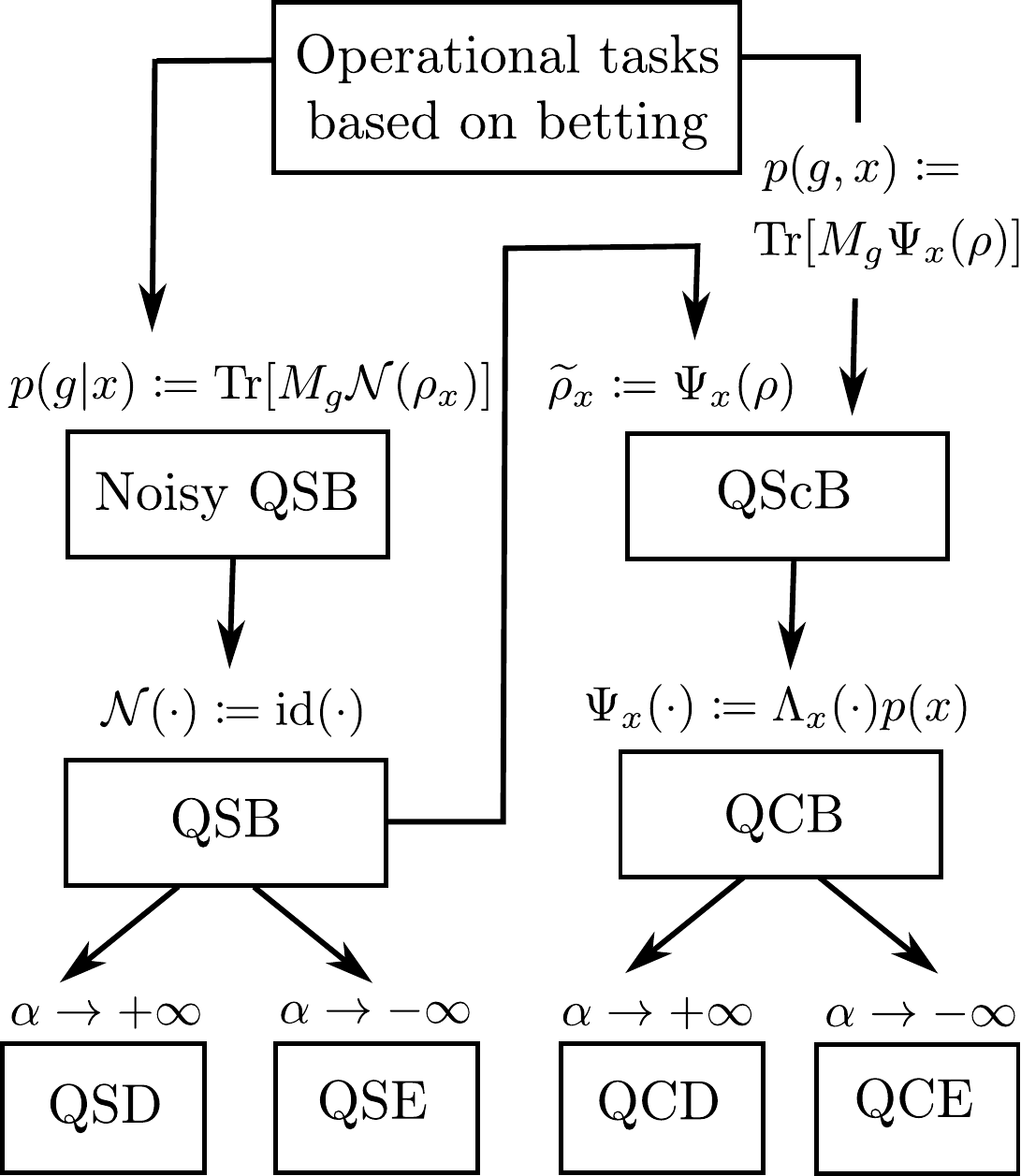}
    \caption{Operational tasks based on betting and risk-aversion. Quantum state betting (QSB), quantum subchannel betting (QScB), quantum channel betting (QCB), quantum state discrimination/exclusion (QSD/QSE), quantum channel discrimination/exclusion (QCD/QCE). $A \rightarrow B$ means that the task $A$ is more general than $B$.} 
    \label{fig:tasks}
\end{figure}

\subsection{Quantum state betting games}
\label{ss:QSB}

Consider two rational agents, a Referee (Alice) and a Gambler (Bob). Alice is in possession of an ensemble of quantum states $\mathcal{E} = \{\rho_x, p(x)\}$, $x\in\{1,...,K\}$, and is going to send one of these states to Bob, say $\rho_x$. We address here a \emph{quantum state}, or \emph{state} for short, as a positive semidefinite ($\rho_x \geq 0$) and trace one ($\tr(\rho_x)=1$) operator in an finite-dimensional Hilbert space.

As above, we will consider two different classes of state betting games, gain games, and loss games. In a gain game, Alice offers Bob odds $o(x)$, which is a positive function ($o(x)>0$, $\forall x$) but not necessarily a PMF, such that if Bob places a unit bet on the state being $\rho_x$, and this is the correct state, then Alice will pay out $o(x)$ to Bob. In a loss game, on the contrary, we take the `odds' to be \emph{negative}, $o(x) < 0$, for all $x$, such that if Bob places a unit bet on $\rho_x$, then he will have to pay out to Alice an amount $|o(x)|$.\footnote{That is, similarly to in thermodynamics, we take the sign of the odds to signify whether this is a gain or a loss for Bob.}

In order to decide how to place his bets, Bob is allowed to first perform a \emph{quantum measurement} on the state given to him by Alice. In general, this will be a positive operator-valued measure (POVM), $\mathbb{M}=\{M_g\}$, $M_g\geq 0$ $\forall g$, $\sum_g M_g=\mathds{1}$,  which will allow him to (hopefully) extract some useful information from the state. 

Let us assume that Bob measures the state he receives from Alice using a measurement $\mathbb{M}=\{M_g\}$, producing a measurement result $g$, with probability given by the Born rule, $p(g|x)=\tr[M_g\rho_x]$. Bob will then use this result to decide on his \emph{betting strategy}. We assume that he bets all of his wealth, and divides this in some way amongst all the possible options $x\in \{1,...K\}$. That is, Bob's strategy is a PMF $b_{X|G}$, such that Bob bets the proportion $b(x|g)$ of his wealth on state $x$ being the sent state, when his measurement outcome was $g$.\footnote{Note that for loss games, Bob can end up having to pay out \emph{more} than the wealth he bet (similarly to how in a game gain Bob can walk away with more wealth than he started with).} We note that Bob's overall strategy is then defined by the pair $(b_{X|G}, \mathbb{M})$. We also note that the PMF $p_X$ from the ensemble of states together with the conditional PMF $p_{G|X}$ from the measurement implemented by Bob, defines the joint PMF $p_{XG} \coloneqq p_{G|X}p_X$.

Therefore, when the quantum state was $\rho_x$, and Bob obtained the measurement outcome $g$, he bet the proportion of his wealth $b(x|g)$ on the actual state, and hence Alice either pays out $w(x,g) = o(x)b(x|g)$ in the case of a gain game, or Bob has to pay Alice the amount $|w(x,g)|$ (i.e. he loses $|w(x,g)|$) in a loss game. We can view gain games as a generalisation of \emph{state discrimination}. Here, since Bob is winning money, it is advantageous, in general, for him to correctly \emph{identify} the state that was sent. On the other hand, we see that loss games can be viewed as a generalisation of \emph{state exclusion}, since now in order to minimise his losses, it is useful for Bob to be able to \emph{avoid} or \emph{exclude} the state that was sent. 

Finally, we note that the settings of the game are specified by the pair $(o_X,\mathcal{E})$. It is important to stress that by assumption Bob is fully aware  of the settings of the game, meaning that the pair $(o_X,\mathcal{E})$ is known to him prior to playing the game, and therefore he can use this knowledge in order to select an optimal betting strategy $b_{X|G}$.
 
\subsection{Figure of merit for quantum state betting games}
\label{ss:QSB1}

Given these two variants of QSB games, we now want analyse the behaviour of different types of Gamblers (represented by different utility functions), according to their risk tendencies. We will consider quantities of interest like in the previous sections such as: expected wealth, expected utility, and similar. In particular, we model Gamblers with utility functions displaying constant relative risk aversion (CRRA) and therefore, the utility functions we consider are isoelastic functions $u_R(w)$ \eqref{eq:isoelastic}. The figure of merit we are interested in is then the \emph{isoelastic certainty equivalent (ICE)} $w^{ICE}_R$ with $R \in \mathds{\overline R}$. For risk $R\in (-\infty,1) \cup (1,\infty)$, this quantity is given by:
\begin{align}
		w^{ICE}_R&
		(b_{X|G}, \mathbb{M}, o_X,\mathcal{E})
		\nonumber
		\\
		&=
    	u_R^{-1}
    	\left(
    	\mathbb{E}_{p_{XG}}
    	\left[
    	u_R
    		(
    		w_{XG}
    		)
    	\right]
    	\right)
    	,
    	\nonumber
    	\\
		&=
		\left[
		\sum_{g,x}
		\big[
		b(x|g)
		o(x)\big]^{1-R}
		p(g|x)
		p(x)
		\right]^\frac{1}{1-R}
		.
		\label{eq:wCE_QSB}
\end{align}	
The cases $R \in\{1,\infty,-\infty\}$ are defined by continuous extension of \eqref{eq:wCE_QSB}. 
In summary, the game is specified by the pair $(o_{X},\mathcal{E})$, the behavioural tendency of Bob is represented by the utility function $u_R(w_{XG})$ with a fixed $R \in \mathds{\overline R}$, the overall strategy of Bob is specified by the pair $(b_{X|G}, \mathbb{M})$, and the figure of merit here considered is the isoelastic certainty equivalent (ICE) \eqref{eq:wCEF}. We can alternatively address these operational tasks as horse betting games with risk and quantum side information, or quantum horse betting (QHB) games for short, and we describe this in more detail later on.

Bob is in charge of the measurement and the betting strategy ($b_{X|G}, \mathds{M}$), so in particular, for a fixed measurement $\mathds{M}$, Bob is interested in maximising the ICE (maximising gains in a gain game, and minimising losses in a loss game) so we are going to be interested in the following quantity:
\begin{align*}
    \max_{b_{X|G}}
	\,
	w^{ICE}_R
	\left(
	b_{X|G}, \mathbb{M}, o_X,\mathcal{E}
	\right)
    ,
\end{align*}
for a fixed QSB game $(o_X,\mathcal{E})$ with either positive or negative odds, and Bob's risk tendencies being fixed, and specified by an isoelastic utility function $u_R$.

\subsection{Quantum state betting games generalise discrimination and exclusion games}
\label{ss:QSB2}

We will now show that quantum state betting games with risk can indeed be seen as generalisations of standard quantum state discrimination and exclusion games. We can see this by considering a risk-neutral ($R=0$) Bob playing a gain game (positive odds) which are constant: $o^{c}(x) \coloneqq C$, $C>0$, $\forall x$, in which case we find that the quantity of interest becomes:
\begin{align}
	\max_{
			b_{X|G}
		}
		\,
	w^{ICE}_{0}
	(b_{X|G}, \mathbb{M}, o^{c}_X,\mathcal{E})
	&
	=
	C
	\max_{
			b_{X|G}
		}
		\,
	\sum_{g,x}
	b(x|g)
	p(g|x)
	p(x)
	,\nonumber 
	\\
	&=
	C \,
	P^{\rm QSD}_{\rm succ}(\mathcal{E},\mathbb{M})
	.
\end{align}	
For more details on standard quantum state discrimination games we refer to \cite{SL, TR2}. Therefore, standard quantum state discrimination can be seen as as special instance of quantum state betting games with constant odds, and played by a risk-neutral player. Similarly, for a loss game, with negative constant odds $o^{-c}(x) \coloneqq -C$, $C>0$, $\forall x$:
\begin{align}
	\max_{
			b_{X|G}
		}
		\,
	w^{ICE}_{0}
	(b_{X|G}, \mathbb{M}, &o^{-c}_X,\mathcal{E})
	\nonumber
	\\
	&
	=
	C
	\max_{
			b_{X|G}
		}
		\,
	-
	\sum_{g,x}
	b(x|g)
	p(g|x)
	p(x)
	,\nonumber
	\\
	&=
	-
	C \,
	P^{\rm QSE}_{\rm err}(\mathcal{E},\mathbb{M})
	.
\end{align}	
For more details on standard quantum state exclusion games we refer to \cite{DS, uola2020}. Therefore, standard quantum state exclusion can be seen as a quantum state betting game constant negative odds, again played by a risk-neutral player. 

\subsection{Noisy quantum state betting games}
\label{ss:nQSB}

We now introduce noisy quantum state betting (nQSB) games. We first note that standard QSB games (from the previous section) are implicitly assuming that the states that Alice (referee) sends to Bob (player) are perfectly transmitted, meaning that they are not affected by undesired interactions due to the environment. This is an idealised situation, and a more realistic scenario including such effects can be addressed by considering a completely-positive trace-preserving (CPTP) map (or quantum channel) $\mathcal{N}$, so that the probability of obtaining an outcome $g$ after receiving the state $\rho_x$ is now given by $p(g|x)=\tr[M_g\, \mathcal{N}(\rho_x)]$. We refer to this more general and realistic scenario as \emph{noisy} QSB (nQSB) games.

\begin{definition}(Noisy quantum state betting games)
The isoelastic certainty equivalent (ICE) for a \emph{noisy quantum state betting} (nQSB) game is given by:
\begin{multline}
		w^{\rm nQSB}_R
		(b_{X|G}, \mathbb{M}, o_X,  \mathcal{E}, \mathcal{N})
		\\
		\coloneqq
		w^{ICE}_R
		(b_{X|G}, \mathcal{N}^\dagger(\mathbb{M}), o_X, \mathcal{E})
		\label{eq:wICE_nQSB}
\end{multline}	
with $p(g|x) = \tr[\mathcal{N}^\dagger(M_g)\, \rho_x] = \tr[M_g\, \mathcal{N}(\rho_x)]$, $\mathcal{N}(\cdot)$ a completely-positive trace-preserving (CPTP) map, $\mathbb{M}=\{M_g\}$ a POVM, and the POVM $\mathcal{N}^\dagger(\mathbb{M})  \coloneqq \{\mathcal{N}^\dagger(M_g)\}$. The cases $R \in\{1,\infty,-\infty\}$ are defined by continuous extension of \eqref{eq:wICE_nQSB}.
\end{definition}

We note that we recover standard QSB games by considering a noiseless scenario $\mathcal{N}(\cdot)={\rm id}(\cdot)$. Whilst noisy QSB games can be seen as noiseless QSB games by considering the POVM $\mathcal{N}^\dagger(\mathbb{M})  \coloneqq \{\mathcal{N}^\dagger(M_g)\}$, it is still important from a physical point to view to make the distinction between both noisy and noiseless scenarios. Later on we will see how this is relevant for the resource theory of non-constant channels.

\subsection{Quantum channel betting (QCB) games}
\label{ss:QCB}

In this subsection we introduce quantum channel betting (QCB) games. Taking inspiration from the previous QSB games, where Bob (player) is asked to bet on an ensemble of states, we now consider Bob being asked to bet instead on a set of channels $\Lambda = \{\Lambda_x\}$, distributed according to a PMF $p_X$. In this scenario, Bob is in possession of a quantum state $\rho$, which he would consequently send to Alice (referee). Alice then proceeds to generate the ensemble $\{\Lambda_x(\rho), p(x)\}$, and send back one of these states to Bob. Bob then proceeds to measure the received state with a fixed POVM $\mathbb{M} = \{M_g\}$, and use the extracted information $g$ in order to place a bet $b_{X|G}$ and effectively play the game. Following a similar logic to the case for QSB games, we can formalise and derive a figure of merit for QCB games in terms of the isoelastic certainty equivalent as follows.
\begin{definition}(Quantum channel betting)
The isoelastic certainty equivalent (ICE) for a \emph{quantum channel betting (QCB)} game is given by:
\begin{multline}
		w^{\rm QCB}_R
		(b_{X|G}, o_X, p_X, \Lambda, \rho,\mathbb{M})
		\\
		\coloneqq
		w^{ICE}_R
		(b_{X|G}, \mathbb{M}, o_X, \mathcal{E}_{\Lambda, \rho})
		\label{eq:wICE_QCB}
\end{multline}	
with $p(g|x)=\tr[M_g \Lambda_x(\rho)]$, $\Lambda = \{\Lambda_x(\cdot)\}$ a set of completely-positive trace-preserving (CPTP) maps, $\mathbb{M}=\{M_g\}$ a POVM, and $\mathcal{E}_{\Lambda,\rho}\coloneqq \{\Lambda_x(\rho), p(x)\}$. The cases $R \in\{1,\infty,-\infty\}$ are defined by continuous extension of \eqref{eq:wICE_QCB}. 
\end{definition}

First, we note here that these tasks can be further extended to \emph{quantum subchannel betting (QScB)} games where we address  a set of subchannels $\Psi=\{\Psi_x(\cdot)\}$, or set of completely-positive trace-nonincresing (CPTNI) maps, with $p(x,g)=\tr[M_g \Psi_x(\rho)]$. Second, whilst QCB can be seen as noiseless QSB games with the ensemble given by $\mathcal{E}_{\Lambda,\rho}\coloneqq \{\Lambda_x(\rho), p(x)\}$, it is still important to distinguish these two cases from a physical point of view, this, because in a QCB game Bob (player) is now allowed to have an influence on the ensemble of states as $\mathcal{E} =  \mathcal{E}_{\Lambda,\rho}$. Third, we can see that QCB games generalise standard channel discrimination and standard channel exclusion as follows. Consider a risk-neutral ($R=0$) Bob playing a gain game (positive odds) which are constant: $o^{c}(x) \coloneqq C$, $C>0$, $\forall x$, in which case we find that the ICE becomes: 
\begin{align}
	\max_{
			b_{X|G}
		}&
		\,
	w^{\rm QCB}_{0}
	(b_{X|G}, o_X^c, p_X, \Lambda,\rho,\mathbb{M})
	\nonumber
	\\
	&
	=
	C
	\max_{
			b_{X|G}
		}
		\,
	\sum_{g,x}
	b(x|g)\,
	\tr[M_g \Lambda_x(\rho)]
	\,
	p(x)
	,\nonumber 
	\\
	&=
	C \,
	P^{\rm QCD}_{\rm succ}
	(\Lambda,\rho,\mathbb{M})
	,
\end{align}	
with $\Lambda=\{\Lambda_x(\cdot)\}$ a set CPTP maps, $\mathbb{M}=\{M_g\}$ a POVM. Therefore, standard quantum channel discrimination can be seen as as special instance of quantum subchannel betting games with constant odds, and played by a risk-neutral player. For more details on standard quantum channel discrimination (QCD) games we refer the reader to \cite{TR1, TR2}. Similarly, for a loss game, with negative constant odds $o^{-c}(x) \coloneqq -C$, $C>0$, $\forall x$: 
\begin{align}
	\max_{
			b_{X|G}
		}&
		\,
	w^{\rm QCB}_{0}
	(b_{X|G}, o_X^{-c}, p_X, \Lambda, \rho,\mathbb{M})
	\nonumber
	\\
	&
	=
	C
	\max_{
			b_{X|G}
		}
		\,
	-
	\sum_{g,x}
	b(x|g)\,
	\tr[M_g \Lambda_x(\rho)]
	p(x)
	,\nonumber
	\\
	&=
	-
	C \,
	P^{\rm QCE}_{\rm err}
	(\Lambda,\rho,\mathbb{M})
	,
\end{align}	
with $\Lambda=\{\Lambda_x(\cdot)\}$ a set of CPTP maps, $\mathbb{M}=\{M_g\}$ a POVM. Therefore, standard quantum channel exclusion can be seen as a quantum channel betting game with constant negative odds, again played by a risk-neutral gambler. For more details on standard quantum channel exclusion (QCE) games we refer the reader to \cite{DS, uola2020}. We now proceed to address our main results.

\section{Main Results}
\label{s:results}

We are now ready to present the main results of this work.

\subsection{Arimoto's $\alpha$-mutual information and quantum state betting games} 
\label{ss:result1}

The main motivation now is to compare the performance of two gamblers via the maximised isoelastic certainty equivalent (ICE)
$\max_{
b_{X|G}
}
w^{ICE}_R
\left(
b_{X|G}, \mathbb{M}, o_X,\mathcal{E}
\right)$. Specifically, we want to compare: i) a general gambler using a fixed measurement $\mathbb{M}$ with ii) the best uninformative gambler, meaning a gambler who can implement any uninformative measurement $\mathbb{N}\in{\rm UI}$, or equivalently, a gambler described by the quantity
$
\max_{\mathds{N} \in {\rm UI}}
\max_{
b_{X|G}
}
w^{ICE}_R
\left(
b_{X|G}, \mathbb{N}, o_X,\mathcal{E}
\right)
$.
We have the following main result.
\begin{result} \label{R_uninformative}
	Consider the a QSB game defined by the pair $(o^{\sgn(\alpha)c}_X, \mathcal{E})$ with constant odds as $o^{\sgn(\alpha)c}(x) \coloneqq \sgn(\alpha)C$, $C>0$, $\forall x$, and an ensemble of states $\mathcal{E}=\{\rho_x,p(x)\}$. Consider a Gambler playing this game using a fixed measurement $\mathbb{M}$ in comparison to a Gambler being allowed to implement any uninformative measurement $\mathbb{N}\in{\rm UI}$. Consider both Gamblers with the same attitude to risk, meaning that they are represented by isoelastic functions $u_R(W)$ with the risk parametrised as $R(\alpha) \coloneqq 1/\alpha$. Each Gambler is allowed to play the game with the optimal betting strategies, meaning they can each propose a betting strategy independently from each other. Remembering that the Gamblers are interested in maximising the isoelastic certainty equivalent (ICE), we have the following relationship:
	\begin{align}
    	&I_{\alpha}
    	(X;G)_{\mathcal{E},\mathbb{M}}
    	\label{eq:result1}
        \\
        \nonumber
        &=
            \sgn(\alpha)
            \log
        	\left[
        	\frac{
        		\displaystyle
        		\max_{
        			b_{X|G}
        		}
        		\,
        		w^{ICE}_{1/\alpha}
        		\left(
        		b_{X|G}
        		,
        		\mathbb{M}
        		,
        		o_X^{\sgn(\alpha)c}
        		,
        		\mathcal{E}
        		\right)
        	}{
        		\displaystyle
        		\max_{\mathbb{N}\in {\rm UI}}
        		\max_{
        			b_{X|G}
        		}
        		\,
        		w^{ICE}_{1/\alpha}
        		\left(
        		b_{X|G}
        		,
        		\mathbb{N}
        		,
        		o_X^{\sgn(\alpha)c}
        		,
        		\mathcal{E}
        		\right)
        	}
        	\right]
        	.
    \end{align}
	This shows that Arimoto's $\alpha$-mutual information quantifies the ratio of the isoelastic certainty equivalent with risk $R(\alpha) \coloneqq 1/\alpha$ of the game defined by $(o^{\sgn(\alpha)c}_X, \mathcal{E})$, when the QSB game is played with the best betting strategy, and when we compare a Gambler implementing a fixed measurement $\mathbb{M}$ against a Gambler using any uninformative measurement $\mathbb{N}\in {\rm UI}$.
\end{result}

The full proof of \cref{R_uninformative} is in \cref{AR_uninformative}. We now analyse two cases of particular interest ($\alpha \in\{\infty, -\infty\}$), as the following corollaries.
\begin{corollary}
	In the case $\alpha\rightarrow\infty$ we recover the result found in \cite{SL}. Explicitly, we have:
	\begin{align}
		C_{\infty}(\Lambda_\mathbb{M})
		&= \max_\mathcal{E} I_{\infty}		(X;G)_{\mathcal{E},\mathbb{M}}, \nonumber	\\ &
		=
		\log
		\left[
		\max_\mathcal{E}
		\frac{
			P^{\rm QSD}_{\rm succ}(\mathcal{E},\mathbb{M})	
		}
		{
			\max_{\mathbb{N}\in {\rm UI}}
			P^{\rm QSD}_{\rm succ}(\mathcal{E},\mathbb{N})
		}
		\right],
	\end{align}
	where  $P^{\rm QSD}_{\rm succ}(\mathcal{E},\mathbb{M})$ is the probability of success in the quantum state discrimination (QSD) game defined by $\mathcal{E}$, with the Gambler using the measurement $\mathbb{M}$, given explicitly by:
	\begin{align}
		P^{\rm QSD}_{\rm succ}(\mathcal{E},\mathbb{M})
		\coloneqq
		\max_{q_{G|A}}
		\sum_{g,a,x}
		\delta^g_x\,
		q(g|a)\,
		p(a|x)\,
		p(x),
	\end{align}
	with $p(a|x)\coloneqq \tr [M_a\rho_x]$, and the maximisation over all classical post-processing $q_{G|A}$. We remark that the R\'enyi capacity of order $\infty$ has also been called as the accessible min-information of a channel, and denoted as $I^{\rm acc}_{\infty}(\Lambda_\mathbb{M})$ \cite{SL, Wilde_book}. This shows that quantum state betting with risk (QSB$_{R(\alpha)}$) becomes equivalent to quantum state discrimination (QSD) when $\alpha \rightarrow  \infty$.
\end{corollary}
\begin{corollary}
	In the case $\alpha\rightarrow -\infty$ we recover the result found in \cite{DS}. Explicitly, we have:
	\begin{align}
		C_{-\infty}(\Lambda_\mathbb{M})
		&= \max_\mathcal{E} I_{-\infty}		(X;G)_{\mathcal{E},\mathbb{M}}, \nonumber	\\ &
		=
		-
		\log
		\left[
		\min_\mathcal{E}
		\frac{
			P^{\rm QSE}_{\rm err}(\mathcal{E},\mathbb{M})	
		}
		{
			\min_{\mathbb{N}\in {\rm UI}}
			P^{\rm QSE}_{\rm err}(\mathcal{E},\mathbb{N})
		}
		\right],
	\end{align}
	where $P^{\rm QSE}_{\rm err}(\mathcal{E},\mathbb{M})$ is the probability of error in the quantum state exclusion (QSE) game defined by $\mathcal{E}$, with the Gambler using the measurement $\mathbb{M}$ explicitly given by:
	\begin{align}
		P^{\rm QSE}_{\rm err}(\mathcal{E},\mathbb{M})
		\coloneqq
		\min_{q_{G|A}}
		\sum_{g,a,x}
		\delta^g_x\,
		q(g|a)\,
		p(a|x)\,
		p(x).
	\end{align}
	with $p(a|x)\coloneqq \tr [M_a\rho_x]$, and the minimisation being performed over all classical post-processing $q_{G|A}$. We remark that the R\'enyi capacity of order $-\infty$ has also been called the excludible information of a channel, and denoted as $I^{\rm exc}_{-\infty}(\Lambda_\mathbb{M})$ \cite{DS, MO}. This shows that quantum state betting with risk (QSB$_{R(\alpha)}$) becomes equivalent to quantum state exclusion (QSE) when $\alpha \rightarrow  -\infty$.
\end{corollary}
In \cref{Acorollaries} we provide further details on these two corollaries.

Result 1 establishes a connection between Arimoto's $\alpha$-mutual information and QSB games, which recovers two known cases at $\alpha \in \{\infty, -\infty\}$ \cite{SL, DS}. We emphasise that the right hand side of \eqref{eq:result1} is a completely \emph{operational} quantity, which represents the advantage that an informative measurement provides when being used as a resource for QSB games, whilst the left hand side is the raw \emph{information-theoretic} mutual information measure proposed by Arimoto and consequently, this result provides an operational interpretation of Arimoto's $\alpha$-mutual information in the quantum domain. 

Furthermore, it shows that the R\'enyi parameter can be interpreted as characterising the risk tendency of the Gamblers as $R = 1/\alpha$. It is also interesting to note that this works for all ensembles $\mathcal{E} = \{\rho_x,p(x)\}$, all measurements $\mathds{M} = \{M_g\}$, as well as for the whole range of the R\'enyi parameter $\alpha \in \mathds{\overline R}$, including negative values. We summarise the interpretation of this result in Fig.~\ref{fig:grid}.
\begin{figure}[h!]
    \centering
    \includegraphics[scale=0.47]{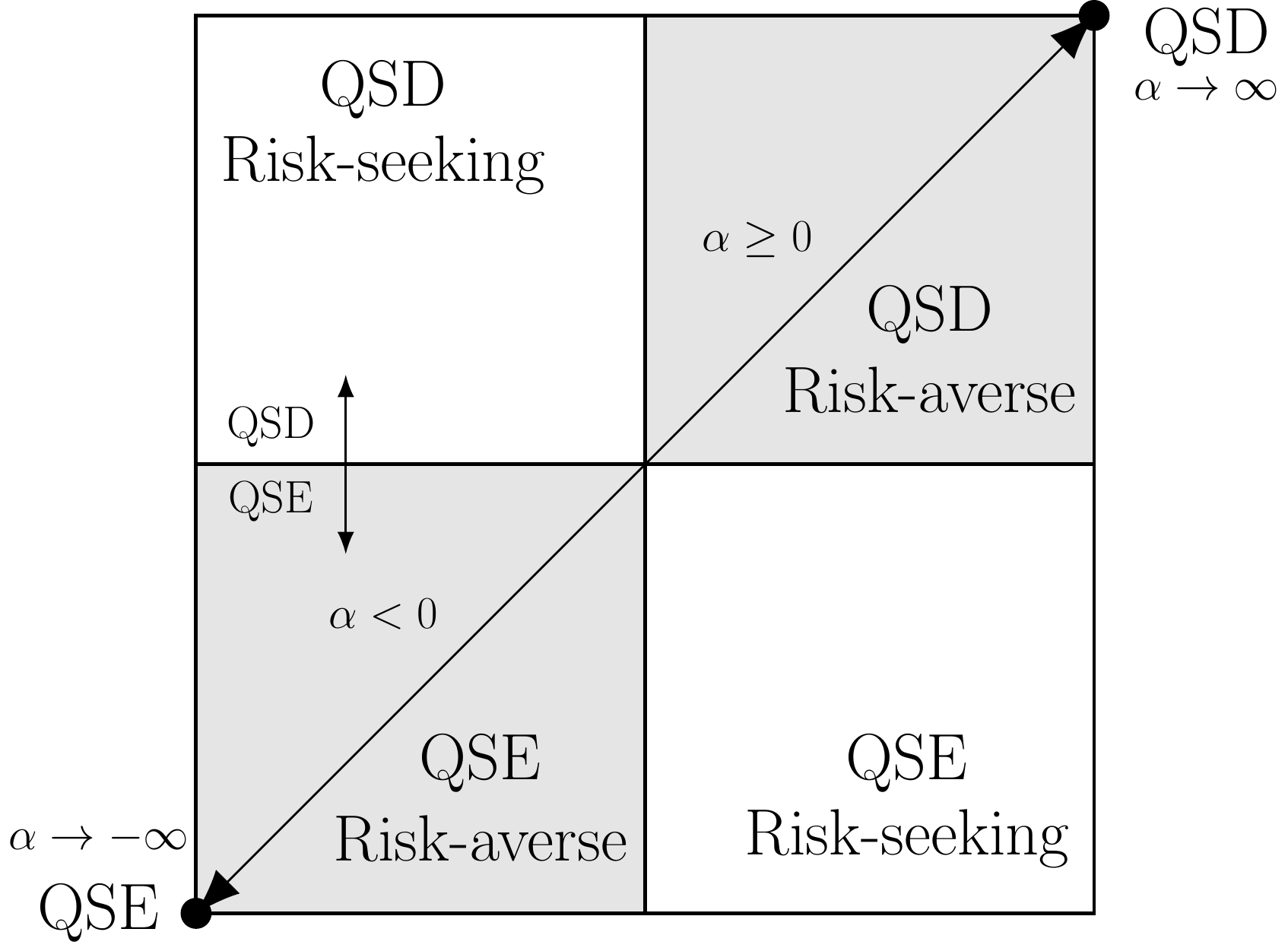}
    \vspace{-0.2cm}
    \caption{
    Possible scenarios for quantum state discrimination (QSD) and quantum state exclusion (QSE) games being played by Gamblers with different risk tendencies: risk-averse, risk-seeking, or risk-neutral, with the risk being parametrised as $R(\alpha)=1/\alpha$. Result 1 establishes that Arimoto's mutual information quantifies the shaded region for $\alpha \in \mathds{\overline R}$, meaning that it characterises risk-averse Gamblers playing either QSD ($\alpha \geq 0$) and QSE games ($\alpha<0$). The left bottom corner $(\alpha \rightarrow -\infty)$ and the top-right corner $(\alpha \rightarrow \infty)$ represent a risk-neutral Gambler $R=0$ playing either standard exclusion or discrimination games, respectively. This means that standard QSD games can be understood as a risk-neutral Gambler playing QSD games with risk. Similarly, standard QSE games can be understood as a risk-neutral Gambler playing QSE games with risk. The middle point at $\alpha \rightarrow 0$ represents the transition between a maximally risk-averse Gambler playing QSD games and a maximally risk-averse Gambler playing QSE games.
    }
    \label{fig:grid}
\end{figure}

We also highlight here that Result 1 lies at the intersection of three major fields: quantum theory, information theory, and the theory of games and economic behaviour. We believe that this result has the potential to spark further cross-fertilisation of ideas between these three major areas of knowledge, with only these particular examples currently being unfolded. We now address the characterisation of additional tasks based on betting and risk-aversion.

\subsection{Arimoto's mutual information and noisy quantum state betting games}
\label{ss:result2}

We now naturally would like to address a characterisation for nQSB games in the same vein that their standard counterpart. Intuitively, we are now addressing a general quantum channel $\mathcal{N}(\cdot)$ as a new ingredient, and that Bob is still in charge of the decoding measurement $\mathbb{M}$. From the noiseless scenario, we understand that Arimoto-like quantities are giving account for the amount of \emph{side information} being conveyed to Bob. When we consider Bob using a fixed measurement, a decisive factor that naturally emerges is the resource of \emph{informativeness}, because this resources defines the frontier for the cases when side information can or cannot be transmitted. In noisy QSB games the other hand, with a general channel $\mathcal{N}(\cdot)$, the same reasoning leads to consider the resource of \emph{non-constant channels}, this, because they will effectively \emph{destroy} the side information carried by the state since $p(g|x) = \tr[N_g \mathcal{N'}(\rho_x)] = \tr(N_g \rho_\mathcal{N'}) = p(g)$, for all constant channels $\mathcal{N'}$, and for all measurements $\mathbb{N}$. The following result confirms this intuition, and consequently characterises nQSB games.
\begin{result}\label{R_nonconstant}
	Consider a nQSB game defined by the pair $(o^{\sgn(\alpha)c}_X, \mathcal{E})$ with constant odds as $o^{\sgn(\alpha)c}(x) \coloneqq \sgn(\alpha) C$, $C>0$, $\forall x$, an ensemble of states $\mathcal{E} = \{\rho_x,p(x)\}$. Consider a Gambler playing this game being able to implement any measurement $\mathbb{M}$, and having access to a fixed channel $\mathcal{N}$. We want to compare this first Gambler against a second Gambler also being allowed to implement any measurement $\mathbb{N}$, but now having access only to constant channels $\mathcal{N}' \in \mathcal{C}$. Consider both Gamblers with the same attitude to risk, meaning that they are represented by isoelastic functions $u_R(W)$ with the risk parametrised as $R(\alpha) \coloneqq 1/\alpha$. Each Gambler is allowed to play the game with optimal betting strategies, meaning they can each propose a betting strategy $b_{X|G}$ independently from each other. Remembering that the Gamblers are interested in maximising the isoelastic certainty equivalent (ICE), we have:
	{\small
	\begin{align}
    	&I_{\alpha}
    	(X;G)_{\mathcal{E},\mathcal{N}}
    	=
    	\label{eq:result5}
        \\
        \nonumber
        &
            \sgn(\alpha)
            \log
        	\left[
        	\frac{
        		\displaystyle
        		\max_{\mathbb{M}}
        		\max_{
        			b_{X|G}
        		}
        		\,
        		w^{\rm nQSB}_{1/\alpha}
        		\left(
        		b_{X|G}
        		,
        		\mathbb{M}
        		,
        		o_X^{\sgn(\alpha)c}
        		,
        		\mathcal{E}
        		,
        		\mathcal{N}
        		\right)
        	}{
        		\displaystyle
        		\max_{\substack{
        		\mathcal{N'}\in
        		\mathcal{C}
        		}}
        		\max_{\mathbb{N}}
        		\max_{
        			b_{X|G}
        		}
        		\,
        		w^{\rm nQSB}_{1/\alpha}
        		\left(
        		b_{X|G}
        		,
        		\mathbb{N}
        		,
        		o_X^{\sgn(\alpha)c}
        		,
        		\mathcal{E}
        		,
        		\mathcal{N}'
        		\right)
        	}
        	\right]
        	.
    \end{align}}
	This means that Arimoto's noisy mutual information quantifies the ratio of the ICE with risk $R(\alpha) \coloneqq 1/\alpha$ of the nQSB game defined by $(o^{\sgn(\alpha)c}_X, \mathcal{E})$, when the nQSB games are being played with the best betting strategy, and when we compare a Gambler implementing a fixed channel $\mathcal{N}$ against a Gambler using any constant channel $\mathcal{N}'\in \mathcal{C}$.
\end{result}

The proof of this result follows a similar argument than that of result 1. We have seen that two natural resources have emerged, or equivalently, two sets of free objects: i) the set of uninformative measurements and ii) the set of constant channels. We then wonder whether the results so far presented are unavoidably linked to these particular resources or, on the other hand, whether they are particular cases of a more general underlying structure governing the relationship between information-theoretic quantities and operational tasks for general QRTs. We address such a question in the next subsection, where we address an extension of these results to general QRTs of measurements and channels with arbitrary resources. 

\subsection{QSB and noisy QSB games for general QRTs of measurements and channels}
\label{ss:result3}

We have seen that both uninformative measurements and non-constant channels are related to Arimoto's mutual information, and we now want to address general resources. In order to do this we can expect to need quantities which are more general than Arimoto's mutual information. We now consider the \emph{Arimoto's gaps} introduced in the previous sections, and provide operational characterisations for these information-theoretic quantities in terms of QSB and nQSB games as follows.

\begin{result}\label{R_measurements_channels}
Consider a set of free measurements as $\mathbb{F}$ and a couple $(\mathcal{E}, \mathbb{M})$, then, \emph{Arimoto's gap on POVMs} of order $\alpha \in \mathds{\overline{R}}$ for such a couple can be written as:
\begin{align}
	&
	G_{\alpha}^{\mathbb{F}}
	(X;G)_{\mathcal{E},\mathbb{M}}
	=
	\label{eq:61}
	\\
    &
        \sgn(\alpha)
        \log
    	\left[
    	\frac{
    		\displaystyle
    		\max_{
    			b_{X|G}
    		}
    		\,
    		w^{\rm QSB}_{1/\alpha}
    		\left(
    		b_{X|G}
    		,
    		o_X^{\sgn(\alpha)c}
    		,
    		\mathcal{E}, \mathbb{M}
    		\right)
    	}{
    		\displaystyle
    		\max_{\mathbb{N}\in \mathbb{F}}
    		\max_{
    			b_{X|G}
    		}
    		\,
    		w^{\rm QSB}_{1/\alpha}
    		\left(
    		b_{X|G}
    		,
    		o_X^{\sgn(\alpha)c}
    		,
    		\mathcal{E}, \mathbb{N}
    		\right)
    	}
    	\right]
    	.
    	\nonumber
\end{align}
Similarly, consider a set of free channels $\mathcal{F}$ and a triple $(\mathcal{E}, \mathbb{M}, \mathcal{N})$, then, \emph{Arimoto's gap on channels} of order $\alpha \in \mathds{\overline{R}}$ for such a triple can be written as:
{\small\begin{align}
	&
	G_{\alpha}^{\mathcal{F}}
	(X;G)_{\mathcal{E},\mathcal{N}}
	=
	\label{eq:62}
	\\
    &
        \sgn(\alpha)
        \log
    	\left[
    	\frac{
    		\displaystyle
    		\max_{\mathbb{M}}
    		\max_{
    			b_{X|G}
    		}
    		\,
    		w^{\rm nQSB}_{1/\alpha}
    		\left(
    		b_{X|G}
    		,
    		o_X^{\sgn(\alpha)c}
    		,
    		\mathcal{E}, 
    		\mathbb{M}, 
    		\mathcal{N}
    		\right)
    	}{
    		\displaystyle
    		\max_{
    		\mathcal{\widetilde{N}}
    		\in
    		\mathcal{F}
    		}
            \max_{\mathbb{N}}
    		\max_{
    			b_{X|G}
    		}
    		\,
    		w^{\rm nQSB}_{1/\alpha}
    		\left(
    		b_{X|G}
    		,
    		o_X^{\sgn(\alpha)c}
    		,
    		\mathcal{E},
    		\mathbb{N},
    		\mathcal{\widetilde{N}}
    		\right)
    	}
    	\right]
    	.
    	\nonumber
\end{align}}
This means that Arimoto-type gaps quantify the usefulness of a given measurement (channel) $\mathbb{M}$ $(\mathcal{N})$ when playing QSB (nQSB) games, in comparison with the best free measurements (channels) $\mathbb{N}\in\mathbb{F}$ $(\mathcal{\widetilde{N}} \in \mathcal{F})$.
\end{result}

The proof of \cref{R_measurements_channels} follows a similar logic to that of result 1 but, for completeness, we present its proof in \cref{AR_measurements_channels}. It is interesting to note the level of generality of this result. This results holds true for any $\alpha \in \mathds{\overline{R}}$, any ensemble $\mathcal{E}$, any measurement $\mathbb{M}$, any channel $\mathcal{N}$, as well as any reasonable and physically motivated choices of sets of free measurements $\mathbb{F}$ and free channels $\mathcal{F}$. In particular, by specifying the sets of free objects we can recover some of the previous results as corollaries.

\begin{corollary}
    Imposing the set of free measurements to be the set of uninformative measurements in \eqref{eq:61} $(\mathbb{F} = \mathbb{UI})$, we recover result 1 \eqref{eq:result1}. Similarly, imposing the set of free channels to be the set of constant channels in \eqref{eq:62} $(\mathcal{F}=\mathcal{C})$, we recover result 2 \eqref{eq:result5}.
\end{corollary}

We have so far addressed QSB games and more generally nQSB games. The main idea behind these operational tasks is the inclusion of the concept of \emph{betting}, which is represented by the constant relative risk aversion (CRRA) coefficient $R$, and which is ultimately related to the R\'enyi parameter as $R=1/\alpha$. We now address the fact that the concept of \emph{betting} is an useful and powerful concept that allows for the generalisation of additional operational tasks. In particular, we now address the characterisation of quantum channel betting (QCB) games.

\subsection{QCB games and QRTs of states and state-measurement pairs}
\label{ss:result4}

Similarly to the case for noisy quantum state betting (nQSB) games, we would now like to characterise quantum subchannel betting (QCB) games in terms of information-theoretic quantities. We now provide an operational interpretation for Arimoto-type quantities for QRTs of states and hybrid multi-object scenarios, in terms of quantum channel betting (QCB) games.

\begin{result} \label{R_states}
Consider a set of free states as ${\rm F}$ and a triple $(\Lambda, \mathbb{M}, \rho)$, then, \emph{Arimoto's gap on states} of order $\alpha \in \mathds{\overline{R}}$ for such a triple can be written as:
\begin{align}
	&
	G_{\alpha}^{{\rm F}}
	(X;G)_{\Lambda,\mathbb{M},\rho}
	=
	\\
    &
        \sgn(\alpha)
        \log
    	\left[
    	\frac{
    		\displaystyle
    		\max_{
    			b_{X|G}
    		}
    		\,
    		w^{\rm QCB}_{1/\alpha}
    		\left(
    		b_{X|G}
    		,
    		o_X^{\sgn(\alpha)c}
    		,
    		\Lambda,\rho,\mathbb{M}
    		\right)
    	}{
    		\displaystyle
    		\max_{\sigma \in {\rm F}}
    		\max_{
    			b_{X|G}
    		}
    		\,
    		w^{\rm QCB}_{1/\alpha}
    		\left(
    		b_{X|G}
    		,
    		o_X^{\sgn(\alpha)c},
    		\Lambda,\sigma,\mathbb{M}
    		\right)
    	}
    	\right]
    	.
    	\nonumber
\end{align}
Similarly, consider a set of free states  ${\rm F}$, a set of free measurements $\mathbb{F},$ and a triple $(\Lambda, \mathbb{M}, \rho)$, then, \emph{Arimoto's gap on state-measurement pairs} of order $\alpha \in \mathds{\overline{R}}$ for such a triple can be written as:
\begin{align}
	&
	G_{\alpha}^{{\rm F}, \mathbb{F}}
	(X;G)_{\Lambda,\mathbb{M},\rho}
	=
	\\
	&
    \sgn(\alpha)
    \log
	\left[
	\frac{
		\displaystyle
		\max_{
			b_{X|G}
		}
		\,
		w^{\rm QCB}_{1/\alpha}
		\left(
		b_{X|G}
		,
		o_X^{\sgn(\alpha)c}
		,
		\Lambda,\rho,\mathbb{M}
		\right)
	}{
		\displaystyle
		\max_{\substack{
		\sigma \in {\rm F}
		\\
        \mathbb{N}\in \mathbb{F}
        }
        }
        \max_{b_{X|G}}
		\,
		w^{\rm QCB}_{1/\alpha}
		\left(
		b_{X|G}
		,
		o_X^{\sgn(\alpha)c}
		,
		\Lambda,\sigma,\mathbb{N}
		\right)
	}
	\right]
	.
	\nonumber
\end{align}
These two statements mean that Arimoto's gap quantifies the usefulness of resourceful objects when compared to gamblers only having access to free objects.
\end{result}

The proof of this result follows a similar logic than that of result 1 but, for completeness, we present its proof in \cref{AR_states}. Similarly to the case for QSB and nQSB games, we have that quantum channel betting (QCB) games can also be characterised by means of Arimoto-type information-theoretic quantities, for single-object QRTs of states with arbitrary resources, but also for more exotic scenarios as the case of multi-object QRTs of state-measurement pairs. The second statement generalises some of the multi-object results presented in \cite{MO}, which considered the cases for $\alpha \in \{+\infty, -\infty\}$, and so this result generalises this to the whole extended line of real numbers $\alpha \in \mathds{\overline{R}}$.

\subsection{Arimoto's mutual information and horse betting games in the classical regime}
\label{ss:result5}

We now consider operational tasks based on betting and risk-aversion in the form of horse betting (HB) games with risk and side information, without making reference to quantum theory, and derive a result interpreting Arimoto's mutual information as quantifying the advantage provided by side information when playing such horse betting games.

We consider here the Gambler now having access to a random variable $G$, which is potentially correlated with the outcome of the `horse race' $X$ and therefore, the Gambler can try to use this for her/his advantage. This means that these horse betting games are defined by the pair $(o_X,p_{GX})$, and the Gambler is in charge of proposing the betting strategy $b_{X|G}$. We highlight here that this contrasts the case of QSB games, because there the Gambler could in principle be in charge of intervening in the conditional PMF $p_{G|X}$, as the Gambler had access to a measurement and $p_{G|X}=\tr(M_g\rho_x)$, whilst here on the other hand, $p_{GX}=p_{G|X}p_X$ is a given, and the Gambler cannot in principle influence the PMF $p_{G|X}$. However, the figure of merit is still the isoelastic certainty equivalent for risk $R \in (-\infty,1) \cup (1,\infty)$ which is now written as:
\begin{multline}
	w^{ICE}_R
	(b_{X|G},o_X,p_{XG})\\
	\coloneqq
	\left[
	\sum_{g,x}
	\big[
	b(x|g)
	o(x)\big]^{1-R}
	p(x,g)
	\right]^\frac{1}{1-R}.
	\label{eq:HB_SI_W}
\end{multline}
The cases $R \in\{1,\infty,-\infty\}$ are defined again by continuous extension of \eqref{eq:HB_SI_W}. A HB game is then specified by the pair $(o_X,p_{GX})$, and the Gambler plays this game with a betting strategy $b_{X|G}$.

Horse betting games were characterised by Bleuler, Lapidoth, and Pfister (BLP), in terms of the BLP-CR divergence \cite{BLP1} (see \cref{AA} for more details on this). We now modify these tasks in order to consider both gain games (when the odds are positive) and loss games (when the odds are negative), and relate Arimoto's mutual information to HB games with the following result, which can be derived in a similar manner as the previous ones.

\begin{result} \label{R5}
	Consider a horse betting game defined by the pair $(o^{\sgn(\alpha)c}_X, p_{XG})$ with constant odds as $o^{\sgn(\alpha)c}(x) \coloneqq \sgn(\alpha)C$, $C>0$, $\forall x$, and a joint PMF $p_{XG}$. Consider a Gambler playing this game having access to the side information $G$, against a Gambler without access to any side information. Consider both Gamblers with the same attitude to risk, meaning they are represented by isoelastic functions $u_R(w)$ with the risk parametrised as $R(\alpha) \coloneqq 1/\alpha$. The Gamblers are allowed to play these games with the optimal betting strategies, which they can each choose independently from each other. Remembering that the Gamblers are interested in maximising the isoelastic certainty equivalent (ICE), we have the following relationship:
	\begin{multline}
    	I_{\alpha}(X;G)
    	\\
    	=
    	\sgn(\alpha)
        \log
    	\left[
    	\frac{
    		\displaystyle
    		\max_{b_{X|G}}
    		\,
    		w^{ICE}_{1/\alpha}
    		(
    		b_{X|G}
    		,
    		o^{\sgn(\alpha)c}_X
    		,
    		p_{XG}
    		)
    	}{
    		\displaystyle
    		\max_{b_{X}}
    		\,
    		w^{ICE}_{1/\alpha}
    		(
    		b_X
    		,
    		o^{\sgn(\alpha)c}_X
    		,
    		p_X
    		)
    	}
    	\right]
    	.
    \end{multline}
	This means that Arimoto's mutual information quantifies the ratio of the isoelastic certainty equivalent with risk $R(\alpha) \coloneqq 1/\alpha$ of the games defined by $(o^{\sgn(\alpha)c}_X, p_{XG})$, when each HB game is played with the best betting strategy, and we compare the performance of a first Gambler who makes use of the side information $G$, against a second gambler which has no access to side information.
\end{result}
We emphasise that this result is purely ``classical", as it does not invoke any elements from quantum theory. This result also complements a previous relationship between HB games and the BLP-CR divergence \cite{BLP1}. Here on the other hand we characterise instead \emph{the ratio} between the two HB scenarios, where we compare a first gambler with access to side information against a second gambler having no access to side information. We now address a particular known case as the following corollary. 

\begin{corollary}
	In the case $\alpha=1$, which means HB games with risk aversion given by $R=1$, we get:
	\begin{multline}
	I(X;G)
	=
	\max_{b_{X|G}}
	U_0
	(
	b_{X|G},o_X^{c},p_{XG}
	)\\
	-
	\max_{b_X}
	U_0
	(
	b_{X},o_X^{c},p_{X}
	),
	\end{multline}
	with $I_{1}(X;G)= I(X;G)$ the standard mutual information, and $U_0 \coloneqq \log w^{ICE}_0$ the logarithm of the isoelastic certainty equivalent. This is a particular case of a relationship known to hold for all odds $o(x)$ \cite{CT, LN_moser}. 
\end{corollary}

We now come back to the QRT of measurement informativeness, and explore further connections between Arimoto's mutual information, QSB games, and additional information-theoretic quantities in the form of quantum R\'enyi divergences and resource monotones.

\subsection{Quantum R\'enyi divergences}
\label{e:divergences} 

Considering that the KL-divergence is of central importance in classical information theory, it is natural to consider quantum-extensions of such quantity. There are many ways to define quantum R\'enyi divergences \cite{review_qrd, petz-renyi, sandwiched1, sandwiched2, geometric, sharp, measured1}, with most of the effort being concentrated on divergences as a functions of quantum states. Recently however, divergences and entropies for additional objects like channels and measurements have been started to be explored \cite{qrd_channels1, qrd_channels2, channel_entropy}. We are now interested in addressing quantum R\'enyi divergences for measurements. The approach we take here takes inspiration from both: measured R\'enyi divergences for states \cite{measured3, measured2, measured1}, as well as R\'enyi conditional divergences in the classical domain \cite{sibson, csiszar, BLP1}. Explicitly, we invoke the measures for R\'enyi conditional  divergences, and use them to define  measured R\'enyi divergences for measurements.
\begin{definition} (Measured quantum R\'enyi divergence of Sibson)
	The measured R\'enyi divergence of Sibson of order $\alpha \in \mathds{\overline R}$ and a set of states $\mathcal{S}=\{\rho_x\}$ of two measurements $\mathds{M}=\{M_g\}$ and $\mathds{N}=\{N_g\}$ is given by:
	\begin{align}
    	D_{\alpha}^{\mathcal{S}} 
    	(\mathds{M}||\mathds{N})
        \coloneqq
            \max_{p_X}
            D_{\alpha}
            \left(
            p_{G|X}^{({\mathbb{M}},\mathcal{S})}
            \Big|\Big|
            q_{G|X}^{({\mathbb{N}},\mathcal{S})}
            \Big|
            \,
            p_X
            \right).
    \end{align}
	with the maximisation over all PMFs $p_X$, and the conditional PMFs $p_{G|X}^{({\mathbb{M}},\mathcal{S})}$ and $q_{G|X}^{({\mathbb{N}},\mathcal{S})}$ given by $p(g|x)\coloneqq \tr (M_g \rho_x)$, $q(g|x)\coloneqq \tr (N_g \rho_x)$, respectively, and $D(\cdot||\cdot|\cdot)$ the conditional R\'enyi divergence of Sibson \cite{sibson} which is defined in \cref{AA}.
\end{definition}

We now use this measured R\'enyi divergence in order to define a distance measure with respect to a free set of interest, the set of uninformative measurements in this case.
\begin{definition} (Measurement informativeness measure of Sibson)
The measurement informativeness measure of Sibson of order $\alpha \in \mathds{\overline R}$ and set of states $\mathcal{S}$ of a measurement $\mathbb{M}$ is given by:
	\begin{align}
	    E_{\alpha}^{\mathcal{S}} 
		(\mathds{M})
        \coloneqq
          \min_{\mathds{N}\in {\rm UI}}
		D_{\alpha}^{\mathcal{S}} 
		(\mathds{M}||\mathds{N}),
        \label{eq:MIM}
    \end{align}
	with the minimisation over all uninformative measurements.
\end{definition}

Interestingly, it turns out that this quantity becomes equal to a quantity which we have already introduced.
\begin{result}\label{Rdivergences}
	The informativeness measure of Sibson is equal to the R\'enyi capacity of order $\alpha \in \mathds{\overline R}$ of the measurement $\mathbb{M}$ as:
	\begin{align}
		E_{\alpha}^{\mathcal{S}}
		(\mathds{M})
		=
		C_\alpha
		\left(
		p_{G|X}^{(\mathds{M}, \mathcal{S})}
		\right),
	\end{align}
	with the quantum-classical channel associated to the measurement $\mathbb{M}$ \eqref{eq:qc}. 
\end{result}
The proof of this result is in \cref{ARdivergences}. This result establishes a connection between R\'enyi mutual information (which are used to define the R\'enyi channel capacity) and quantum R\'enyi divergences of measurements (which are used to define the measurement informativeness measure). We now consider the quantity $E_{\alpha}(\mathds{M}) \coloneqq \max_{\mathcal{S}} E_{\alpha}^{\mathcal{S}} (\mathds{M})=C_\alpha(\Lambda_\mathds{M})$ and analyse the particular cases of $\alpha \in\{\infty,-\infty\}$.
\begin{corollary}
	The measurement informativeness measure of Sibson recovers the generalised robustness and the weight of resource at the extremes $\alpha \in\{\infty,-\infty\}$ as:
	\begin{align}
		E_{\infty}(\mathds{M})
		&=
		\log
		\left[
		1+{\rm R}(\mathbb{M})
		\right],\\
		E_{-\infty}(\mathds{M})
		&=
		-
		\log
		\left[
		1-{\rm W}(\mathbb{M})
		\right],
	\end{align}
	with the generalised robustness of informativeness \eqref{eq:RoI} \cite{SL}, and the weight of informativeness \eqref{eq:WoI} \cite{DS}.
\end{corollary}
This result follows from the fact that the R\'enyi channel capacity becomes the accessible min-information and the excludible information at the extremes  $\alpha \in\{\infty,-\infty\}$, together with the results from \cite{SL} and \cite{DS}. Result 3 therefore establishes a connection between R\'enyi mutual informations and quantum R\'enyi divergences of measurements. Inspired by these results, we now proceed to propose a family of resource monotones.

\subsection{Resource monotones}
\label{e:monotones} 

Resource quantifiers are special cases of resource monotones, which are central objects of study within QRTs \cite{RT_review, TG}. Two common families of resource monotones are the so-called \emph{robustness-based} \cite{RoE, RoNL_RoS_RoI, RoS, RoA, RoC, SL, RoNL_RoS_RoI, RoT, RoT2, RT_magic, citeme1, citeme2} and \emph{weight-based} \cite{WoE, EPR2, WoS, RoNL_RoS_RoI, WoAC} resource monotones. Inspired by the previous results, we now define measures which turn out to be monotones for the order induced by the simulability of measurements and furthermore, that this new family of monotones recover, at its extremes, the generalised robustness and the weight of informativeness.
\begin{definition} ($\alpha$-measure of informativeness)
	The $\alpha$-measure of informativeness of order $\alpha \in \mathbb{\overline R}$ of a measurement $\mathbb{M}$ is given by:
	\begin{align}
	{\rm M}_{\alpha}
	(\mathds{M})
    \coloneqq
      \sgn(\alpha)
		2^{
			\sgn(\alpha)
			E_{\alpha}
			(\mathds{M})
		}
		-
		\sgn(\alpha)
		,
    \label{eq:aR}
    \end{align}
	with $E_{\alpha}(\mathds{M}) \coloneqq \max_{\mathcal{S}} E_{\alpha}^{\mathcal{S}} (\mathds{M})$ and the measurement informativeness measure defined in \eqref{eq:MIM}.
\end{definition}
The motivation behind the proposal of this resource measure is because: i) it recovers the generalised robustness and the weight of resource as ${\rm M}_{\infty}(\mathds{M})
={\rm R}(\mathds{M})$ and ${\rm M}_{-\infty}(\mathds{M})
={\rm W}(\mathds{M})$ and, ii) it allows the following operational characterisation.

\begin{remark}
	They $\alpha$-measure of informativeness of order $\alpha \in \mathbb{\overline R}$ of a measurement $\mathbb{M}$ characterises the performance of the measurement $\mathbb{M}$, when compared to the performance of all possible uninformative measurements, when playing the same QSB game as:
	\begin{align}
	\max_{\mathcal{E}}
	\frac{
		\displaystyle
		\max_{b_{X|G}}
		\,
		w^{ICE}_{1/\alpha}
		\left(
		b_{X|G}
		,
		\mathbb{N}
		,
		o^{c}_X
		,
		\mathcal{E}
		\right)
	}{
		\displaystyle
		\max_{\mathbb{N}\in {\rm UI}}
		\max_{
			b_{X|G}
		}
		\,
		w^{ICE}_{1/\alpha}
		\left(
		b_{X|G}
		,
		\mathbb{N}
		,
		o^{c}_X
		,
		\mathcal{E}
		\right)
	}
	&=
	1+{\rm M}_{\alpha}(\mathbb{M}),
	\\
	\min_{\mathcal{E}}
	\frac{
		\displaystyle
		\max_{b_{X|G}}
		\,
		w^{ICE}_{1/\alpha}
		\left(
		b_{X|G}
		,
		\mathbb{N}
		,
		o^{-c}_X
		,
		\mathcal{E}
		\right)
	}{
		\displaystyle
		\max_{\mathbb{N}\in {\rm UI}}
		\max_{
			b_{X|G}
		}
		\,
		w^{ICE}_{1/\alpha}
		\left(
		b_{X|G}
		,
		\mathbb{N}
		,
		o^{-c}_X
		,
		\mathcal{E}
		\right)
	}
	&=
	1-{\rm M}_{\alpha}(\mathbb{M})
	,
	\end{align}
	for $\alpha\geq 0$ and $\alpha<0$, respectively. These two equalities follow directly from the definitions and the previous results. 
\end{remark}
This result is akin to the connections between generalised robustness characterising discrimination games, and the weight of resource characterising exclusion games. We now also have that the $\alpha$-measure of informativeness defines a resource monotone for the simulability of measurements.
\begin{result} \label{Rmonotones} (The $\alpha$-measure of informativeness is a resource monotone)
	The $\alpha$-measure of informativeness \eqref{eq:aR} defines a resource monotone for the simulability of measurements, meaning that it satisfies the following properties. (i) Faithfulness: ${\rm M_\alpha} (\mathbb{M}) = 0 
	\leftrightarrow \mathbb{M} = \{M_a=q(a )\mathds{1}\}$ and (ii) Monotonicity under measurement simulation: $\mathbb{N} \preceq \mathbb{M} \rightarrow {\rm M_\alpha}(\mathbb{N})\leq {\rm M_\alpha}(\mathbb{M})$.
\end{result}
The proof of this result is in \cref{ARmonotones}. It would be interesting to find a geometric interpretation of this measure, in a similar manner that its two extremes admit a geometric interpretation as in \eqref{eq:RoI} and \eqref{eq:WoI}, as well as to explore additional properties, like convexity, in order to talk about it being a resource quantifier. It would also be interesting to explore additional monotones, in particular, whether the isoelastic certainty equivalent forms a complete set of monotones for the simulability of measurements, this, given that this holds for the two extremes at plus and minus infinity. 

Altogether, the above results establish a four-way task-mutual information-divergence-monotone correspondence for the QRT of measurement informativeness, by means of a risk aversion factor parametrised by the R\'enyi parameter $\alpha$ as $R(\alpha)=1/\alpha$, as quantitatively summarised as four chains of equalities below,  and qualitatively depicted in \cref{fig:fig}.

\begin{figure}[h!]
	\includegraphics[scale=0.5]{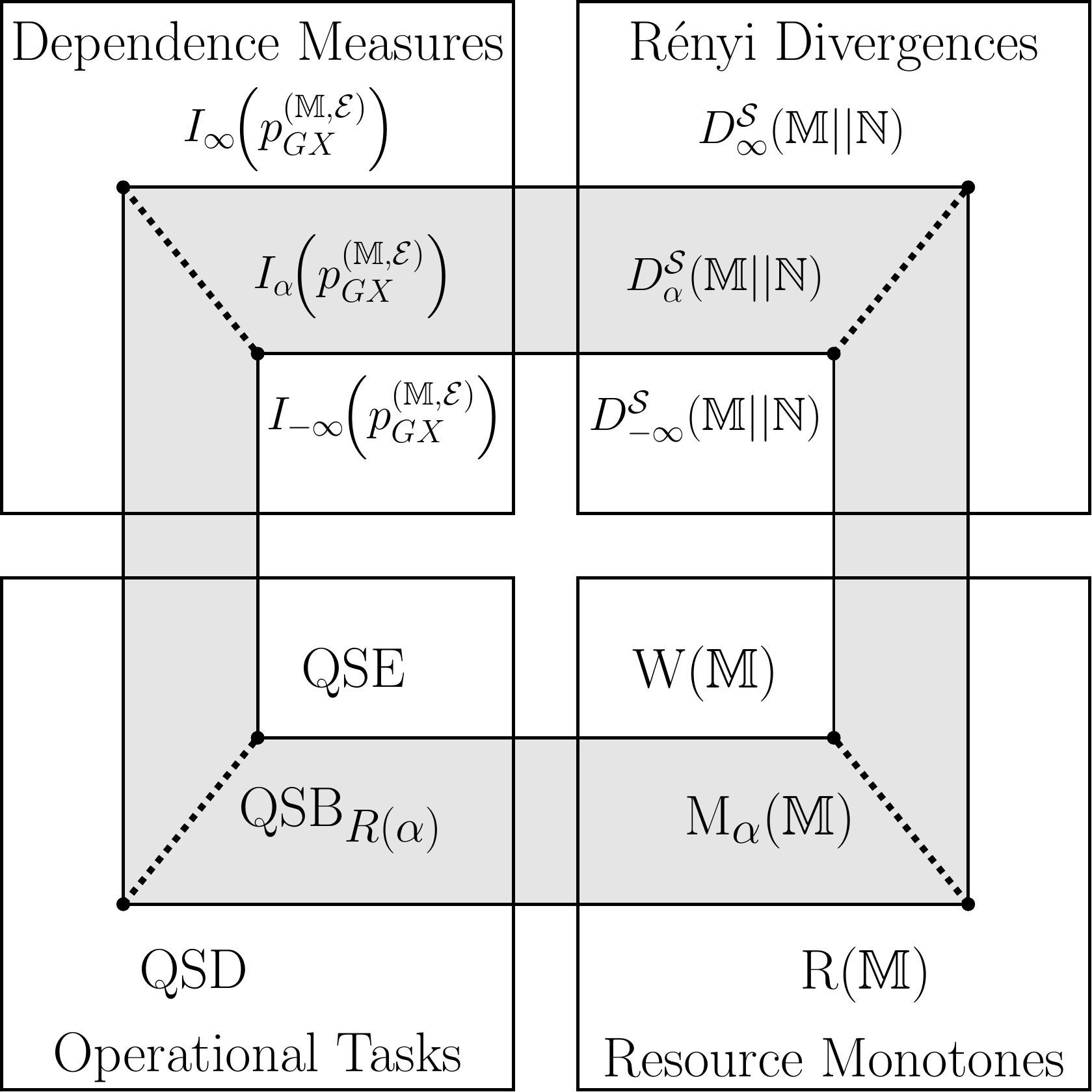}
	\caption{\label{fig:fig}
		A four-way correspondence for the QRT of measurement informativeness. The correspondence is parametrised by the R\'enyi parameter $\alpha \in \mathds{R}\cup\{\infty, -\infty\}$. The outer rectangle represents $\alpha=\infty$, the inner rectangle represents $\alpha=-\infty$, and the shaded region in-between represents the values $\alpha \in \mathds{R}$. This four-way correspondence links: operational tasks, dependence measures, R\'enyi divergences, and resource monotones. The operational task is quantum state betting (QSB) played by a Gambler with risk aversion $R(\alpha)=1/\alpha$. This task generalises quantum state discrimination (QSD) (recovered when $\alpha\rightarrow \infty$), and quantum state exclusion (QSE) (recovered when  $\alpha \rightarrow -\infty$). $I_\alpha(p_{GX}^{(\mathbb{M},\mathcal{E})})$ is Arimoto's dependence measure, from which we recover the accessible information $I_{\infty}^{\rm acc}(\Lambda_\mathds{M})$ when $\alpha\rightarrow \infty$, and the excludible information $I_{-\infty}^{\rm exc}(\Lambda_\mathds{M})$ when $\alpha\rightarrow -\infty$, with $\Lambda_{\mathds{M}}$ the measure-prepare channel of the measurement $\mathds{M}$.
		We introduce $D_{\alpha}^{\mathcal{S}} ( \mathds{M} || \mathds{N})$, the quantum R\'enyi divergence of two measurements $\mathbb{M}$ and $\mathbb{N}$ for a given set of states $\mathcal{S}=\{\rho_x\}$.  We also introduce ${\rm M}_\alpha(\mathds{M})$, a new family of resource monotones, which generalise the robustness of informativeness $\rm R(\mathds{M})$ (when $\alpha\rightarrow \infty)$ and the weight of informativeness $\rm W(\mathds{M})$ (when $\alpha\rightarrow -\infty$). The outer rectangle was uncovered in \cite{SL}, whilst the inner rectangle was first uncovered in \cite{DS}. The main set of results of this paper is to fill the shaded region, and connect these two correspondences for all $\alpha\in \mathds{R}$.
	}
\end{figure}

\onecolumngrid
\begin{summary*}
	Four-way quantum correspondence between operational tasks, mutual information measures, quantum R\'enyi divergences, and resource monotones for the QRT of measurement informativeness. This is quantitatively summarised in the following four chains of equalities, going from $\alpha \rightarrow \infty$ passing through $\alpha=0$ until $\alpha \rightarrow -\infty$ as follows (definitions and proofs in the main text and appendices):
	\begin{align*}
	\max_\mathcal{E}
	I_{\infty}(X;G)_{\mathcal{E},\mathbb{M}}
	&=
	\log
	\left[
	\max_\mathcal{E}
	\frac{
		p^{\rm QSD}_{\rm \, succ}
		\left(
		\mathcal{E}, \mathbb{M}
		\right)
	}{
		\max_{\mathbb{N}\in {\rm UI}}
		P^{\rm QSD}_{\rm succ}(\mathcal{E},\mathbb{N})
	}
	\right]
	=	
	\log
	\left[1+{\rm R}(\mathbb{M})\right]
	=
	\max_{\mathcal{S}}
	\min_{\mathbb{N}\in {\rm UI}}
	D_{\infty}^{\mathcal{S}}
	(\mathbb{M}||\mathbb{N}),
	\\
	\max_\mathcal{E}
	I_{\alpha \geq 0}(X;G)_{\mathcal{E},\mathbb{M}}
	&=
	\log
	\left[
	\max_{\mathcal{E}}
	\frac{
		\displaystyle
		\max_{b_{X|G}}
		\,
		w^{ICE}_{1/\alpha}
		\left(
		b_{X|G}, 
		\mathbb{M}
		, 
		o^{c}_X
		,
		\mathcal{E}
		\right)
	}{
		\displaystyle
		\max_{\mathbb{N}\in {\rm UI}}
		\max_{
			b_{X|G}
		}
		\,
		w^{ICE}_{1/\alpha}
		\left(
		b_{X|G}
		, 
		\mathbb{N}
		, 
		o^{c}_X
		,
		\mathcal{E}
		\right)
	}
	\right]
	=
	\log
	\left[1+{\rm M}_{\alpha}
	(\mathbb{M})\right]
	=
	\max_{\mathcal{S}}
	\min_{\mathbb{N}\in {\rm UI}}
	D_{\alpha}^{\mathcal{S}}
	(\mathbb{M}||\mathbb{N}),
	\end{align*}
	\begin{align*}
	\max_\mathcal{E}
	I_{\alpha<0}(X;G)_{\mathcal{E},\mathbb{M}}
	&=
	-
	\log
	\left[
	\min_{\mathcal{E}}
	\frac{
		\displaystyle
		\max_{b_{X|G}}
		\,
		w^{ICE}_{1/\alpha}
		\left(
		b_{X|G}
		, 
		\mathbb{M}
		, 
		o_X^{-c}
		,
		\mathcal{E}
		\right)
	}{
		\displaystyle
		\max_{\mathbb{N}\in {\rm UI}}
		\max_{
			b_{X|G}
		}
		\,
		w^{ICE}_{1/\alpha}
		\left(
		b_{X|G}
		,
		\mathbb{N}
		, 
		o_X^{-c}
		,
		\mathcal{E}
		\right)
	}
	\right]
	=
	-
	\log
	\left[1-{\rm M}_{\alpha}
	(\mathbb{M})\right]
	=
	\max_{\mathcal{S}}
	\min_{\mathbb{N}\in {\rm UI}}
	D_{\alpha}^{\mathcal{S}} 
	(\mathbb{M}||\mathbb{N}),
	\\
	\max_\mathcal{E}
	I_{-\infty}(X;G)_{\mathcal{E},\mathbb{M}}
	&=
	-
	\log
	\left[
	\min_\mathcal{E}
	\frac{
		P^{\rm QSE}_{\rm err}(\mathcal{E},\mathbb{M})	
	}
	{
		\min_{\mathbb{N}\in {\rm UI}}
		P^{\rm QSE}_{\rm err}(\mathcal{E},\mathbb{N})
	}
	\right]
	=
	-
	\log
	\left[1-{\rm W}(\mathbb{M})\right]
	=
	\max_{\mathcal{S}}
	\min_{\mathbb{N}\in {\rm UI}}
	D_{-\infty}^{\mathcal{S}}
	(\mathbb{M}||\mathbb{N})
	.
	\end{align*}
\end{summary*}

\twocolumngrid
\section{Conclusions}
\label{s:conclusions}

In this work, we have proposed that using the ideas of betting, risk-aversion, and utility theory are a powerful way of extending the well studied tasks of quantum state discrimination and quantum state exclusion. We have used this to introduce various quantum operational tasks based on betting and risk-aversion, or \emph{quantum betting tasks}. In particular, we have shown that this places two recently discovered four-way correspondences \cite{SL, DS} into a much broader continuous family of correspondences. For the first time, this shows that there exist deep connections between operational state identification tasks, mutual information measures, R\'enyi divergences, and resource monotones.

The seven main results in this manuscript are the following. First, we relate Arimoto's $\alpha$-mutual information (in the quantum domain) to the quantum state betting games with risk, for the QRT of measurement informativeness. As corollaries of this result, we recover the previous two known relationships relating: i) the accessible information to quantum state discrimination, and ii) the excludible information to quantum state exclusion. Second, we characterise nQSB games for the QRT of non-constant channels in terms of Arimoto's mutual information. Third, we consider a generalisation of the two previous scenarios to general QRTs of measurements with arbitrary resources (beyond that of informativeness) and QRTs of channels with general resources (beyond that of non-constant channels), and relate QSB/nQSB games to Arimoto-type measures. Fourth, we address quantum channel betting (QCB) games, and consider this task for QRTs of states with arbitrary resources, as well as an hybrid scenario in a multi-object regime, addressing QRTs of state-measurements pairs, in which states and measurements are simultaneously considered  in possession of valuable resources. Fifth, we relate  Arimoto's mutual information to horse betting (HB) games with side information in the classical regime, without invoking quantum theory. This result can be seen as giving a very clean operational interpretation of Arimoto's mutual information, showing that it exactly quantifies the advantage provided by \emph{side information}, and that the R\'enyi parameter can be understood operationally as quantifying the risk aversion of a gambler. Sixth, using the insights from the results on the QRT of measurement informativeness, we derive new quantum measured R\'enyi divergences for measurements. Seventh, we introduce resource monotones for the order generated by the simulability of measurements, which additionally recover the resource monotones of \emph{generalised robustness}, as well as the \emph{weight} of informativeness. Finally, results 1, 6, and 7 are elegantly connected via a four-way correspondence, which substantially extended the two correspondences previously uncovered \cite{SL, DS}, which we now understand to be the two extremes of a continuous spectrum.

We believe our results are the start of a much broader and deeper investigation into the use of betting, risk-aversion, utility theory, and other ideas from economics, to obtain a broader unified understanding of many topics in quantum information theory. Our results raise many questions and open up various avenues for future research, a number of which we briefly describe below.

\subsection{Open problems, perspectives, and avenues for future research}\label{s:perspective}

\begin{enumerate}
    \item An exciting broad possibility, is to explore more generally the concept of risk aversion in quantum information theory. This is a concept which we are just starting to understand and incorporate into the theory of information and therefore, we believe this is an exciting avenue of research which could have far-reaching implications when considered for additional operational tasks, like Bell-nonlocal games, and interactive proof systems.
    \item Similarly, the scenario here considered represents the convergence of three major research fields: i) quantum theory, ii) information theory and iii) the theory of games and economic behaviour. Specifically, we borrowed the concept of risk aversion from the economic sciences in order to solve an open problem in quantum information theory. We believe that this is just an example of the benefits that can be obtained from considering the cross-fertilisation of ideas between these three major current research fields. Consequently, it would be interesting to keep importing further concepts (in addition to risk aversion), as well as to explore the other direction, i. e., whether quantum information theory can provide insights into the theory of games and economic behaviour. We believe this can be a fruitful approach for future research. In particular, horse betting games are a particular family of a larger family of tasks which are related to the investment in portfolios \cite{CT}, and it therefore would be interesting to explore quantum versions of the operational tasks that emerge in these scenarios. 
    \item The set of connections we have established here are by means of the R\'enyi entropies, and we have seen that the parameter $\alpha$ is intimately linked to the risk aversion of a gambler. It is interesting to speculate whether other types of connections might be possible. For example, Brandao \cite{WoE_Brandao} previously found a family of entanglement witnesses that encompassed both the generalised robustness and the weight of entanglement. We do not know if this is intimately related with our findings here, or whether our insights might shed further light, e.g.~operational significance, on these entanglement witnesses and their generalisations. 
    \item We were led to introduce new measured quantum R\'enyi divergences for measurements in this work. We believe that they should find relevance and application in settings far removed from the specific setting we considered here. It would also be interesting to further explore their relevance in other other areas within quantum information theory.
    \item We have also introduced new resource monotones, for which we do not yet have a full understanding. In particular, unlike numerous other monotones, these do not yet have an obvious geometric interpretation. It would be interesting to develop such ideas further. 
    \item It would be interesting to explore additional monotones, in particular, whether the isoelastic certainty equivalent $w^{ICE}_{R(\alpha)}$ forms (for all $\alpha$) a complete set of monotones for the order induced by the simulability of measurements, this, given that this is the case for the two extremes at $\alpha \in \{\infty,-\infty\}$ \cite{SL, DS}.
    \item We point out that we have used information-theoretic quantities with the R\'enyi parameter $\alpha$ taking both positive and negative values. Whilst negative values have been explored in the literature, it is fair to say that they have not been the main focus of attention. Here we have proven that information-theoretic quantities with negative orders posses a descriptive power \emph{different} from their positive counterparts and therefore, it would be interesting to explore their usefulness in other information-theoretic scenarios.
\end{enumerate}

\emph{Note added:} After the completion of this manuscript, we became aware of the related work by S. Tirone, et. al. \cite{salvatore}, where they address Kelly betting with a quantum payoff in a continuous variable setting.

\section*{Acknowledgements}

We thank Patryk Lipka-Bartosik, Tom Purves, Noah Linden, and Roope Uola for insightful discussions. We also thank Ryuji Takagi and two anonymous referees for valuable comments on the first version of this manuscript. A.F.D. acknowledges support from COLCIENCIAS 756-2016 and the UK EPSRC (EP/L015730/1). P.S. acknowledges support from a Royal Society URF (UHQT).

\onecolumngrid
\appendix

\section{R\'enyi divergence, conditional-R\'enyi divergences, mutual informations, and R\'enyi channel capacity} \label{AA}

In this Appendix we address the information-theoretic quantities represented in Fig.~\ref{fig:fig2} namely, the R\'enyi divergence, the conditional R\'enyi (CR) divergences of Sibson, Csisz\'ar, and Bleuler-Lapidoth-Pfister, their respective mutual informations, and the R\'enyi channel capacity. 
\begin{figure}[h!]
	\includegraphics[scale=0.58]{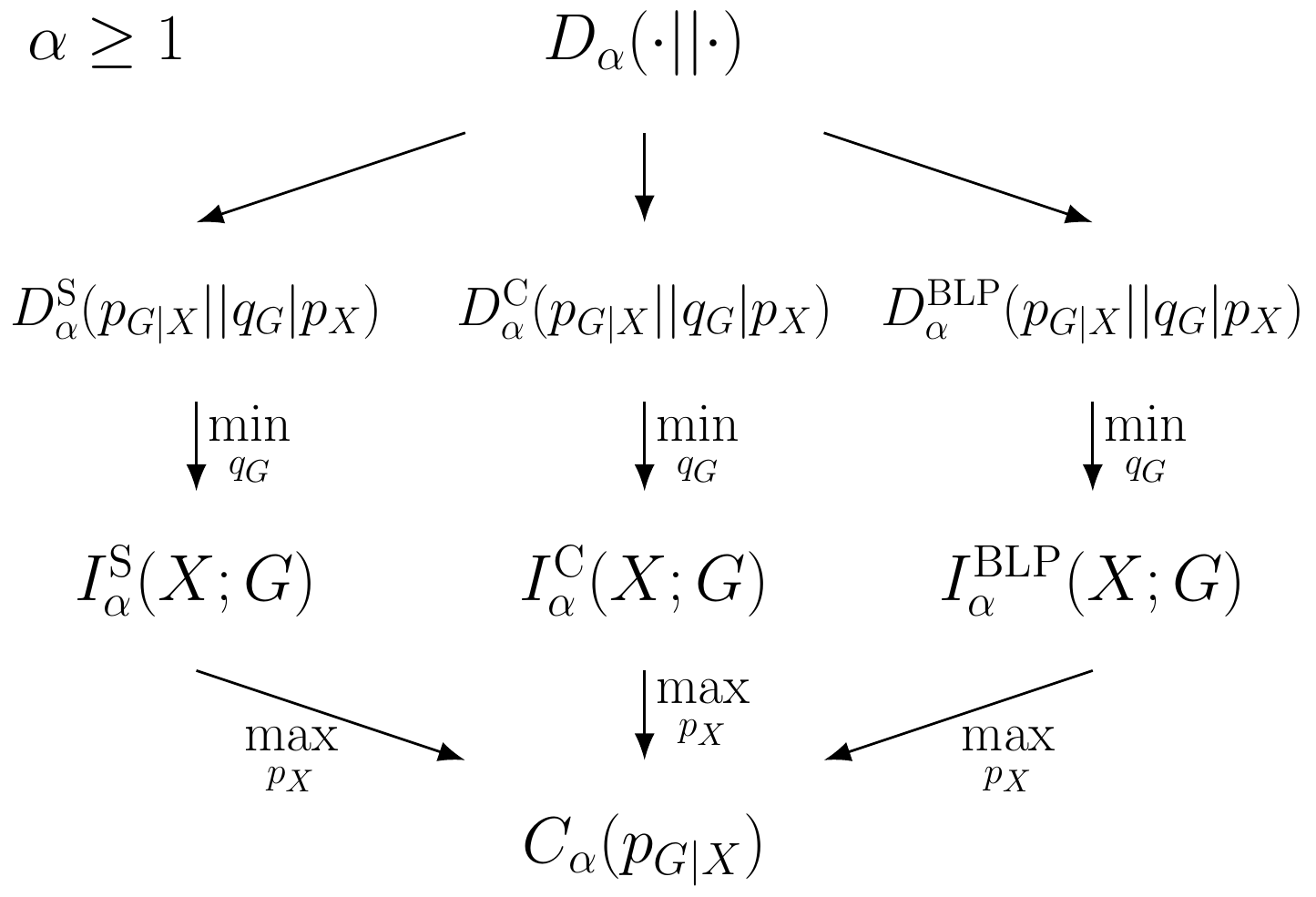}
	\caption{
	Hierarchical relationship between the R\'enyi divergence $D_\alpha(\cdot||\cdot)$, conditional R\'enyi divergences $D^{\rm V}_\alpha(\cdot||\cdot|\cdot)$, mutual informations $I^{\rm V}_\alpha(X;G)$, and the R\'enyi channel capacity $C_\alpha (p_{G|X})$ with $\alpha \geq 1$, and $\rm V \in\{S,C,BLP\}$ a label specifying the measures of Sibson \cite{sibson}, Csisz\'ar \cite{csiszar}, and Bleuler-Lapidoth-Pfister \cite{BLP1}. The mutual information associated to the BLP-conditional-R\'enyi divergence was independently derived by Lapidoth-Pfister \cite{LP} and Tomamichel-Hayashi \cite{TH}. In this work we address the capacities generated by Sibson and Arimoto.
	}
	\label{fig:fig2}
\end{figure}

\subsection{The R\'enyi divergence}

\begin{definition} (R\'enyi divergence \cite{renyi, RD})
	The R\'enyi divergence (R-divergence) of order $\alpha \in \mathds{\overline R}$ of PMFs $p_X$ and $q_X$ is denoted as $D_\alpha(p_X||q_X)$. The orders $\alpha \in (-\infty,0)\cup(0,1)\cup(1,\infty)$ are defined as:
	\begin{align}
	D_\alpha(p_X||q_X)
	&\coloneqq
	\frac{
	\sgn(\alpha)
	}{
	\alpha-1
	}
	\log
	\left[
	\sum_x
	p(x)^\alpha
	q(x)^{1-\alpha}
	\right]
	.
	\label{eq:RD}
	\end{align}
	The orders $\alpha \in\{1,0,\infty,-\infty\}$ are defined define by continuous extension of \eqref{eq:RD} as:
	\begin{align}
	D_1(p_X||q_X)
	&\coloneqq D(p_X||q_X),
	\\
	D_0(p_X||q_X)
	&\coloneqq
	-\log 
	\sum_{x \in {\rm supp}(p_X)}
	q(x)
	,
	\\
	D_{\infty}(p_X||q_X)
	&\coloneqq
	\log \max_x \frac{p(x)}{q(x)},
	\\
	D_{-\infty}(p_X||q_X)
	&\coloneqq
	-
	\log \min_x \frac{p(x)}{q(x)}.
	\end{align}
	with the standard Kullback-Leibler (KL) divergence given by $D(p_X||q_X) \coloneqq \sum_x p(x) \log \frac{p(x)}{q(x)}$ \cite{KL1951, CT}.
\end{definition}

\subsection{Conditional-R\'enyi divergences}

\begin{definition} (Sibson's conditional-R\'enyi divergence \cite{sibson}) 
	The Sibson's conditional-R\'enyi divergence (S-CR-divergence) of order $\alpha \in \mathds{\overline R}$ of PMFs $p_{X|G}$, $q_{X|G}$, and $p_X$ is denoted as $D^{\rm S}_\alpha
	(
	p_{G|X}
	||
	q_{G|X}
	|
	p_X
	)$. The orders $\alpha \in (-\infty,0)\cup(0,1)\cup(1,\infty)$ are defined as:
	\begin{equation}
	D^{\rm S}_\alpha
	(
	p_{G|X}
	||
	q_{G|X}
	|
	p_X
	)
	\coloneqq
	\frac{
	\sgn(\alpha)
	}{
	\alpha-1
	}
	\log
	\sum_x
	p(x)
	\sum_g
	p(g|x)^\alpha
	q(g|x)^{1-\alpha}
	.
	\label{eq:SCRD}
	\end{equation}	
	The orders $\alpha \in\{1,0,\infty,-\infty\}$ are defined by continuous extension of \eqref{eq:SCRD} as:
	\begin{align}
		D^{\rm S}_1
		(
		p_{G|X}
		||
		q_{G|X}
		|
		p_X
		)
		&\coloneqq
		D
		(
		p_{G|X}
		||
		q_{G|X}
		|
		p_X
		),
		\\
		D^{\rm S}_0
		(
		p_{G|X}
		||
		q_{G|X}
		|
		p_X
		)
		&\coloneqq
		-\log
		\sum_{x\in {\rm supp}(p_X)}
		p(x)
		\sum_{g\in {\rm supp}(p_{G|X=x})}
		q(g|x),
		\\
		D^{\rm S}_{\infty}
		(
		p_{G|X}
		||
		q_{G|X}
		|
		p_X
		)
		&\coloneqq
		\log 
		\max_{x\in {\rm supp}(p_X)}
		\max_g
		\frac{p(g|x)}{q(g|x)},
		\\
		D^{\rm S}_{-\infty}
		(
		p_{G|X}
		||
		q_{G|X}
		|
		p_X
		)
		& \coloneqq
		-
		\log 
		\min_{x\in {\rm supp}(p_X)}
		\min_g
		\frac{p(g|x)}{q(g|x)},
	\end{align}
with the conditional R\'enyi divergence given by $D
	\left(
	p_{G|X}
	||
	q_{G|X}
	|
	p_X
	\right)
	\coloneqq
	D
	\left(
	p_{G|X}
	p_X
	||
	q_{G|X}
	p_X
	\right)$, the latter being the standard KL-divergence \cite{KL1951, CT}.
\end{definition}
\begin{definition} (Csisz\'ar's conditional-R\'enyi divergence \cite{csiszar}) 
	The Csisz\'ar's conditional-R\'enyi divergence (C-CR-divergence) of order $\alpha \in \mathds{\overline R}$ of PMFs $p_{X|G}$, $q_{X|G}$, and $p_X$ is denoted as $D^{\rm C}_\alpha
	(
	p_{G|X}
	||
	q_{G|X}
	|
	p_X
	)$. The orders $\alpha \in (-\infty,0)\cup(0,1)\cup(1,\infty)$ are defined as:
	\begin{equation}
	D^{\rm C}_\alpha
	(
	p_{G|X}
	||
	q_{G|X}
	|
	p_X
	)
	\coloneqq
	\frac{
	\sgn(\alpha)
	}{
	\alpha-1
	}
	\sum_x
	p(x)
	\log
	\left[
	\sum_g
	p(g|x)^\alpha
	q(g|x)^{1-\alpha}
	\right]
	.
	\label{eq:CCRD}
	\end{equation}	
	The orders $\alpha \in \{1,0,\infty,-\infty\}$ are defined by continuous extension of \eqref{eq:CCRD} as:
	\begin{align}
		D^{\rm C}_1
		(
		p_{G|X}
		||
		q_{G|X}
		|
		p_X
		)
		&\coloneqq
		D
		(
		p_{G|X}
		||
		q_{G|X}
		|
		p_X
		),
		\\
		D^{\rm C}_0
		(
		p_{G|X}
		||
		q_{G|X}
		|
		p_X
		)
		&\coloneqq
		-
		\sum_{x\in {\rm supp}(p_X)}
		p(x)
		\log
		\sum_{g\in {\rm supp}(p_{G|X=x})}
		q(g|x)
		,
		\\
		D^{\rm C}_{\infty}
		(
		p_{G|X}
		||
		q_{G|X}
		|
		p_X
		)
		&\coloneqq
		\sum_{x\in {\rm supp}(p_X)}
		p(x)
		\log 
		\left[
		\max_g
		\frac{p(g|x)}{q(g|x)}
		\right],
		\\
		D^{\rm C}_{-\infty}
		(
		p_{G|X}
		||
		q_{G|X}
		|
		p_X
		)
		& \coloneqq
		-
		\sum_{x\in {\rm supp}(p_X)}
		p(x)
		\log 
		\left[
		\min_g
		\frac{p(g|x)}{q(g|x)}
		\right],
		\end{align}
	with the conditional R\'enyi divergence given by $D
	\left(
	p_{G|X}
	||
	q_{G|X}
	|
	p_X
	\right)
	\coloneqq
	D
	\left(
	p_{G|X}
	p_X
	||
	q_{G|X}
	p_X
	\right)$, the latter being the standard KL-divergence \cite{KL1951, CT}.
\end{definition}
\begin{definition} (Bleuler-Lapidoth-Pfister conditional-R\'enyi divergence \cite{BLP1, thesis_CP}) 
	The Bleuler-Lapidoth-Pfister conditional-R\'enyi divergence (BLP-CR-divergence) of order $\alpha \in \mathds{\overline R}$ of PMFs $p_{X|G}$, $q_{X|G}$, and $p_X$ is denoted as $D^{\rm BLP}_\alpha
	(
	p_{G|X}
	||
	q_{G|X}
	|
	p_X
	)$. The orders $\alpha\in(-\infty,0)\cup(0,1)\cup(1,\infty)$ are defined as:
	\begin{equation}
	D^{\rm BLP}_\alpha
	(
	p_{G|X}
	||
	q_{G|X}
	|
	p_X
	)
	\coloneqq
	\frac{|\alpha|}{\alpha-1}
	\log
	\sum_x
	p(x)
	\left[
	\sum_g
	p(g|x)^\alpha
	q(g|x)^{1-\alpha}
	\right]^{\frac{1}{\alpha}}	
	.
	\label{eq:BLPCRD}
	\end{equation}	
	The orders $\alpha \in\{1,0,\infty,-\infty\}$ are defined by continuous extension of \eqref{eq:BLPCRD} as:
\begin{align}
		D^{\rm BLP}_1
		(
		p_{G|X}
		||
		q_{G|X}
		|
		p_X
		)
		&\coloneqq
		D
		(
		p_{G|X}
		||
		q_{G|X}
		|
		p_X
		),
		\\
		D^{\rm BLP}_0
		(
		p_{G|X}
		||
		q_{G|X}
		|
		p_X
		)
		&\coloneqq
		-\log
		\max_{x\in {\rm supp}(p_X)}
		\sum_{g\in {\rm supp}(p_{G|X=x})}
		q(g|x)
		,
		\\
		D^{\rm BLP}_{\infty}
		(
		p_{G|X}
		||
		q_{G|X}
		|
		p_X
		)
		&\coloneqq
		\log 
		\sum_{x\in {\rm supp}(p_X)}
		p(x)
		\max_g
		\frac{p(g|x)}{q(g|x)},
		\\
		D^{\rm BLP}_{-\infty}
		(
		p_{G|X}
		||
		q_{G|X}
		|
		p_X
		)
		& \coloneqq
		-
		\log 
		\sum_{x\in {\rm supp}(p_X)}
		p(x)
		\min_g
		\frac{p(g|x)}{q(g|x)}.
		\end{align}
	with the conditional R\'enyi divergence given by $D
	\left(
	p_{G|X}
	||
	q_{G|X}
	|
	p_X
	\right)
	\coloneqq
	D
	\left(
	p_{G|X}
	p_X
	||
	q_{G|X}
	p_X
	\right)$, the latter being the standard KL-divergence \cite{KL1951, CT}.
\end{definition}

\subsection{Relationship between the R\'enyi divergence and CR divergences}

\begin{remark} (\cite{BLP1, thesis_CP})
    Relating conditional R\'enyi divergences to the R\'enyi divergence. For any conditional PMFs $p_{G|X}$, $q_{G|X}$, and any PMF $p_X$ we have:
    \begin{align}
	D^{\rm S}_\alpha
	(
	p_{G|X}
	||
	q_{G|X}
	|
	p_X
	)
    &
    =
    D_\alpha(p_{G|X}p_X||q_{G|X}p_X)
    ,
	\label{eq:SR}
	\\
	D^{\rm C}_\alpha
	(
	p_{G|X}
	||
	q_{G|X}
	|
	p_X
	)
	&
	=
	\sum_x
	p(x)
	D_\alpha(p_{G|X=x}||q_{G|X=x})
	.
	\label{eq:CR}
	\end{align}
\end{remark}

\subsection{Mutual informations}

\begin{definition} (Mutual informations of: Sibson \cite{sibson}, Csisz\'ar \cite{csiszar}, and Bleuler-Lapidoth-Pfister \cite{BLP1})
	The mutual information of Sibson, Csisz\'ar, and Bleuler-Lapidoth-Pfister of order $\alpha \in \mathds{\overline R}$ of a joint PMF $p_{XG}$ are defined as:
	\begin{align}
    I^{\rm V}_{\alpha}
		(X;G)
    \coloneqq
      \min_{q_G}
		D^{\rm V}_{\alpha}
		\left(
		p_{G|X}
		||
		q_G
		|
		p_X
		\right),
    \label{eq:DM}
    \end{align}
	with the label $\rm V\in\{S,C,BLP\}$ denoting each case, the minimisation being performed over all PMFs $q_G$, and $D^{\rm V}_\alpha(\cdot||\cdot|\cdot)$ the conditional R\'enyi (CR) divergences of: Sibson, Csisz\'ar, and Bleuler-Lapidoth-Pfister, of order $\alpha \in \mathds{\overline R}$, as defined previously. The case $\alpha=1$ reduces, for all three cases, to the standard mutual information \cite{CT} $I^{\rm V}_1(X;G) = I(X;G)$. Similarly to Arimoto's measure, we also use the notation $I^{\rm V}_\alpha(p_{XG})$ and $I_\alpha^{\rm V}(p_{G|X}p_X)$ interchangeably.
\end{definition}

\subsection{Relationship between CR-divergences}

\begin{lemma} 
	Consider the conditional-R\'enyi divergences of Sibson, Csisz\'ar, and Bleuler-Lapidoth-Pfister, then: 
	\begin{align}
	\alpha \in [-\infty,0],
	\hspace{0.1cm}
	D^{\rm BLP}_\alpha
	\left(
	\cdot
	\cdot
	\cdot
	\right)
	\leq
	D^{\rm C}_\alpha
	\left(
	\cdot
	\cdot
	\cdot
	\right)
	\leq
	D^{\rm S}_\alpha
	\left(
	\cdot
	\cdot
	\cdot
	\right),
	\\
	\alpha \in [0,1],
	\hspace{0.1cm}
	D^{\rm BLP}_\alpha
	\left(
	\cdot
	\cdot
	\cdot
	\right)
	\leq
	D^{\rm S}_\alpha
	\left(
	\cdot
	\cdot
	\cdot
	\right)
	\leq
	D^{\rm C}_\alpha
	\left(
	\cdot
	\cdot
	\cdot
	\right),
	\\
	\alpha \in [1,\infty],
	\hspace{0.1cm}
	D^{\rm C}_\alpha
	\left(
	\cdot
	\cdot
	\cdot
	\right)
	\leq
	D^{\rm BLP}_\alpha
	\left(
	\cdot
	\cdot
	\cdot
	\right)
	\leq
	D^{\rm S}_\alpha
	\left(
	\cdot
	\cdot
	\cdot
	\right),
	\end{align}
\end{lemma}

\begin{proof}
	The cases $\alpha \in [0,1]$ and $\alpha \in [1,\infty]$  have already been proven in the literature \cite{BLP1}. A similar argument can be followed in order to prove the cases $\alpha \in [-\infty,0]$. For completeness, we address it in what follows.
	\\
	\textbf{Part i)}
		We start by proving that for $\alpha \in [-\infty,0]$ we have 
		$
		D^{\rm C}_\alpha
		\left(
		\cdot
		||
		\cdot
		|
		\cdot
		\right)
		\leq
		D^{\rm S}_\alpha
		\left(
		\cdot
		||
		\cdot
		|
		\cdot
		\right)$. We prove it for $\alpha \in (-\infty,0)$ and the extremes follow because of continuity. Starting from the Sibson's measure times the positive factor $(\alpha-1)\sgn(\alpha)$ we get:
		\begin{align}
		(\alpha-1)
		\sgn(\alpha)
		D^{\rm S}_\alpha
		\left(
		p_{G|X}
		||
		q_G
		|
		p_X
		\right)
		&
		=
		\log
		\left[
		\sum_x
		p(x)
		\sum_g
		p(g|x)^\alpha
		q(g|x)^{1-\alpha}
		\right],
		\\
		&
		\geq
		\sum_x
		p(x)
		\log
		\left[
		\sum_g
		p(g|x)^\alpha
		q(g|x)^{1-\alpha}
		\right],
		\\
		&
		=
		\sgn(\alpha)
		(\alpha-1)
		D^{\rm C}_\alpha
		\left(
		p_{G|X}
		||
		q_G
		|
		p_X
		\right).
		\end{align}
		In the first equality we use the definition of the Sibson's conditional R\'enyi divergence \eqref{eq:SCRD}. The inequality follows because of Jensen's inequality \cite{jensen}, and because $\log(\cdot)$ is a concave function. In the last equality we use the definition of the Csisz\'ar's conditional R\'enyi divergence \eqref{eq:CCRD}. Dividing both sides by $\sgn(\alpha)(\alpha-1)$, which is positive because $\alpha \in (-\infty,0)$, proves the claim. 
		\\	
		\textbf{Part ii)} 
		We now want to prove that for $\alpha \in [-\infty,0]$, we have $
		D^{\rm BLP}_\alpha
		\left(
		\cdot
		||
		\cdot
		|
		\cdot
		\right)
		\leq
		D^{\rm C}_\alpha
		\left(
		\cdot
		||
		\cdot
		|
		\cdot
		\right)$. Similarly, we prove it for cases $\alpha \in (-\infty, 0)$ with the extremes following because of continuity. Starting from Csisz\'ar's measure:
		\begin{align}
D^{\rm C}_\alpha
		\left(
		p_{G|X}
		||
		q_G
		|
		p_X
		\right)
		&
		=
		\frac{
		\sgn(\alpha)
		}{\alpha-1}
		\sum_x
		p(x)
		\log
		\left[
		\sum_g
		p(g|x)^\alpha
		q(g|x)^{1-\alpha}
		\right]
		,
		\\
		&
		=
		\frac{
		|\alpha|
		}{\alpha-1}
		\sum_x
		p(x)
		\log
		\left[
		\sum_g
		p(g|x)^\alpha
		q(g|x)^{1-\alpha}
		\right]^{\frac{1}{\alpha}}
		,
		\\
		&
		\geq
		\frac{
		|\alpha|
		}{\alpha-1}
		\log
		\left[
		\sum_x
		p(x)
		\left(
		\sum_g
		p(g|x)^\alpha
		q(g|x)^{1-\alpha}
		\right)^{\frac{1}{\alpha}}
		\right]
		,
		\\
		&
		=
		D^{\rm BLP}_\alpha
		\left(
		p_{G|X}
		||
		q_G
		|
		p_X
		\right)
		.
		\end{align}
		The first equality we use the definition of Csisz\'ar's conditional R\'enyi divergence \eqref{eq:CCRD}. In the second equality we multiply by one $1=\frac{\alpha}{\alpha}$ and re-organise conveniently. The inequality follows because of Jensen's inequality \cite{jensen}, because $\log(\cdot)$ is a concave function, and because the coefficient $\frac{\sgn(\alpha)\alpha}{\alpha-1}$ is negative for $\alpha \in (-\infty,0)$. In the last equality we use the definition of the Bleuler-Lapidoth-Pfister conditional R\'enyi divergence \eqref{eq:BLPCRD}. 
\end{proof}

\subsection{Relationship between mutual informations}

\begin{lemma} 
	Consider the mutual informations of Sibson, Csisz\'ar, and Bleuler-Lapidoth-Pfister, then:
	\begin{align}
	\alpha \in [-\infty,0],
	\hspace{0.5cm}
	I^{\rm BLP}_\alpha
	\left(
	\cdot
	|
	\cdot
	\right)
	\leq
	I^{\rm C}_\alpha
	\left(
	\cdot
	|
	\cdot
	\right)
	\leq
	I^{\rm S}_\alpha
	\left(
	\cdot
	|
	\cdot
	\right),
	\\
	\alpha \in [0,1],
	\hspace{0.5cm}
	I^{\rm BLP}_\alpha
	\left(
	\cdot
	|
	\cdot
	\right)
	\leq
	I^{\rm S}_\alpha
	\left(
	\cdot
	|
	\cdot
	\right)
	\leq
	I^{\rm C}_\alpha
	\left(
	\cdot
	|
	\cdot
	\right),
	\\
	\alpha \in [1,\infty],
	\hspace{0.5cm}
	I^{\rm C}_\alpha
	\left(
	\cdot
	|
	\cdot
	\right)
	\leq
	I^{\rm BLP}_\alpha
	\left(
	\cdot
	|
	\cdot
	\right)
	\leq
	I^{\rm S}_\alpha
	\left(
	\cdot
	|
	\cdot
	\right),
	\end{align}
\end{lemma}
\begin{proof}
	The cases $\alpha \in [0,1]$ and $\alpha \in [1,\infty]$ were proven in \cite{BLP1}, and they follow by considering the previous Lemma on the less or equal order between the conditional R\'enyi divergences and, by considering that the mutual informations are defined in terms of the conditional R\'enyi divergences by minimising over $p_X$ \eqref{eq:DM}. The cases $\alpha \in [-\infty,0]$ follow the same argument.
\end{proof}

\subsection{The Rényi channel capacity}

Having defined these mutual informations, we now address the fact that some of them become equal when maximising over PMFs $p_X$, whilst keeping fixed the conditional PMF $p_{G|X}$.
\begin{lemma} (R\'enyi channel capacity \cite{arimoto, csiszar, remarks}) The mutual information of Arimoto and Sibson of order $\alpha \in \mathds{\overline R}$) become equal when maximised over $p_X$, and we refer to this quantity as the R\'enyi capacity of order $\alpha$. The R\'enyi capacity of order $\alpha \in \mathds{\overline R}$,  of a conditional PMF $p_{G|X}$ is:
\begin{align}
    C_{\alpha}
		(p_{G|X})
    \coloneqq
      \max_{p_X}
		I^{\rm V}_{\alpha}
		(p_{G|X}p_X),
\end{align}
	with ${\rm V\in\{A,S\}}$, the maximisation over all PMFs $p_X$, and the mutual information of Sibson as in \eqref{eq:DM}, and Arimoto's mutual information as in the main text. The case $\alpha=1$ reduces to the standard channel capacity \cite{CT} $C_1(p_{G|X})=C(p_{G|X})=\max_{p_X}I(X;G)$.
\end{lemma}

This Lemma, for the cases $\alpha \geq 0$, has been proven in different places in the literature \cite{sibson, csiszar, remarks}. For completeness, here we provide a proof for the cases $\alpha < 0$. We can understand this result as $C_\alpha(p_{G|X})$ being the R\'enyi capacity of the classical channel specified by the conditional PMF $p_{G|X}$, which simultaneously represents the mutual information of Arimoto and Sibson. On can similarly address R\'enyi capacities using the rest of mutual informations, but using these two are enough for our purposes.

\begin{proof}
	The cases for $\alpha \in [0, \infty)$ have been proven in different places in the literature \cite{arimoto, csiszar, remarks}. We therefore only address here the interval $(-\infty,0)$. Addressing Arimoto's measure for $\alpha \in (-\infty,0)$:
	\begin{align}
		\max_{p_X}
		I^{\rm A}_{\alpha}
		\left(
		p_{G|X}
		p_X
		\right)
		&
		\overset{1}{=}
		\max_{p_X}
		\frac{
		|\alpha|
		}{\alpha-1}
		\log
		\sum_g
		\left(
		\sum_x
		p(g|x)^\alpha
		\frac{
		p(x)^\alpha
		}{
		\sum_{x'}p(x')^\alpha
		}
		\right)^\frac{1}{\alpha}
		,\\
		&
		\overset{2}{=}
		\max_{r_X}
		\frac{
		|\alpha|
		}{
		\alpha-1
		}
		\log
		\sum_g
		\left(
		\sum_x
		p(g|x)^\alpha
		r(x)
		\right)^\frac{1}{\alpha}
		.
		\label{eq:AP1}
	\end{align}
	In the first equality we replaced and reorganised the definition of Arimoto's mutual information \eqref{eq:DMA}. In the second equality we use the fact that both maximisations are equal, because from an optimal $p^*_X$, we can construct a feasible $r_X$ as $
	r(x)
	\coloneqq 
	p^*(x)^\alpha
	/
    (\sum_{x'}p^*(x')^\alpha)
	$ and conversely, from an optimal $r^*_X$, we can construct a feasible $p_X$ as $p(x) = r^*(x)^\frac{1}{\alpha}/(\sum_{x'}r^*(x')^\frac{1}{\alpha})$. 
	We now relate the quantity in \eqref{eq:AP1} to the quantity obtained from Sibson's. We now consider Sibson's CR-divergence and invoke the identity \cite{csiszar} $\forall p_{G|X}, q_G,	p_X$:
	\begin{align}
	    D^{\rm S}_{\alpha}
		\left(
		p_{G|X}
		||
		q_G
		|
		p_X
		\right)
		=
		D^{\rm S}_{\alpha}
		\left(
		p_{G|X}
		||
		q_G^*
		|
		p_X
		\right)
		+
		D_{\alpha}
		\left(
		q_G^*
		|
		q_G
		\right),
		\label{eq:identity}
	\end{align}
	with the PMF $q_G^*$ given by:
	\begin{align}
	    q_G^*
	    (g)
	    \coloneqq
	    \frac{
    	    \left(
    	    \sum_x
    	    p(x)
    	    p(g|x)^\alpha
    	    \right)^\frac{1}{\alpha}
        }
        {
        \sum_g
            \left(
            \sum_x
    	    p(x)
    	    p(g|x)^\alpha
    	    \right)^\frac{1}{\alpha}
        }
        .
        \label{eq:qs}
	\end{align}
	This identity can be checked by directly substituting \eqref{eq:qs} into the RHS of \eqref{eq:identity}. We can now get an explicit expression for Sibson's mutual information, because minimising \eqref{eq:identity} over $q_G$ is obtained for $q_G=q_G^*$, this, because the R\'enyi divergence is non-negative for $\alpha \in (-\infty,0)$ \cite{RD}. We therefore get:
	\begin{align}
	    I^{\rm S}_{\alpha}
		\left(
		p_{G|X}
		p_X
		\right)
		&=
		\min_{q_G}
		D^{\rm S}_{\alpha}
		\left(
		p_{G|X}
		||
		q_G
		|
		p_X
		\right)
		,
		\\
		&=
		D^{\rm S}_{\alpha}
		\left(
		p_{G|X}
		||
		q_G^*
		|
		p_X
		\right)
		,
		\\
		&=
		\frac{
		|\alpha|
		}{
		\alpha-1
		}
		\log
		\sum_g
		\left(
		    \sum_x
		    p(x)
		    p(g|x)^\alpha
		\right)^\frac{1}{\alpha}
		.
	\end{align}
	Maximising this quantity over $p_X$ we get:
	\begin{equation}
	    \max_{p_X}
	    I^{\rm S}_{\alpha}
		\left(
		p_{G|X}
		p_X
		\right)
		=
		\max_{p_X}
		\frac{
		|\alpha|
		}{\alpha-1}
		\log
		\sum_g
		\left(
		    \sum_x
		    p(x)
		    p(g|x)^\alpha
		\right)^\frac{1}{\alpha}
		,
		\label{eq:S1}
	\end{equation}
	which is the same quantity than in \eqref{eq:AP1} for Arimoto's measure. Altogether, we have that starting from either Sibson or Arimoto, we arrive to the same expression when maximising over $p_X$, as per equations \eqref{eq:S1} and \eqref{eq:AP1}.
	Consequently, the capacities they each define is the same, and thus proving the claim. 
\end{proof}

\subsection{Information-theoretic measures in the quantum domain}

Mutual informations in the quantum domain are defined via their R\'enyi conditional divergences counterparts as:
\begin{align}
    I_{\alpha}^{\rm V}
		(X;G)_{\mathcal{E},\mathbb{M}}
    \coloneqq
      \min_{q_G}
		D^{\rm V}_{\alpha}
		\left(
		p_{G|X}^{({\mathbb{M}},\mathcal{S})}
		\Big|\Big|
		q_{G}
		\Big|
		\,
		p_X
		\right),
\end{align}
	with the quantum conditional PMFs $p_{G|X}^{({\mathbb{M}},\mathcal{E})}$ and $q_{G|X}^{({\mathbb{N}}, \mathcal{E})}$ given by $p(g|x)\coloneqq \tr (M_g \rho_x)$, $q(g|x)\coloneqq \tr (N_g \rho_x)$, respectively, the minimisation over all PMFs $q_G$, and the classical conditional R\'enyi divergences of: Sibson, Csisz\'ar, and Bleuler-Lapidoth-Pfister, which we address with  a label $\rm V\in\{S,C,BLP\}$.

\section{Proof of \cref{R_uninformative}}
\label{AR_uninformative}

We start by mentioning that the tasks which are of interest to us are quantum state betting (QSB) games, but that from an operational point of view, they are equivalent to ``horse betting games with risk and quantum side information", or quantum horse betting (QHB) games for short. Given this equivalence, in this appendix we would address QSB games as QHB or HB games only.

In order to prove Result 1 we need two Theorems on horse betting (HB) with risk: one for HB games \emph{without} side information, and other for HB games \emph{with} side information. These two Theorems depend on the of R\'enyi divergence and the BLP conditional R\'enyi divergence from \cref{AA}.

\subsection{Preliminary steps}

We start by addressing a simplified notation.
\begin{align}
	w^{ICE}_R
	(b_{X|G},o_X,p_{XG}
	)
	\coloneqq
	w^{ICE}_R
	(b_{X|G}, \mathbb{M}, o_X,\mathcal{E})
	,
\end{align}	
with $p(x,g)=p(g|x)p(x)$, $p(g|x)=\tr[M_g\rho_x]$. We also notice that that optimising over uninformative measurements $\mathds{N}\in {\rm UI}$, meaning $N_g=p(g)\mathds{1}$, $\forall g$, is equivalent to a horse betting game with risk but \emph{without} side information because $p(g|x)=\tr(N_g\rho_x)=p(g)\tr(\mathds{1}\rho_x)=p(g)$ and then:
\begin{align}
    \max_{b_{X|G}}
    \max_{\mathds{N}\in {\rm UI}}
	w^{ICE}_R
	(b_{X|G},o_X,p_{XG}
	)
	&=
	\max_{b_{X|G}}
	\max_{\mathds{N}\in {\rm UI}}
	\left[
		\sum_{g,x}
		\big[
		b(x|g)
		o(x)\big]^{1-R}
		p(g|x)
		p(x)
		\right]^\frac{1}{1-R}
		,
	\\
	&=
	\max_{b_{X|G}}
	\max_{p_G}
	\left[
		\sum_{g,x}
		\big[
		b(x|g)
		o(x)\big]^{1-R}
		p(g)
		p(x)
		\right]^\frac{1}{1-R}
		,
	\\
	&=
	\max_{b_{X|G}}
	\left[
		\sum_{x}
		\left(
		\max_{p_G}
		\sum_{g}
		b(x|g)^{1-R}
		p(g)
		\right)
		o(x)^{1-R}
		p(x)
		\right]^\frac{1}{1-R}
		,
		\\
	&=
	\max_{b_{X|G}}
	\left[
		\sum_{x}
		\left(
		\max_{g}
		b(x|g)^{1-R}
		\right)
		o(x)^{1-R}
		p(x)
		\right]^\frac{1}{1-R}
		,
	\\
	&=
	\max_{b_{X}}
	\left[
		\sum_{x}
		\big[
		b(x)
		o(x)\big]^{1-R}
		p(x)
		\right]^\frac{1}{1-R}
		,
	\\
	&=
	\max_{b_{X}}
	w^{ICE}_R
	(b_X,o_X,p_X)
	.
\end{align}	
This defines a HB game \emph{without} side information, meaning without the random variable $G$. We now define the auxiliary function of \emph{the logarithm of the isoelastic certainty equivalent} as:
\begin{equation}
		U_R
		(b_{X|G},o_X,p_{XG})
		\coloneqq
		\sgn(
		o
		)
		\log
		\left|
		w^{ICE}_R
		(b_{X|G},o_X,p_{XG})
		\right|
		,
\end{equation}
and similarly without side information as:
\begin{align}
		U_R
		(b_X,o_X,p_X)
		\coloneqq
		\sgn(
		o
		)
		\log
		\left|
		w^{ICE}_R
		(b_X,o_X,p_X)
		\right|
		,
		\label{eq:QSB_SI_U}
\end{align}
with $\sgn(o)$ as a shorthand for the sign of the odds $o(x)$, $\forall x$. We also highlight here that we are interested in the strategy that achieves:
\begin{align}
    \max_{b_X}
    w^{ICE}_R
	(b_X,o_X,p_X),
\end{align}
and we can see that this is equivalent to finding the best strategy for the auxiliary optimisation:
\begin{align}
    \max_{b_X}
    U_R
	(b_X,o_X,p_X).
\end{align}

\subsection{Horse betting games with risk}

We now present two results on horse betting games. We remark here that we invoke these results, in contrast with the original presentation in \cite{BLP1}, with the following modifications in the notation: i) the original version involves a parameter $\beta$, here instead we directly use the risk aversion parameter $R$, taking into account that these two parameters are related as $\beta=1-R$, ii) we have defined the R\'enyi divergence as a non-negative quantity, for all $\alpha \in \mathds{\overline R}$, even for negative values of alpha, and this explains the appearance of the term $\sgn(R)$, iii) we allow for the odds and consequently the wealth to be negative, and this explains the appearance of the term $\sgn(o)$. We now address a result that characterises this task in terms of the R-divergence. 
\begin{theorem}(Bleuler-Lapidoth-Pfister \cite{BLP1, thesis_CP}) 
	Consider a HB game with risk defined by the triple $(o_X, p_X,R)$, and a Gambler playing this game with a betting strategy $b_X$. The logarithm of the isoelastic certainty equivalent is characterised by the R-divergence $D_\alpha(\cdot||\cdot)$ as:
	\begin{equation}
    U_R
    	(b_X,o_X,p_X)
    	=
    	\sgn(o)
    	\log 
    	\left| 
    	c^o
    	\right|
    	+
    	\sgn(o)
    	\sgn(R)\,
    	D_{1/R}
    	(p_X||r^o_X)
    	-
    	\sgn(o)
    	\sgn(R)\,
    	D_{R}
    	(h^{(R,o,p)}_X||b_X)
    	,
	\end{equation}
	with the parameter and the PMF:
	\begin{align}
    	c^o
    	\coloneqq
    	\left(
    	\sum_x \frac{1}{o(x)}
    	\right)^{-1},
    	\hspace{0.5cm}
    	r^o(x)
    	\coloneqq
    	\frac{c^o}{o(x)},
	\end{align}
	and the PMF:
	\begin{align}
	h^{(R,o,p)}(x)
	\coloneqq
	\frac{
		p(x)^{\frac{1}{R}}
		o(x)^{\frac{1-R}{R}}
	}
	{
		\sum_{x'}
		p(x')^{\frac{1}{R}}
		o(x')^{\frac{1-R}{R}}
	}.
	\label{eq:HB_b}
	\end{align}
	Note that the quantities $r^o_X$ and $h^{(R,o,p)}_X$ define valid PMFs even for negative odds ($o(x)<0$, $\forall x$).
\end{theorem}
We are particularly interested in the best possible betting strategy for a given game $(o_X, p_X)$ and fixed $R$, so we have the following two corollaries.
\begin{corollary} (Bleuler-Lapidoth-Pfister \cite{BLP1, thesis_CP}) 
	Consider a classical horse discrimination (HD) game ($o_X^{+}$, meaning $\sgn(o)=1$) being played by a risk-averse Gambler ($R\geq0$, meaning $\sgn(R)=1$). We then want to maximise the logarithm of the isoelastic certainty equivalent over all possible betting strategies. The gambler plays optimally when choosing $b^*(x)=h^{(R,o,p)}(x)$ and then:
	\begin{align}
	\max_{b_X}
	U_R
	(b_X,o_X^{+},p_X)
	&=
	U_R
	(b_X^*,o_X^{+},p_X),
	\nonumber
	\\
	&=
	\log
	\left| 
    	c^o
	\right|
	+
	D_{1/R}(p_X||r^o_X).
	\end{align}
	This is because the R\'enyi divergence $D_{R}(\cdot||\cdot)$ is non-negative $\forall R\in \mathds{\overline R}$.
\end{corollary}
\begin{corollary}
	Consider a classical horse exclusion (HE) game ($o_X^{-}$, meaning $\sgn(o)=-1$)	being played by a risk-averse Gambler ($R<0$, meaning $\sgn(R)=-1$). We then want to maximise the logarithm of the isoelastic certainty equivalent over all possible betting strategies. The gambler plays optimally when choosing $b^*(x)=h^{(R,o,p)}(x)$ and then:
	\begin{align}
	\max_{b_X}
	U_R
	(b_X,o_X^{-},p_X)
	&=
	U_R
	(b_X^*,o_X^{-},p_X),
	\nonumber
	\\
	&=
	-\log 
	\left| 
    	c^o
	\right|
	+
	D_{1/R}(p_X||r^o_X).
	\end{align}
	This is because the R\'enyi divergence $D_{R}(\cdot||\cdot)$ is non-negative $\forall R\in \mathds{\overline R}$.
\end{corollary}

\subsection{Horse betting with risk and side information}

We now address a result that characterises this task in terms of the BLP-CR-divergence and the R-divergence. 

\begin{theorem}(Bleuler-Lapidoth-Pfister \cite{BLP1, thesis_CP})
	Consider a HB game with risk and side information defined by the triple $(o_X,p_{XG},R)$, and a Gambler playing this game with a betting strategy $b_{X|G}$. The utility function of log-wealth is characterised by the the BLP-CR-divergence $D^{\rm BLP}_\alpha(\cdot||\cdot|\cdot)$ and R-divergence $D_\alpha(\cdot||\cdot)$ as:
    \begin{multline}
		U_R
		(b_{X|G},o_X,p_{XG})
		\\=
		\sgn(o)
		\log 
		\left|
		c^o
		\right|
		+
		\sgn(o)
		\sgn(R)
		D^{\rm BLP}_{1/R}
		\left(
		p_{X|G}||r_X^o|p_G
		\right)
		-
		\sgn(o)
		\sgn(R)
		D_{R}
		\left(
		h^{(R,o,p)}_{X|G} 
		h_G^{(R,o,p)}
		\Big| \Big|
		b_{X|G}\,
		h_G^{(R,o,p)}
		\right),
	\end{multline}
	with the parameter and the PMF:
	\begin{align}
	c^o
	\coloneqq
	\left(
	\sum_x \frac{1}{o(x)}
	\right)^{-1},
	\hspace{0.5cm}
	r^o(x)
	\coloneqq
	\frac{c^o}{o(x)},
	\end{align}
	and the conditional PMF and PMF:
	\begin{align}
		h^{(R,o,p)}(x|g)
		&\coloneqq \frac{
			p(x|g)^{\frac{1}{R}}
			o(x)^{\frac{1-R}{R}}
		}
		{
			\sum_{x'}
			p(x'|g)^{\frac{1}{R}}
			o(x')^{\frac{1-R}{R}}
		},
		\label{eq:HB_SI_b}
		\\
		h^{(R,o,p)}(g)
		&\coloneqq \frac{
		    p(g)
		    \left[
		    \sum_{x'}
			p(x'|g)^{\frac{1}{R}}
			o(x')^{\frac{1-R}{R}}
			\right]^{R}
		}
		{
			\sum_{g'}
			p(g')
		    \left[
		    \sum_{x'}
			p(x'|g')^{\frac{1}{R}}
			o(x')^{\frac{1-R}{R}}
			\right]^{R}
		}.
		\end{align}
		Note that the quantities $r^o_X$, $h^{(R,o,p)}_{X|G}$, $h^{(R,o,p)}_G$ define valid PMFs even for negative odds ($o(x)<0$, $\forall x$).
\end{theorem}
We are particularly interested in the best possible betting strategy $b_{X|G}$ for a given game $(o_X,p_{XG})$ and fixed $R$, so we have the following two corollaries.
\begin{corollary}(Bleuler-Lapidoth-Pfister \cite{BLP1, thesis_CP})
	Consider a horse discrimination (HD) game ($o_X^{+}$, meaning $\sgn(o)=1$) being played by a risk-averse Gambler ($R>0$, meaning $\sgn(R)=1$) with access to side information. We then want to maximise the logarithm of the isoelastic certainty equivalent over all possible betting strategies. The Gambler plays optimally when choosing $b^*(x|g) = h^{(R,o,p)}(x|g)$ and then:
	\begin{align}  
	\max_{b_{X|G}}
	U_R
	(b_{X|G},o_X^{+},p_{XG})
	&=
	U_R 
	(b^*_{X|G},o_X^{+},p_{XG}),
	\nonumber
	\\
	&=
	\log 
	\left|
	c^o
	\right|
	+
	D^{\rm BLP}_{1/R}
	(p_{X|G}
	||
	r^o_X 
	|
	p_G
	),
	\end{align}
	with the BLP-CR-divergence $D^{\rm BLP}_\alpha(\cdot||\cdot|\cdot)$. This is because the R\'enyi divergence $D_{R}(\cdot||\cdot)$ is non-negative $\forall R\in \mathds{\overline R}$.
\end{corollary}
\begin{corollary}
	Consider a classical horse exclusion (HE) game ($o_X^{-}$)	being played by a risk-averse Gambler ($R<0$) with access to side information. We then want to maximise the logarithm of the isoelastic certainty equivalent over all possible betting strategies. The Gambler plays optimally when choosing $b^*(x|g)=h^{(R,o,p)}(x|g)$ and then:
	\begin{align} 
	\max_{b_{X|G}}
	U_R
	(b_{X|G},o_X^{-},p_{XG})
	&=
	U_R
	(b^*_{X|G},o_X^{-},p_{XG}),
	\nonumber
	\\
	&=
	-
	\log 
	\left|
	c^o
	\right|
	+
	D^{\rm BLP}_{1/R}
	(p_{X|G}
	||
	r^o_X 
	|
	p_G
	),
	\end{align}
	with the BLP-CR-divergence $D^{\rm BLP}_\alpha(\cdot||\cdot|\cdot)$. This is because the R\'enyi divergence $D_{R}(\cdot||\cdot)$ is non-negative $\forall R\in \mathds{\overline R}$.
\end{corollary}

\subsection{Proving Result 1}

In order to prove Result 1 we need two Lemmas. Let us start by rewriting the R\'enyi entropy in a more convenient form:
\begin{align}
		H_{\alpha}(X)
		& 
		=
		- 
		\log
		\left[
		p_{\alpha}
		\left(
		X 
		\right)
		\right],
		\label{eq:RE appendix}\\
		p_{\alpha}
		\left(
		X 
		\right)
		& \coloneqq
		\left(
		\sum_x p(x)^\alpha
		\right)^{\frac{1}{(\alpha-1)}}	.
\end{align}
We are now ready to establish a first Lemma.
\begin{lemma}
	(Operational interpretation of the R\'enyi entropy) \label{lemma4} Consider a PMF $p_X$, the R\'enyi probability of order $\alpha \in \mathds{\overline R}$ can be written as:
	\begin{align}
	\sgn(\alpha)
	\,
    C
    \,
	p_\alpha(X)
	=
    \max_{b_X}
    w^{ICE}_{1/\alpha}
    (
    b_{X}
    ,
    o_X^{\sgn(\alpha)c}
    ,
    p_X)
    ,
    \end{align}
	with the maximisation over all possible betting strategies $b_X$, and constant odds $o^{\sgn(\alpha)c}(x) \coloneqq \sgn(\alpha)C$, $C>0$, $\forall x$.
\end{lemma}
\begin{proof}
	We start by considering a HB game with constant odds $o^{\sgn(\alpha)}(x) \coloneqq \sgn(\alpha)C$, $C>0$, $\forall x$, and consider a risk-aversion coefficient parametrised as $R(\alpha) \coloneqq 1/\alpha$. We first notice that the best strategy for the Gambler is given by \eqref{eq:HB_b}:
	\begin{align}
		b^*(x)
		=
		\frac{
			p(x)^\alpha
		}
		{
			\sum_{x'}
			p(x')^\alpha
		}.
	\end{align}
	Considering now the isoelastic certainty equivalent and replacing the constant odds and the best strategy we get:
	\begin{align}
		w^{ICE}_{1/\alpha}
		(b^*_{X},o_X^{\sgn(\alpha)c},p_X)
		&=
		\left[
		\sum_{x}
		p(x)
		\big[
		b^*(x)
		o^{\sgn(\alpha)c}(x)
		\big]^{\frac{\alpha-1}{\alpha}}
		\right]^\frac{\alpha}{\alpha-1}
		,\\
		&=
		\sgn(\alpha)
		C
		\left[
		\sum_{x}
		p(x)
		\big[
		b^*(x)
		\big]^{\frac{\alpha-1}{\alpha}}
		\right]^\frac{\alpha}{\alpha-1}
		,\\
		&=
		\sgn(\alpha)
		C
		\left[
		\sum_{x}
		p(x)
		\left[
		\frac{
			p(x)^\alpha
		}
		{
			\sum_{x'}
			p(x')^\alpha
		}
		\right]^{\frac{\alpha-1}{\alpha}}
		\right]^\frac{\alpha}{\alpha-1}
		,\\
		&=
		\sgn(\alpha)
		C
		\left[
		\sum_{x}
		p(x)
		\frac{
			p(x)^{\alpha-1}
		}
		{
			\left[ \sum_{x'}
			p(x')^\alpha \right]^\frac{\alpha-1}{{\alpha}}
		}
		\right]^\frac{\alpha}{\alpha-1}
		.
		\end{align}
		Reorganising we get:
		\begin{align}
		w^{ICE}_{1/\alpha}
		(b^*_{X},o_X^{\sgn(\alpha)c},p_X)
		&=
		\sgn(\alpha)
		C
		\left[
		\sum_{x}
		\frac{
			p(x)^{\alpha}
		}
		{
			\left[ \sum_{x'}
			p(x'|g)^\alpha \right]^\frac{\alpha-1}{{\alpha}}
		}
		\right]^\frac{\alpha}{\alpha-1}
		,\\
		&=
		\sgn(\alpha)
		C
		\frac{
			1
		}
		{
			\sum_{x'}
			p(x')^\alpha
		}
		\left[
		\sum_{x}
		p(x)^{\alpha}
		\right]^\frac{\alpha}{\alpha-1}
		,\\
		&=
		\sgn(\alpha)
		C
		\left[
		\sum_{x}
		p(x)^{\alpha}
		\right]^\frac{1}{\alpha-1}
		,\\
		&=
		\sgn(\alpha)
		C 
		p_\alpha(X)
		,
		\end{align}
	and therefore proving the claim.
\end{proof}
We now move on to rewrite the Arimoto-R\'enyi conditional entropy in a more convenient form:
\begin{align}
	H_{\alpha}(X|G)
	& =
	- 
	\log
	\left[
	p_{\alpha}
	\left(
	X|G
	\right)
	\right],
	\label{eq:ARCE appendix}\\
	p_{\alpha}
	\left(
	X|G
	\right)
	&\coloneqq
	\left(
	\sum_g
	\left(
	\sum_x
	p(x,g)^\alpha
	\right)^\frac{1}{\alpha}	
	\right)^\frac{\alpha}{(\alpha-1)}
	.
\end{align}
We are now ready to establish a second Lemma.
\begin{lemma}(Operational interpretation of the Arimoto-R\'enyi conditional entropy) \label{lemma5}
	Consider a joint PMF $p_{XG}$, the Arimoto-R\'enyi conditional entropy of order $\alpha \in \mathds{\overline R}$ can be written as:
	\begin{align}
	\sgn(\alpha)
	\,
	C
	\,
	p_\alpha(X|G)
	=
    \max_{b_{X|G}}
    w^{ICE}_{1/\alpha}
    (
    b_{X|G}
    ,
    o_X^{\sgn(\alpha)c}
    ,
    p_{XG})
    ,
    \end{align}
	with the maximisation over all possible betting strategies $b_{X|G}$, and constant odds $o^{\sgn(\alpha)c}(x) \coloneqq \sgn(\alpha)C$, $C>0$, $\forall x$.
\end{lemma}
\begin{proof}
	We start by considering a HB game with constant odds $o^{\sgn(\alpha)}(x) \coloneqq \sgn(\alpha)C$, $C>0$, $\forall x$, and consider a risk-aversion coefficient parametrised as $R(\alpha) \coloneqq 1/\alpha$. We now notice that the best strategy for the Gambler with access to side information is given by \eqref{eq:HB_SI_b}:
	\begin{align}
	b^*(x|g)
	&=
	g^{(R,o,p)}(x|g)
	,\\
	& = \frac{
		p(x|g)^{\frac{1}{R}}
		o^{\sgn(\alpha)c}(x)^{\frac{1-R}{R}}
	}
	{
		\sum_{x'}
		p(x'|g)^{\frac{1}{R}}
		o^{\sgn(\alpha)c}(x')^{\frac{1-R}{R}}
	}
	,\\
	& = \frac{
		p(x|g)^{\frac{1}{R}}
		(\sgn(\alpha)C)^{\frac{1-R}{R}}
	}
	{
		\sum_{x'}
		p(x'|g)^{\frac{1}{R}}
		(\sgn(\alpha)C)^{\frac{1-R}{R}}
	}
	,\\
	& = \frac{
		p(x|g)^{\frac{1}{R}}
	}
	{
		\sum_{x'}
		p(x'|g)^{\frac{1}{R}}
	}
	,\\
	& = \frac{
		p(x|g)^{\alpha}
	}
	{
		\sum_{x'}
		p(x'|g)^{\alpha}
	}
	.
	\end{align}
	Considering now the isoelastic certainty equivalent and replacing the constant odds and the best strategy we get:
	\begin{align}
    w^{ICE}_{1/\alpha}
	(b^*_{X|G},o_X^{\sgn(\alpha)c},p_{XG})
	&=
	\left[
	\sum_{x,g}
	p(x,g)
	\big[
	b^*(x|g)
	o^{\sgn(\alpha)c}(x)\big]^{\frac{\alpha-1}{\alpha}}
	\right]^\frac{\alpha}{\alpha-1}
	,\\
		&=
		\sgn(\alpha)
		C
		\left[
		\sum_{x,g}
		p(x,g)
		\left[
		\frac{
			p(x|g)^\alpha
		}
		{
			\sum_{x'}
			p(x'|g)^\alpha
		}
		\right]^{\frac{\alpha-1}{\alpha}}
		\right]^\frac{\alpha}{\alpha-1}
		,\\
		&=
		\sgn(\alpha)
		C
		\left[
		\sum_{x,g}
		p(x,g)
		\frac{
			p(x|g)^{\alpha-1}
		}
		{
			\left[ \sum_{x'}
			p(x'|g)^\alpha \right]^\frac{\alpha-1}{{\alpha}}
		}
		\right]^\frac{\alpha}{\alpha-1}
		.
\end{align}
Using $p(x,g)=p(x|g)p(g)$ and reorganising:
\begin{align}
	w^{ICE}_{1/\alpha}
	(b^*_{X|G},o_X^{\sgn(\alpha)c},p_{XG})	&=
		\sgn(\alpha)
		C
		\left[
		\sum_{x,g}
		p(g)
		\frac{
			p(x|g)^{\alpha}
		}
		{
			\left[ \sum_{x'}
			p(x'|g)^\alpha \right]^\frac{\alpha-1}{{\alpha}}
		}
		\right]^\frac{\alpha}{\alpha-1}
		,\\
		&=
		\sgn(\alpha)
		C
		\left[
		\sum_{g}
		p(g)
		\frac{
			\sum_x p(x|g)^{\alpha}
		}
		{
			\left[ \sum_{x'}
			p(x'|g)^\alpha \right]^\frac{\alpha-1}{{\alpha}}
		}
		\right]^\frac{\alpha}{\alpha-1}
		,\\
		&=
		\sgn(\alpha)
		C
		\left[
		\sum_{g}
		p(g)
		\left[
		\sum_x p(x|g)^{\alpha}
		\right]^{\frac{1}{\alpha}}
		\right]^\frac{\alpha}{\alpha-1}
		,\\
		&=
		\sgn(\alpha)
		C
		\left[
		\sum_{g}
		\left[
		\sum_x 
		p(x|g)^{\alpha}
		p(g)^\alpha
		\right]^{\frac{1}{\alpha}}
		\right]^\frac{\alpha}{\alpha-1}
		,\\
		&=
		\sgn(\alpha)
		C
		\left[
		\sum_{g}
		\left[
		\sum_x
		p(x,g)^{\alpha}
		\right]^{\frac{1}{\alpha}}
		\right]^\frac{\alpha}{\alpha-1}
		,\\
		&=
		\sgn(\alpha)
		C
		p_\alpha(X|G),
		\end{align}
	and therefore proving the claim.
\end{proof}
We are now ready to prove Result 1.
\begin{proof}(of Result 1)
	Consider the Arimoto's mutual information of order $\alpha \in \mathds{\overline R}$, we have the following chain of equalities:
	\begin{align}
	I_\alpha(X;G)
	&=
	\sgn(\alpha)
	[
	H_{\alpha}(X)
	-
	H_{\alpha}(X|G)
	]
	,\\
	&=
	\sgn(\alpha)
	\log
	\left[
	\frac{
		p_{\alpha}
		(
		X|G
		)
	}
	{
		p_{\alpha}	
		(
		X 
		)
	}
	\right]
	,\\
	&=
	\sgn(\alpha)
	\log
	\left[
	\frac{
		\sgn(\alpha)\,C\,p_{\alpha}
		(
		X|G
		)
	}
	{
		\sgn(\alpha)
		\,C\,p_{\alpha}	
		(
		X 
		)
	}
	\right]
	,\\
	&=
	\sgn(\alpha)
	\log
	\left[
	\frac{
		\displaystyle
		\max_{b_{X|G}}
		w^{ICE}_{1/\alpha}
		(
		b_{X|G}
		,
		o_X^{\sgn(\alpha)}
		,
		p_{XG}
		)
	}{
		\displaystyle
		\max_{b_X}
		w^{ICE}_{1/\alpha}
		(
		b_X
		,
		o_X^{\sgn(\alpha)}
		,
		p_X
		)
	}
	\right]
	.
	\end{align}
	The first equality is the definition of the Arimoto's mutual information. The second equality comes from replacing the R\'enyi entropy and the Arimoto-R\'enyi conditional entropy. The third inequality we have multiplied and divided by $\sgn(\alpha)\,C$. The fourth and last equality follows from invoking \cref{lemma4} and \cref{lemma5}. This proves the claim.
\end{proof}

\section{Proof of Corollaries 2 and 3}
\label{Acorollaries}

\begin{proof}(of Corollary 2)
	In the case $\alpha\rightarrow\infty$, we have:
	\begin{align}
		\max_\mathcal{E}
		I_{\infty}
		(X;G)_{\mathcal{E},\mathbb{M}}
		&=
		\log
		\left[
		\max_\mathcal{E}
		\frac{
			\displaystyle
			\max_{b_{X|G}}
			\,
			w^{ICE}_0
			(
			b_{X|G},\mathbb{M},o_X^{c},\mathcal{E}
			)
		}{
			\displaystyle
			\max_{\mathbb{N}\in {\rm UI}}
			\max_{b_{X|G}}
			\,
			w^{ICE}_0
			(
			b_{X|G}, \mathbb{N},o_X^{c}, \mathcal{E}
			)
		}
		\right].
	\end{align}
	To prove the claim, it is enough to prove:
	\begin{align}
	\displaystyle
	\max_{b_{X|G}}
	\,
	w^{ICE}_0
	(
	b_{X|G},\mathbb{M},o_X^{c},\mathcal{E}
	)
	=
	C
	P^{\rm QSD}_{\rm succ}(\mathcal{E},\mathbb{M}).
	\end{align}
	We have already shown this in the main document, but we can also double check it from \cref{lemma5} from which we have that for $\alpha \geq 0$:
	\begin{align}
	\max_{b_{X|G}}
	w^{ICE}_{1/\alpha}
	(b_{X|G},\mathbb{M},o_X^{c},\mathcal{E})
	&=
	C
	p_\alpha(X|G),
	\\
	&=
	C
	\left[
	\sum_{g}
	\left[
	\sum_x
	p(x,g)^{\alpha}
	\right]^{\frac{1}{\alpha}}
	\right]^\frac{\alpha}{\alpha-1}.
	\end{align}
	Considering now $\alpha \rightarrow \infty$ we have:
	\begin{align}
	&
	\max_{b_{X|G}}
	w^{ICE}_{0}
	(b_{X|G},\mathbb{M},o_X^{c},\mathcal{E})
	=
	C
	\sum_g
	\max_x
	p(x,g)
	.
	\end{align}
	Further analysing this quantity we have:
	\begin{align}
	\sum_g
	\max_x
	p(x,g)
	&=
	\sum_g
	\max_{q(x|g)}
	\sum_x
	q(x|g)
	p(x,g),
	\\
	&=
	\max_{q(x|g)}
	\sum_{g,x}
	q(x|g)
	p(g|x)p(x),
	\\
	&=
	\max_{q(x|g)}
	\sum_{g,x}
	\left[
	\sum_a \delta^a_x \,
	q(a|g)
	\right]
	p(g|x)
	p(x),
	\\
	&=
	\max_{q(a|g)}
	\sum_{a,g,x}
	\delta^a_x\,
	q(a|g)
	p(g|x)
	p(x),
	\\
	&=
	P^{\rm QSD}_{\rm succ}(\mathcal{E},\mathbb{M}).
	\end{align}
	In the first line we use the identity:
	\begin{equation}
		\max_{q(x)}
		\sum_x
		q(x)
		f(x)
		=
		\max_x f(x).
		\label{eq:PS_trick}
	\end{equation}
	This proves the claim.
\end{proof}
\begin{proof}(of Corollary 3)
The proof of Corollary 3 follows a similar argument than that of Corollary 2.
\end{proof}

\section{Proof of \cref{R_measurements_channels} on noisy quantum state betting (nQSB) games}
\label{AR_measurements_channels}

The proof of this result is similar to that of result 1, and we write below for completeness. We start with the case for QRTs of measurements with general resources.
\begin{proof}(of first part)
	Consider the Arimoto's gap of order $\alpha \in \mathds{\overline R}$, we have the following chain of equalities:
	\begin{align}
	G_{\alpha}^{\mathbb{F}}
	(X;G)_{\mathcal{E},\mathbb{M}}
	&=
	I_{\alpha}
	(X;G)_{\mathcal{E},\mathbb{M}}
	-
	\max_{\mathbb{N} \in \mathbb{F}}
	I_{\alpha}
	(X;G)_{\mathcal{E},\mathbb{N}}
    ,
	\\
	&=
	\sgn(\alpha)
	\log
	\left[
	\frac{
		p_{\alpha}
		(X_{\mathcal{E}}|G_\mathbb{M})
	}
	{
		p_{\alpha}	
		(
		X
		)
	}
	\right]
	-
	\max_{\mathbb{N} \in \mathbb{F}}
	\sgn(\alpha)
	\log
	\left[
	\frac{
		p_{\alpha}
		(X_{\mathcal{E}}|G_\mathbb{N})
	}
	{
		p_{\alpha}	
		(
		X
		)
	}
	\right]
	,
	\\
	&=
	\sgn(\alpha)
	\log
	\left[
	\frac{
		p_{\alpha}
		(X_{\mathcal{E}}|G_\mathbb{M})
	}
	{
	    \displaystyle
	    \max_{\sigma \in {\rm F}}
		p_{\alpha}	
		(X_{\mathcal{E}};G_\mathbb{N})
	}
	\right]
	,\\
	&=
	\sgn(\alpha)
	\log
	\left[
	\frac{
		\sgn(\alpha)\,C\,
		p_{\alpha}
		(X_{\mathcal{E}}|G_\mathbb{M})
	}
	{
	    \displaystyle
	    \max_{\mathbb{N} \in \mathbb{F}}
		\sgn(\alpha)
		\,C\,
		p_{\alpha}	
		(X_{\mathcal{E}};G_\mathbb{N})
	}
	\right]
	,\\
	&=
	\sgn(\alpha)
	\log
	\left[
	\frac{
		\displaystyle
		\max_{b_{X|G}}
		w^{\rm QSB}_{1/\alpha}
		(
		b_{X|G}
		,
		o_X^{\sgn(\alpha)}
		,\mathcal{E},\mathbb{M}
		)
	}{
		\displaystyle
		\max_{\mathbb{N} \in \mathbb{F}}
		\max_{b_{X|G}}
		w^{\rm QSB}_{1/\alpha}
		(
		b_{X|G}
		,
		o_X^{\sgn(\alpha)}
		,\mathcal{E},\mathbb{N}
		)
	}
	\right]
	.
	\end{align}
	The first equality is the definition of the Arimoto's gap for a fixed couple $(\mathcal{E}, \mathbb{M})$. The second equality comes from replacing the R\'enyi entropy and the Arimoto-R\'enyi conditional entropy. In the third equality we reorganised the expression. In the fourth equality we have multiplied and divided by $\sgn(\alpha)\,C$. The fifth and last equality follows from invoking \cref{lemma5}. This proves the claim.
\end{proof}
We now consider the case for QRTs of channels with arbitrary resources.
\begin{proof}(of second part)
	Consider the Arimoto's gap of order $\alpha \in \mathds{\overline R}$, we have the following chain of equalities:
	\begin{align}
	G_{\alpha}^{\mathcal{F}}
	(X;G)_{\mathcal{E},\mathbb{M},\mathcal{N}}
	&=
	I_{\alpha}
	(X;G)_{\mathcal{E},\mathbb{M},\mathcal{N}}
	-
	\max_{\mathcal{\widetilde{N}} 
	\in \mathcal{F}}
	\max_{\mathbb{N}}
	I_{\alpha}
	(X;G)_{\mathcal{E},\mathbb{N},\mathcal{\widetilde{N}}}
    ,
	\\
	&=
	\sgn(\alpha)
	\log
	\left[
	\frac{
		p_{\alpha}
		(X_\mathcal{E}|G_\mathbb{M})_\mathcal{N}
	}
	{
		p_{\alpha}	
		(
		X
		)
	}
	\right]
	-
	\max_{\mathcal{\widetilde{N}} \in \mathcal{F}}
	\max_{\mathbb{N}}
	\sgn(\alpha)
	\log
	\left[
	\frac{
	    p_\alpha
		(X_\mathcal{E}|G_\mathbb{N})_\mathcal{\widetilde{N}}
	}
	{
		p_{\alpha}	
		(
		X
		)
	}
	\right]
	,
	\\
	&=
	\sgn(\alpha)
	\log
	\left[
	\frac{
		p_{\alpha}
		(X_\mathcal{E}|G_\mathbb{M})_\mathcal{N}
	}
	{
	    \displaystyle
	    \max_{\mathcal{\widetilde{N}} \in \mathcal{F}}
	    \max_{\mathbb{N}}
		p_{\alpha}	
		(X_\mathcal{E}|G_\mathbb{N})_\mathcal{\widetilde{N}}
	}
	\right]
	,\\
	&=
	\sgn(\alpha)
	\log
	\left[
	\frac{
		\sgn(\alpha)\,C\,
		p_{\alpha}
		(X_\mathcal{E}|G_\mathbb{M})_\mathcal{N}
	}
	{
	    \displaystyle
	    \max_{\mathcal{\widetilde{N}} \in \mathcal{F}}
	    \max_{\mathbb{N}}
		\sgn(\alpha)
		\,C\,
		p_{\alpha}	
		(X_\mathcal{E}|G_\mathbb{N})_\mathcal{\widetilde{N}}
	}
	\right]
	,\\
	&=
	\sgn(\alpha)
	\log
	\left[
	\frac{
		\displaystyle
		\max_{b_{X|G}}
		w^{\rm QSB}_{1/\alpha}
		(
		b_{X|G}
		,
		o_X^{\sgn(\alpha)}
		,\mathcal{E},\mathbb{M}
		,\mathcal{N})
	}{
		\displaystyle
		\max_{\mathcal{\widetilde{N}} \in \mathcal{F}}
		\max_{\mathbb{N}}
		\max_{b_{X|G}}
		w^{\rm QSB}_{1/\alpha}
		(
		b_{X|G}
		,
		o_X^{\sgn(\alpha)}
		,\mathcal{E},\mathbb{N}
		,\mathcal{\widetilde{N}})
	}
	\right]
	.
	\end{align}
	The first equality is the definition of the Arimoto's gap for a fixed triple $(\mathcal{E}, \mathbb{M}, \mathcal{N})$. The second equality comes from replacing the R\'enyi entropy and the Arimoto-R\'enyi conditional entropy. In the third equality we reorganised the expression. In the fourth equality we have multiplied and divided by $\sgn(\alpha)\,C$. The fifth and last equality follows from invoking \cref{lemma5}. This proves the claim.
\end{proof}

\section{Proof of \cref{R_states} on quantum channel betting (QCB) games}
\label{AR_states}

The proof of this result similar to that of result 1, and we write below for completeness. We start with the case for QRTs of states with arbitrary resources.
\begin{proof}(of first part)
	Consider the Arimoto's gap of order $\alpha \in \mathds{\overline R}$, we have the following chain of equalities:
	\begin{align}
	G_{\alpha}^{{\rm F}}
	(X;G)_{\Lambda,\mathbb{M},\rho}
	&=
	I_{\alpha}
	(X;G)_{\Lambda,\mathbb{M},\rho}
	-
	\max_{\sigma \in {\rm F}}
	I_{\alpha}
	(X;G)_{\Lambda,\mathbb{M},\sigma}
    ,
	\\
	&=
	\sgn(\alpha)
	\log
	\left[
	\frac{
		p_{\alpha}
		(X_\Lambda|G_\mathbb{M})_\rho
	}
	{
		p_{\alpha}	
		(
		X
		)
	}
	\right]
	-
	\max_{\sigma \in {\rm F}}
	\sgn(\alpha)
	\log
	\left[
	\frac{
		p_{\alpha}
		(X_\Lambda;G_\mathbb{M})_\sigma
	}
	{
		p_{\alpha}	
		(
		X
		)
	}
	\right]
	,
	\\
	&=
	\sgn(\alpha)
	\log
	\left[
	\frac{
		p_{\alpha}
		(X_\Lambda|G_\mathbb{M})_\rho
	}
	{
	    \displaystyle
	    \max_{\sigma \in {\rm F}}
		p_{\alpha}	
		(X_\Lambda|G_\mathbb{M})_\sigma
	}
	\right]
	,\\
	&=
	\sgn(\alpha)
	\log
	\left[
	\frac{
		\sgn(\alpha)\,C\,p_{\alpha}
		(X_\Lambda|G_\mathbb{M})_\rho
	}
	{
	    \displaystyle
	    \max_{\sigma \in {\rm F}}
		\sgn(\alpha)
		\,C\,
		p_{\alpha}	
		(X_\Lambda|G_\mathbb{M})_\sigma
	}
	\right]
	,\\
	&=
	\sgn(\alpha)
	\log
	\left[
	\frac{
		\displaystyle
		\max_{b_{X|G}}
		w^{\rm QCB}_{1/\alpha}
		(
		b_{X|G}
		,
		o_X^{\sgn(\alpha)}
		,
		\Lambda,\rho,\mathbb{M}
		)
	}{
		\displaystyle
		\max_{\sigma \in {\rm F}}
		\max_{b_{X|G}}
		w^{\rm QCB}_{1/\alpha}
		(
		b_{X|G}
		,
		o_X^{\sgn(\alpha)}
		,\Lambda,\sigma,\mathbb{M}
		)
	}
	\right]
	.
	\end{align}
	The first equality is the definition of the Arimoto's gap for a fixed triple $(\Lambda, \rho, \mathbb{M})$. The second equality comes from replacing the R\'enyi entropy and the Arimoto-R\'enyi conditional entropy. In the third equality we reorganised the expression. In the fourth equality we have multiplied and divided by $\sgn(\alpha)\,C$. The fifth and last equality follows from invoking \cref{lemma5}. This proves the claim.
\end{proof}

We now consider the case for multi-object QRTs of state-measurement pairs.

\begin{proof}(of second part)
	Consider the Arimoto's gap of order $\alpha \in \mathds{\overline R}$, we have the following chain of equalities:
	\begin{align}
	G_{\alpha}^{{\rm F},\mathbb{F}}
	(X;G)_{\Lambda,\mathbb{M},\rho}
	&=
	I_{\alpha}
	(X;G)_{\Lambda,\mathbb{M},\rho}
	-
	\max_{\sigma \in {\rm F}}
	\max_{\mathbb{N} \in \mathbb{F}}
	I_{\alpha}
	(X;G)_{\Lambda,\mathbb{N},\sigma}
    ,
	\\
	&=
	\sgn(\alpha)
	\log
	\left[
	\frac{
		p_{\alpha}
		(X_\Lambda|G_\mathbb{M})_\rho
	}
	{
		p_{\alpha}	
		(
		X
		)
	}
	\right]
	-
	\max_{\sigma \in {\rm F}}
	\max_{\mathbb{N} \in \mathbb{F}}
	\sgn(\alpha)
	\log
	\left[
	\frac{
		p_{\alpha}
		(X_\Lambda|G_\mathbb{N})_\sigma
	}
	{
		p_{\alpha}	
		(
		X
		)
	}
	\right]
	,
	\\
	&=
	\sgn(\alpha)
	\log
	\left[
	\frac{
		p_{\alpha}
		(X_\Lambda|G_\mathbb{M})_\rho
	}
	{
	    \displaystyle
	    \max_{\sigma \in {\rm F}}
	    \max_{\mathbb{N} \in \mathbb{F}}
		p_{\alpha}	
		(X_\Lambda|G_\mathbb{N})_\sigma
	}
	\right]
	,\\
	&=
	\sgn(\alpha)
	\log
	\left[
	\frac{
		\sgn(\alpha)\,C\,
		p_{\alpha}
		(X_\Lambda|G_\mathbb{M})_\rho
	}
	{
	    \displaystyle
	    \max_{\sigma \in {\rm F}}
	    \max_{\mathbb{N} \in \mathbb{F}}
		\sgn(\alpha)
		\,C\,
		p_{\alpha}	
		(X_\Lambda;G_\mathbb{N})_\sigma
	}
	\right]
	,\\
	&=
	\sgn(\alpha)
	\log
	\left[
	\frac{
		\displaystyle
		\max_{b_{X|G}}
		w^{\rm QCB}_{1/\alpha}
		(
		b_{X|G}
		,
		o_X^{\sgn(\alpha)}
		,
		\Lambda,\rho,\mathbb{M}
		)
	}{
		\displaystyle
		\max_{\sigma \in {\rm F}}
		\max_{\mathbb{N} \in \mathbb{F}}
		\max_{b_{X|G}}
		w^{\rm QCB}_{1/\alpha}
		(
		b_{X|G}
		,
		o_X^{\sgn(\alpha)}
		,\Lambda,\sigma,\mathbb{N}
		)
	}
	\right]
	.
	\end{align}
	The first equality is the definition of the Arimoto's gap for a fixed triple $(\Lambda, \rho, \mathbb{M})$. The second equality comes from replacing the R\'enyi entropy and the Arimoto-R\'enyi conditional entropy. In the third equality we reorganised the expression. In the fourth equality we have multiplied and divided by $\sgn(\alpha)\,C$. The fifth and last equality follows from invoking \cref{lemma5}. This proves the claim.
\end{proof}

\section{Proof of \cref{Rdivergences}}
\label{ARdivergences}

\begin{proof}(of \cref{Rdivergences})
	For $\alpha > 1$ we have:
	\begin{align}
	E_{\alpha}^{\mathcal{S}} 
	(\mathds{M})
	&
	\overset{1}{=}
	\min_{\mathds{N}\in {\rm UI}}
	D_{\alpha}^{\mathcal{S}}
	(\mathds{M}||\mathds{N}),
	\\
	&
	\overset{2}{=}
	\min_{\mathds{N}\in {\rm UI}}
	\max_{p_X}
	D_{\alpha}^{\rm S}
	\left(
	p_{G|X}^{({\mathbb{M}},\mathcal{S})}
	\Big|\Big|
	q_{G|X}^{({\mathbb{N}},\mathcal{S})}
	\Big|
	\,
	p_X
	\right),
	\\
	&
	\overset{3}{=}
	\min_{q_G}
	\max_{p_X}
	D_{\alpha}^{\rm S}
	\left(
	p_{G|X}^{({\mathbb{M}},\mathcal{S})}
	\Big|\Big|
	q_{G}
	\Big|
	\,
	p_X
	\right),
	\\
	&
	\overset{4}{=}
	\min_{q_G}
	\max_{p_X}
	\frac{1}{\alpha-1}
	\log
	\left[
	\sum_x
	p(x)
	\sum_g
	p(g|x)^\alpha
	q(g)^{1-\alpha}
	\right]	,
	\\
	&
	\overset{5}{=}
	\frac{1}{\alpha-1}
	\log
	\left[
	\min_{q_G}
	\max_{p_X}
	\sum_x
	p(x)
	\sum_g
	p(g|x)^\alpha
	q(g)^{1-\alpha}
	\right]	,
	\\
	&
	\overset{6}{=}
	\frac{1}{\alpha-1}
	\log
	\left[
	\min_{q_G}
	\max_{p_X}
	f_\alpha^{\rm S}(q_G,p_X)
	\right]
	,
	\\
	&
	\overset{7}{=}
	\frac{1}{\alpha-1}
	\log
	\left[
	\max_{p_X}
	\min_{q_G}
	f_\alpha^{\rm S}(q_G,p_X)
	\right]
	,
	\\
	&
	\overset{8}{=}
	\max_{p_X}
	\min_{q_G}
	\frac{1}{\alpha-1}
	\log
	\left[
	f_\alpha^{\rm S}(q_G,p_X)
	\right]
	,
	\\
	&
	\overset{9}{=}
	\max_{p_X}
	\min_{q_G}
	D_{\alpha}^{\rm S}
	\left(
	p_{G|X}^{({\mathbb{M}},\mathcal{S})}
	\Big|\Big|
	q_{G}
	\Big|
	\,
	p_X
	\right),
	\\
	&
	\overset{10}{=}
	\max_{p_X}
	I_{\alpha}^{\rm S}
	\left(
	p_{G|X}^{({\mathbb{M}},\mathcal{S})}
	p_X
	\right),
	\\
	&
	\overset{11}{=}
	C_{\alpha}
	\left(
	p_{G|X}^{({\mathbb{M}},\mathcal{S})}
	\right).
\end{align}
	In the first equality we use the definition of $E_{\alpha,\mathcal{S}}(\mathds{M})$. In the second equality we replace $D_{\alpha}^{\mathcal{S}} (\mathds{M} || \mathds{N})$. In the third equality we notice that minimising over uninformative measurements is equivalent to minimising over PMFs $q_G$. In the fourth equality we replace the Sibson's CR-divergence. In the fifth equality we move the optimisation inside $\log(\cdot)$  because the term $\alpha-1$ is positive and because $\log(\cdot)$ is an increasing function. In the sixth equality we introduce the function:
	\begin{align}
	    f_\alpha^{\rm S}(q_G,p_X)
    	\coloneqq
    	\sum_x
    	p(x)
    	\sum_g
    	p(g|x)^\alpha
    	q(g)^{1-\alpha}
    	.
	\end{align}
	In the seventh equality we use Sion's minimax theorem \cite{sion1, sion2} because the function $f_\alpha^{\rm S}(q_G,p_X)$ is being optimised over convex and compact sets, and because it is a convex-concave function. Specifically, the function $f_\alpha^{\rm S}(q_G,p_X)$ is convex in $g_G$ because the function $f(q)=q^{1-\alpha}$ with $\alpha>1$ and positive values of $q$, is convex,	and because the sum of convex functions is convex. The function $f_\alpha^{\rm S}(q_G,p_X)$ is concave in $p_X$ because it is linear in $p_X$. In the eight equality we take the maximisation out of $\log(\cdot)$ because $\alpha-1$ is positive and because $\log(\cdot)$ is an increasing function. In the ninth equality we use the definition of Sibson's CR-divergence. In the tenth equality we use the definition of Sibson's mutual information. In the eleventh and final equality we use Lemma 3. The cases for $0<\alpha<1$ and $\alpha<0$ follow a similar argument, taking into account the sign of $\alpha-1$, and the convexity/concavity of the function $f(q)=q^{1-\alpha}$.
\end{proof}

\section{Proof of \cref{Rmonotones}}
\label{ARmonotones}

\begin{proof}(of \cref{Rmonotones})
	It is straightforward to check that $M_\alpha(\mathds{M})$ is a resource monotone (meaning that it satisfies i)faithfulness and ii)monotonicity) if and only if $E_\alpha(\mathds{M})$ is a resource monotone. We now prove these properties for $E_\alpha(\mathds{M})$. In short, we will expand this function in terms of the R\'enyi divergence, and exploit the properties of this function.
	\\
	\textbf{Part i)} Faithfulness. Consider $\mathbb{M}\in{\rm UI}$, and let us see that this implies $E_\alpha(\mathds{M})=0$ with $\alpha\geq 0$:
	\begin{align}
	E_{\alpha}
	(\mathds{M})
	&
	\overset{1}{=}
	\max_{\mathcal{S}}
	\min_{q_G}
	\max_{p_X}
	D_{\alpha}^{\rm S}
	\left(
	p_{G|X}^{({\mathbb{M}},\mathcal{S})}
	\Big|\Big|
	q_{G}
	\Big|
	\,
	p_X
	\right),
	\\
	&
	\overset{2}{=}
	\max_{\mathcal{S}}
	\min_{q_G}
	\max_{p_X}
	D_{\alpha}
	\left(
	p_{G|X}^{({\mathbb{M}},\mathcal{S})}
	p_X
	\Big|\Big|
	q_{G}
	\,
	p_X
	\right),
	\\
	&
	\overset{3}{=}
	\max_{\mathcal{S}}
	\min_{q_G}
	\max_{p_X}
	D_{\alpha}
	\left(
	p_{G}
	\,
	p_X
	\Big|\Big|
	q_{G}
	\,
	p_X
	\right),
	\\
	&
	\overset{4}{=}
	\max_{\mathcal{S}}
	\max_{p_X}
	\min_{q_G}
	D_{\alpha}
	\left(
	p_{G}
	\,
	p_X
	\Big|\Big|
	q_{G}
	\,
	p_X
	\right),
	\\
	&
	\overset{5}{\leq}
	\max_{\mathcal{S}}
	\max_{p_X}
	D_{\alpha}
	\left(
	p_{G}
	\,
	p_X
	\Big|\Big|
	p_{G}
	\,
	p_X
	\right),
	\\
	&
	\overset{6}{=}
	\max_{\mathcal{S}}
	\max_{p_X}
	\,
	0
	=0.
	\end{align}
	In the first equality we use the definition of the measure. In the second equality we write Sibson's mutual information in terms of the R\'enyi divergence. In the third equality we use the assumption that $\mathds{M}\in {\rm UI}$. In the fourth equality we use Sion's minimax theorem \cite{sion1, sion2}, using the same arguments as in Result 2. In the fifth inequality we use that $q_G=p_G$ is a feasible option. In the sixth equality we invoke the property of the R\'enyi divergence which reads $D_\alpha(p_X||q_X)=0$ if and only if $q_X=p_X$. This chain means that $E_\alpha(\mathds{M})\leq 0$, and remembering that that $E_\alpha(\mathds{M})$ is non-negative (being an optimisation over the R\'enyi divergence which is itself non-negative) implies $E_\alpha(\mathds{M})=0$ as desired.
	
	Consider now that $\mathds{M}$ achieves $E_\alpha(\mathds{M})=0$, and let us prove that $\mathds{M}\in {\rm UI}$. We have:
	\begin{align}
	0
	&
	\overset{1}{=}
	E_{\alpha}
	(\mathds{M})
	\\
	&
	\overset{2}{=}
	\max_{\mathcal{S}}
	\min_{q_G}
	\max_{p_X}
	D_{\alpha}^{\rm S}
	\left(
	p_{G|X}^{({\mathbb{M}},\mathcal{S})}
	\Big|\Big|
	q_{G}
	\Big|
	\,
	p_X
	\right),
	\\
	&
	\overset{3}{=}
	\max_{\mathcal{S}}
	\max_{p_X}
	\min_{q_G}
	D_{\alpha}^{\rm S}
	\left(
	p_{G|X}^{({\mathbb{M}},\mathcal{S})}
	\Big|\Big|
	q_{G}
	\Big|
	\,
	p_X
	\right),
	\\
	&
	\overset{4}{=}
	\max_{\mathcal{S}}
	\max_{p_X}
	\min_{q_G}
	D_{\alpha}
	\left(
	p_{G|X}^{({\mathbb{M}},\mathcal{S})}
	p_X
	\Big|\Big|
	q_{G}
	\,
	p_X
	\right),
	\\
	&
	\overset{5}{=}
	\max_{\mathcal{S}}
	\max_{p_X}
	D_{\alpha}
	\left(
	p_{G|X}^{({\mathbb{M}},\mathcal{S})}
	p_X
	\Big|\Big|
	q_{G}^*
	\,
	p_X
	\right).
	\end{align}
	The first equality is the assumption. In the second equality we invoke the definition of the measure. In the third equality we use Sion's minimax theorem \cite{sion1, sion2} as per Result 2. In the fourth equality we expand Sibson's CR-divergence in terms of the R\'enyi divergence. In the fifth equality we denote the optimal PMF as $q_G^*$. We now notice that the latter equality implies:
	\begin{align}
    	D_{\alpha}
    	\left(
    	p_{G|X}^{({\mathbb{M}},\mathcal{S})}
	    p_X
    	\Big|\Big|
    	q_{G}^*
    	\,
    	p_X
    	\right)
    	=
    	0
    	,
	\end{align}
	from which we get that $p_{G|X}^{({\mathbb{M}},\mathcal{S})}
	    =
	    q_{G}^*$. This means that $p(g|x)=q(g)$, $\forall g,x$, or that $\tr[M_g\rho_x]=\tr[q(g)\mathds{1}\rho_x]$, $\tr[(M_g-q(g)\mathds{1})\rho_x]=0$, $\forall g,x$ which implies $M_g=q(g)\mathds{1}$, $\forall g$, or that $\mathds{M}\in {\rm UI}$ as desired.
	\\
	\textbf{Part ii)} Monotonicity for the order induced by the simulability of measurements. Given two measurements $\mathds{N}=\{N_g\}$, $\mathds{M}=\{M_y\}$  such that $\mathds{N} \leq \mathds{M}$, we now show that this implies $E_\alpha(\mathds{N})\leq E_\alpha(\mathds{M})$. Let us consider that $\mathds{N} \leq \mathds{M}$, meaning that $\forall g$ and some $s_{G|Y}$ we have:
	\begin{align}
	    N_g=\sum_y s(g|y) M_y.
	\end{align}
	This implies that for any set of states $\mathcal{S}=\{\rho_x\}$:
	\begin{align}
	    r(g|x)
	    \coloneqq
	    \tr[N_g\rho_x]
	    =
	    \sum_y
	    s(g|y)
	    p(y|x)
	    ,
	\end{align}
	with $p(y|x)=\tr[M_y\rho_x]$. We now invoke the data processing inequality for the R\'enyi divergence \cite{RD} and get:
	{\small\begin{align}
	    D_{\alpha}
    	\left(
    	r_{G|X}^{({\mathbb{N}},\mathcal{S})}
    	\,
	    p_X
    	\Big|\Big|
    	q_{G}
    	\,
    	p_X
    	\right)
    	\leq
    	D_{\alpha}
    	\left(
    	p_{G|X}^{({\mathbb{M}},\mathcal{S})}
    	\,
	    p_X
    	\Big|\Big|
    	q_{G}
    	\,
    	p_X
    	\right),
	\end{align}}
	with arbitrary PMFs $p_X$ and $p_G$. Recognising that these quantities are the Sibson's CR-divergence leads to:
	\begin{align}
	    D_{\alpha}^{\rm S}
    	\left(
    	r_{G|X}^{({\mathbb{N}},\mathcal{S})}
    	\Big|\Big|
    	q_{G}
    	\Big|
    	p_X
    	\right)
    	\leq
    	D_{\alpha}^{\rm S}
    	\left(
    	p_{G|X}^{({\mathbb{M}},\mathcal{S})}
    	\Big|\Big|
    	q_{G}
    	\Big|
    	p_X
    	\right)
    	.
	\end{align}
	We now perform the optimisations $\max_{\mathcal{S}}$, $\min_{q_G}$, $\max_{p_X}$ on both sides of the inequality and get:
	\begin{align}
	    E_\alpha(\mathds{N})
	   \leq 
	   E_\alpha(\mathds{M})
	   .
	\end{align}
	This finishes the proof for the cases $\alpha \geq 0$. The cases $\alpha < 0$ follow a similar argument.
\end{proof} 

\bibliography{bibliography.bib}

\begin{thebibliography}{86}%
\makeatletter
\providecommand \@ifxundefined [1]{%
 \@ifx{#1\undefined}
}%
\providecommand \@ifnum [1]{%
 \ifnum #1\expandafter \@firstoftwo
 \else \expandafter \@secondoftwo
 \fi
}%
\providecommand \@ifx [1]{%
 \ifx #1\expandafter \@firstoftwo
 \else \expandafter \@secondoftwo
 \fi
}%
\providecommand \natexlab [1]{#1}%
\providecommand \enquote  [1]{``#1''}%
\providecommand \bibnamefont  [1]{#1}%
\providecommand \bibfnamefont [1]{#1}%
\providecommand \citenamefont [1]{#1}%
\providecommand \href@noop [0]{\@secondoftwo}%
\providecommand \href [0]{\begingroup \@sanitize@url \@href}%
\providecommand \@href[1]{\@@startlink{#1}\@@href}%
\providecommand \@@href[1]{\endgroup#1\@@endlink}%
\providecommand \@sanitize@url [0]{\catcode `\\12\catcode `\$12\catcode
  `\&12\catcode `\#12\catcode `\^12\catcode `\_12\catcode `\%12\relax}%
\providecommand \@@startlink[1]{}%
\providecommand \@@endlink[0]{}%
\providecommand \url  [0]{\begingroup\@sanitize@url \@url }%
\providecommand \@url [1]{\endgroup\@href {#1}{\urlprefix }}%
\providecommand \urlprefix  [0]{URL }%
\providecommand \Eprint [0]{\href }%
\providecommand \doibase [0]{http://dx.doi.org/}%
\providecommand \selectlanguage [0]{\@gobble}%
\providecommand \bibinfo  [0]{\@secondoftwo}%
\providecommand \bibfield  [0]{\@secondoftwo}%
\providecommand \translation [1]{[#1]}%
\providecommand \BibitemOpen [0]{}%
\providecommand \bibitemStop [0]{}%
\providecommand \bibitemNoStop [0]{.\EOS\space}%
\providecommand \EOS [0]{\spacefactor3000\relax}%
\providecommand \BibitemShut  [1]{\csname bibitem#1\endcsname}%
\let\auto@bib@innerbib\@empty
\bibitem [{\citenamefont {Nielsen}\ and\ \citenamefont {Chuang}(2000)}]{NC}%
  \BibitemOpen
  \bibfield  {author} {\bibinfo {author} {\bibfnamefont {Michael~A.}\
  \bibnamefont {Nielsen}}\ and\ \bibinfo {author} {\bibfnamefont {Isaac~L.}\
  \bibnamefont {Chuang}},\ }\href@noop {} {\emph {\bibinfo {title} {Quantum
  Computation and Quantum Information}}}\ (\bibinfo  {publisher} {Cambridge
  University Press},\ \bibinfo {year} {2000})\BibitemShut {NoStop}%
\bibitem [{\citenamefont {Dowling}\ and\ \citenamefont {Milburn}(2003)}]{SQR1}%
  \BibitemOpen
  \bibfield  {author} {\bibinfo {author} {\bibfnamefont {Jonathan~P.}\
  \bibnamefont {Dowling}}\ and\ \bibinfo {author} {\bibfnamefont {Gerard~J.}\
  \bibnamefont {Milburn}},\ }\bibfield  {title} {\enquote {\bibinfo {title}
  {Quantum technology: the second quantum revolution},}\ }\href {\doibase
  10.1098/rsta.2003.1227} {\bibfield  {journal} {\bibinfo  {journal}
  {Philosophical Transactions of the Royal Society of London. Series A:
  Mathematical, Physical and Engineering Sciences}\ }\textbf {\bibinfo {volume}
  {361}},\ \bibinfo {pages} {1655--1674} (\bibinfo {year} {2003})}\BibitemShut
  {NoStop}%
\bibitem [{\citenamefont {Pironio}\ \emph {et~al.}(2016)\citenamefont
  {Pironio}, \citenamefont {Scarani},\ and\ \citenamefont {Vidick}}]{SQR2}%
  \BibitemOpen
  \bibfield  {author} {\bibinfo {author} {\bibfnamefont {S}~\bibnamefont
  {Pironio}}, \bibinfo {author} {\bibfnamefont {V}~\bibnamefont {Scarani}}, \
  and\ \bibinfo {author} {\bibfnamefont {T}~\bibnamefont {Vidick}},\ }\bibfield
   {title} {\enquote {\bibinfo {title} {Focus on device independent quantum
  information},}\ }\href {\doibase 10.1088/1367-2630/18/10/100202} {\bibfield
  {journal} {\bibinfo  {journal} {New Journal of Physics}\ }\textbf {\bibinfo
  {volume} {18}},\ \bibinfo {pages} {100202} (\bibinfo {year}
  {2016})}\BibitemShut {NoStop}%
\bibitem [{\citenamefont {Horodecki}\ and\ \citenamefont
  {Oppenheim}(2012)}]{QRT_QP}%
  \BibitemOpen
  \bibfield  {author} {\bibinfo {author} {\bibfnamefont {Michal}\ \bibnamefont
  {Horodecki}}\ and\ \bibinfo {author} {\bibfnamefont {Jonathan}\ \bibnamefont
  {Oppenheim}},\ }\bibfield  {title} {\enquote {\bibinfo {title}
  {({Quantumness} {in} {the} {context} {of}) {resource} {theories}},}\ }\href
  {\doibase 10.1142/s0217979213450197} {\bibfield  {journal} {\bibinfo
  {journal} {Int. J. Mod. Phys. B}\ }\textbf {\bibinfo {volume} {27}},\
  \bibinfo {pages} {1345019} (\bibinfo {year} {2012})}\BibitemShut {NoStop}%
\bibitem [{\citenamefont {Chitambar}\ and\ \citenamefont
  {Gour}(2019)}]{RT_review}%
  \BibitemOpen
  \bibfield  {author} {\bibinfo {author} {\bibfnamefont {Eric}\ \bibnamefont
  {Chitambar}}\ and\ \bibinfo {author} {\bibfnamefont {Gilad}\ \bibnamefont
  {Gour}},\ }\bibfield  {title} {\enquote {\bibinfo {title} {Quantum resource
  theories},}\ }\href {\doibase 10.1103/RevModPhys.91.025001} {\bibfield
  {journal} {\bibinfo  {journal} {Rev. Mod. Phys.}\ }\textbf {\bibinfo {volume}
  {91}},\ \bibinfo {pages} {025001} (\bibinfo {year} {2019})}\BibitemShut
  {NoStop}%
\bibitem [{\citenamefont {Vidal}\ and\ \citenamefont {Tarrach}(1999)}]{RoE}%
  \BibitemOpen
  \bibfield  {author} {\bibinfo {author} {\bibfnamefont {Guifr\'e}\
  \bibnamefont {Vidal}}\ and\ \bibinfo {author} {\bibfnamefont {Rolf}\
  \bibnamefont {Tarrach}},\ }\bibfield  {title} {\enquote {\bibinfo {title}
  {Robustness of entanglement},}\ }\href {\doibase 10.1103/PhysRevA.59.141}
  {\bibfield  {journal} {\bibinfo  {journal} {Phys. Rev. A}\ }\textbf {\bibinfo
  {volume} {59}},\ \bibinfo {pages} {141--155} (\bibinfo {year}
  {1999})}\BibitemShut {NoStop}%
\bibitem [{\citenamefont {Steiner}(2003)}]{GRoE}%
  \BibitemOpen
  \bibfield  {author} {\bibinfo {author} {\bibfnamefont {Michael}\ \bibnamefont
  {Steiner}},\ }\bibfield  {title} {\enquote {\bibinfo {title} {Generalized
  robustness of entanglement},}\ }\href {\doibase 10.1103/PhysRevA.67.054305}
  {\bibfield  {journal} {\bibinfo  {journal} {Phys. Rev. A}\ }\textbf {\bibinfo
  {volume} {67}},\ \bibinfo {pages} {054305} (\bibinfo {year}
  {2003})}\BibitemShut {NoStop}%
\bibitem [{\citenamefont {Cavalcanti}\ and\ \citenamefont
  {Skrzypczyk}(2016)}]{RoNL_RoS_RoI}%
  \BibitemOpen
  \bibfield  {author} {\bibinfo {author} {\bibfnamefont {D.}~\bibnamefont
  {Cavalcanti}}\ and\ \bibinfo {author} {\bibfnamefont {P.}~\bibnamefont
  {Skrzypczyk}},\ }\bibfield  {title} {\enquote {\bibinfo {title} {Quantitative
  relations between measurement incompatibility, quantum steering, and
  nonlocality},}\ }\href {\doibase 10.1103/PhysRevA.93.052112} {\bibfield
  {journal} {\bibinfo  {journal} {Phys. Rev. A}\ }\textbf {\bibinfo {volume}
  {93}},\ \bibinfo {pages} {052112} (\bibinfo {year} {2016})}\BibitemShut
  {NoStop}%
\bibitem [{\citenamefont {Piani}\ and\ \citenamefont {Watrous}(2015)}]{RoS}%
  \BibitemOpen
  \bibfield  {author} {\bibinfo {author} {\bibfnamefont {Marco}\ \bibnamefont
  {Piani}}\ and\ \bibinfo {author} {\bibfnamefont {John}\ \bibnamefont
  {Watrous}},\ }\bibfield  {title} {\enquote {\bibinfo {title} {Necessary and
  sufficient quantum information characterization of einstein-podolsky-rosen
  steering},}\ }\href {\doibase 10.1103/PhysRevLett.114.060404} {\bibfield
  {journal} {\bibinfo  {journal} {Phys. Rev. Lett.}\ }\textbf {\bibinfo
  {volume} {114}},\ \bibinfo {pages} {060404} (\bibinfo {year}
  {2015})}\BibitemShut {NoStop}%
\bibitem [{\citenamefont {Piani}\ \emph {et~al.}(2016)\citenamefont {Piani},
  \citenamefont {Cianciaruso}, \citenamefont {Bromley}, \citenamefont {Napoli},
  \citenamefont {Johnston},\ and\ \citenamefont {Adesso}}]{RoA}%
  \BibitemOpen
  \bibfield  {author} {\bibinfo {author} {\bibfnamefont {Marco}\ \bibnamefont
  {Piani}}, \bibinfo {author} {\bibfnamefont {Marco}\ \bibnamefont
  {Cianciaruso}}, \bibinfo {author} {\bibfnamefont {Thomas~R.}\ \bibnamefont
  {Bromley}}, \bibinfo {author} {\bibfnamefont {Carmine}\ \bibnamefont
  {Napoli}}, \bibinfo {author} {\bibfnamefont {Nathaniel}\ \bibnamefont
  {Johnston}}, \ and\ \bibinfo {author} {\bibfnamefont {Gerardo}\ \bibnamefont
  {Adesso}},\ }\bibfield  {title} {\enquote {\bibinfo {title} {Robustness of
  asymmetry and coherence of quantum states},}\ }\href {\doibase
  10.1103/PhysRevA.93.042107} {\bibfield  {journal} {\bibinfo  {journal} {Phys.
  Rev. A}\ }\textbf {\bibinfo {volume} {93}},\ \bibinfo {pages} {042107}
  (\bibinfo {year} {2016})}\BibitemShut {NoStop}%
\bibitem [{\citenamefont {Napoli}\ \emph {et~al.}(2016)\citenamefont {Napoli},
  \citenamefont {Bromley}, \citenamefont {Cianciaruso}, \citenamefont {Piani},
  \citenamefont {Johnston},\ and\ \citenamefont {Adesso}}]{RoC}%
  \BibitemOpen
  \bibfield  {author} {\bibinfo {author} {\bibfnamefont {Carmine}\ \bibnamefont
  {Napoli}}, \bibinfo {author} {\bibfnamefont {Thomas~R.}\ \bibnamefont
  {Bromley}}, \bibinfo {author} {\bibfnamefont {Marco}\ \bibnamefont
  {Cianciaruso}}, \bibinfo {author} {\bibfnamefont {Marco}\ \bibnamefont
  {Piani}}, \bibinfo {author} {\bibfnamefont {Nathaniel}\ \bibnamefont
  {Johnston}}, \ and\ \bibinfo {author} {\bibfnamefont {Gerardo}\ \bibnamefont
  {Adesso}},\ }\bibfield  {title} {\enquote {\bibinfo {title} {Robustness of
  coherence: An operational and observable measure of quantum coherence},}\
  }\href {\doibase 10.1103/PhysRevLett.116.150502} {\bibfield  {journal}
  {\bibinfo  {journal} {Phys. Rev. Lett.}\ }\textbf {\bibinfo {volume} {116}},\
  \bibinfo {pages} {150502} (\bibinfo {year} {2016})}\BibitemShut {NoStop}%
\bibitem [{\citenamefont {Skrzypczyk}\ and\ \citenamefont {Linden}(2019)}]{SL}%
  \BibitemOpen
  \bibfield  {author} {\bibinfo {author} {\bibfnamefont {Paul}\ \bibnamefont
  {Skrzypczyk}}\ and\ \bibinfo {author} {\bibfnamefont {Noah}\ \bibnamefont
  {Linden}},\ }\bibfield  {title} {\enquote {\bibinfo {title} {Robustness of
  measurement, discrimination games, and accessible information},}\ }\href
  {\doibase 10.1103/PhysRevLett.122.140403} {\bibfield  {journal} {\bibinfo
  {journal} {Phys. Rev. Lett.}\ }\textbf {\bibinfo {volume} {122}},\ \bibinfo
  {pages} {140403} (\bibinfo {year} {2019})}\BibitemShut {NoStop}%
\bibitem [{\citenamefont {\ifmmode \check{S}\else
  \v{S}\fi{}upi\ifmmode~\acute{c}\else \'{c}\fi{}}\ \emph
  {et~al.}(2019)\citenamefont {\ifmmode \check{S}\else
  \v{S}\fi{}upi\ifmmode~\acute{c}\else \'{c}\fi{}}, \citenamefont
  {Skrzypczyk},\ and\ \citenamefont {Cavalcanti}}]{RoT}%
  \BibitemOpen
  \bibfield  {author} {\bibinfo {author} {\bibfnamefont {Ivan}\ \bibnamefont
  {\ifmmode \check{S}\else \v{S}\fi{}upi\ifmmode~\acute{c}\else \'{c}\fi{}}},
  \bibinfo {author} {\bibfnamefont {Paul}\ \bibnamefont {Skrzypczyk}}, \ and\
  \bibinfo {author} {\bibfnamefont {Daniel}\ \bibnamefont {Cavalcanti}},\
  }\bibfield  {title} {\enquote {\bibinfo {title} {Methods to estimate
  entanglement in teleportation experiments},}\ }\href {\doibase
  10.1103/PhysRevA.99.032334} {\bibfield  {journal} {\bibinfo  {journal} {Phys.
  Rev. A}\ }\textbf {\bibinfo {volume} {99}},\ \bibinfo {pages} {032334}
  (\bibinfo {year} {2019})}\BibitemShut {NoStop}%
\bibitem [{\citenamefont {Lipka-Bartosik}\ and\ \citenamefont
  {Skrzypczyk}(2020)}]{RoT2}%
  \BibitemOpen
  \bibfield  {author} {\bibinfo {author} {\bibfnamefont {Patryk}\ \bibnamefont
  {Lipka-Bartosik}}\ and\ \bibinfo {author} {\bibfnamefont {Paul}\ \bibnamefont
  {Skrzypczyk}},\ }\bibfield  {title} {\enquote {\bibinfo {title} {Operational
  advantages provided by nonclassical teleportation},}\ }\href {\doibase
  10.1103/PhysRevResearch.2.023029} {\bibfield  {journal} {\bibinfo  {journal}
  {Phys. Rev. Research}\ }\textbf {\bibinfo {volume} {2}},\ \bibinfo {pages}
  {023029} (\bibinfo {year} {2020})}\BibitemShut {NoStop}%
\bibitem [{\citenamefont {Howard}\ and\ \citenamefont
  {Campbell}(2017)}]{RT_magic}%
  \BibitemOpen
  \bibfield  {author} {\bibinfo {author} {\bibfnamefont {Mark}\ \bibnamefont
  {Howard}}\ and\ \bibinfo {author} {\bibfnamefont {Earl}\ \bibnamefont
  {Campbell}},\ }\bibfield  {title} {\enquote {\bibinfo {title} {Application of
  a resource theory for magic states to fault-tolerant quantum computing},}\
  }\href {\doibase 10.1103/PhysRevLett.118.090501} {\bibfield  {journal}
  {\bibinfo  {journal} {Phys. Rev. Lett.}\ }\textbf {\bibinfo {volume} {118}},\
  \bibinfo {pages} {090501} (\bibinfo {year} {2017})}\BibitemShut {NoStop}%
\bibitem [{\citenamefont {Uola}\ \emph {et~al.}(2019)\citenamefont {Uola},
  \citenamefont {Kraft}, \citenamefont {Shang}, \citenamefont {Yu},\ and\
  \citenamefont {G\"uhne}}]{citeme1}%
  \BibitemOpen
  \bibfield  {author} {\bibinfo {author} {\bibfnamefont {Roope}\ \bibnamefont
  {Uola}}, \bibinfo {author} {\bibfnamefont {Tristan}\ \bibnamefont {Kraft}},
  \bibinfo {author} {\bibfnamefont {Jiangwei}\ \bibnamefont {Shang}}, \bibinfo
  {author} {\bibfnamefont {Xiao-Dong}\ \bibnamefont {Yu}}, \ and\ \bibinfo
  {author} {\bibfnamefont {Otfried}\ \bibnamefont {G\"uhne}},\ }\bibfield
  {title} {\enquote {\bibinfo {title} {Quantifying quantum resources with conic
  programming},}\ }\href {\doibase 10.1103/PhysRevLett.122.130404} {\bibfield
  {journal} {\bibinfo  {journal} {Phys. Rev. Lett.}\ }\textbf {\bibinfo
  {volume} {122}},\ \bibinfo {pages} {130404} (\bibinfo {year}
  {2019})}\BibitemShut {NoStop}%
\bibitem [{\citenamefont {Carmeli}\ \emph {et~al.}(2019)\citenamefont
  {Carmeli}, \citenamefont {Heinosaari}, \citenamefont {Miyadera},\ and\
  \citenamefont {Toigo}}]{citeme2}%
  \BibitemOpen
  \bibfield  {author} {\bibinfo {author} {\bibfnamefont {Claudio}\ \bibnamefont
  {Carmeli}}, \bibinfo {author} {\bibfnamefont {Teiko}\ \bibnamefont
  {Heinosaari}}, \bibinfo {author} {\bibfnamefont {Takayuki}\ \bibnamefont
  {Miyadera}}, \ and\ \bibinfo {author} {\bibfnamefont {Alessandro}\
  \bibnamefont {Toigo}},\ }\href@noop {} {\enquote {\bibinfo {title}
  {Witnessing incompatibility of quantum channels},}\ } (\bibinfo {year}
  {2019}),\ \Eprint {http://arxiv.org/abs/arXiv:1906.10904} {arXiv:1906.10904}
  \BibitemShut {NoStop}%
\bibitem [{\citenamefont {Lipka-Bartosik}\ \emph {et~al.}(2021)\citenamefont
  {Lipka-Bartosik}, \citenamefont {Ducuara}, \citenamefont {Purves},\ and\
  \citenamefont {Skrzypczyk}}]{LBDPS}%
  \BibitemOpen
  \bibfield  {author} {\bibinfo {author} {\bibfnamefont {Patryk}\ \bibnamefont
  {Lipka-Bartosik}}, \bibinfo {author} {\bibfnamefont {Andr\'es~F.}\
  \bibnamefont {Ducuara}}, \bibinfo {author} {\bibfnamefont {Tom}\ \bibnamefont
  {Purves}}, \ and\ \bibinfo {author} {\bibfnamefont {Paul}\ \bibnamefont
  {Skrzypczyk}},\ }\bibfield  {title} {\enquote {\bibinfo {title} {Operational
  significance of the quantum resource theory of buscemi nonlocality},}\ }\href
  {\doibase 10.1103/PRXQuantum.2.020301} {\bibfield  {journal} {\bibinfo
  {journal} {PRX Quantum}\ }\textbf {\bibinfo {volume} {2}},\ \bibinfo {pages}
  {020301} (\bibinfo {year} {2021})}\BibitemShut {NoStop}%
\bibitem [{\citenamefont {Regula}\ \emph {et~al.}(2020)\citenamefont {Regula},
  \citenamefont {Bu}, \citenamefont {Takagi},\ and\ \citenamefont {Liu}}]{TR3}%
  \BibitemOpen
  \bibfield  {author} {\bibinfo {author} {\bibfnamefont {Bartosz}\ \bibnamefont
  {Regula}}, \bibinfo {author} {\bibfnamefont {Kaifeng}\ \bibnamefont {Bu}},
  \bibinfo {author} {\bibfnamefont {Ryuji}\ \bibnamefont {Takagi}}, \ and\
  \bibinfo {author} {\bibfnamefont {Zi-Wen}\ \bibnamefont {Liu}},\ }\bibfield
  {title} {\enquote {\bibinfo {title} {Benchmarking one-shot distillation in
  general quantum resource theories},}\ }\href {\doibase
  10.1103/PhysRevA.101.062315} {\bibfield  {journal} {\bibinfo  {journal}
  {Phys. Rev. A}\ }\textbf {\bibinfo {volume} {101}},\ \bibinfo {pages}
  {062315} (\bibinfo {year} {2020})}\BibitemShut {NoStop}%
\bibitem [{\citenamefont {Fang}\ and\ \citenamefont {Liu}(2020)}]{FL2020}%
  \BibitemOpen
  \bibfield  {author} {\bibinfo {author} {\bibfnamefont {Kun}\ \bibnamefont
  {Fang}}\ and\ \bibinfo {author} {\bibfnamefont {Zi-Wen}\ \bibnamefont
  {Liu}},\ }\bibfield  {title} {\enquote {\bibinfo {title} {No-go theorems for
  quantum resource purification},}\ }\href {\doibase
  10.1103/PhysRevLett.125.060405} {\bibfield  {journal} {\bibinfo  {journal}
  {Phys. Rev. Lett.}\ }\textbf {\bibinfo {volume} {125}},\ \bibinfo {pages}
  {060405} (\bibinfo {year} {2020})}\BibitemShut {NoStop}%
\bibitem [{\citenamefont {Elitzur}\ \emph {et~al.}(1992)\citenamefont
  {Elitzur}, \citenamefont {Popescu},\ and\ \citenamefont {Rohrlich}}]{EPR2}%
  \BibitemOpen
  \bibfield  {author} {\bibinfo {author} {\bibfnamefont {Avshalom~C.}\
  \bibnamefont {Elitzur}}, \bibinfo {author} {\bibfnamefont {Sandu}\
  \bibnamefont {Popescu}}, \ and\ \bibinfo {author} {\bibfnamefont {Daniel}\
  \bibnamefont {Rohrlich}},\ }\bibfield  {title} {\enquote {\bibinfo {title}
  {Quantum nonlocality for each pair in an ensemble},}\ }\href {\doibase
  https://doi.org/10.1016/0375-9601(92)90952-I} {\bibfield  {journal} {\bibinfo
   {journal} {Physics Letters A}\ }\textbf {\bibinfo {volume} {162}},\ \bibinfo
  {pages} {25 -- 28} (\bibinfo {year} {1992})}\BibitemShut {NoStop}%
\bibitem [{\citenamefont {Lewenstein}\ and\ \citenamefont
  {Sanpera}(1998)}]{WoE}%
  \BibitemOpen
  \bibfield  {author} {\bibinfo {author} {\bibfnamefont {Maciej}\ \bibnamefont
  {Lewenstein}}\ and\ \bibinfo {author} {\bibfnamefont {Anna}\ \bibnamefont
  {Sanpera}},\ }\bibfield  {title} {\enquote {\bibinfo {title} {Separability
  and entanglement of composite quantum systems},}\ }\href {\doibase
  10.1103/PhysRevLett.80.2261} {\bibfield  {journal} {\bibinfo  {journal}
  {Phys. Rev. Lett.}\ }\textbf {\bibinfo {volume} {80}},\ \bibinfo {pages}
  {2261--2264} (\bibinfo {year} {1998})}\BibitemShut {NoStop}%
\bibitem [{\citenamefont {Skrzypczyk}\ \emph {et~al.}(2014)\citenamefont
  {Skrzypczyk}, \citenamefont {Navascu\'es},\ and\ \citenamefont
  {Cavalcanti}}]{WoS}%
  \BibitemOpen
  \bibfield  {author} {\bibinfo {author} {\bibfnamefont {Paul}\ \bibnamefont
  {Skrzypczyk}}, \bibinfo {author} {\bibfnamefont {Miguel}\ \bibnamefont
  {Navascu\'es}}, \ and\ \bibinfo {author} {\bibfnamefont {Daniel}\
  \bibnamefont {Cavalcanti}},\ }\bibfield  {title} {\enquote {\bibinfo {title}
  {Quantifying einstein-podolsky-rosen steering},}\ }\href {\doibase
  10.1103/PhysRevLett.112.180404} {\bibfield  {journal} {\bibinfo  {journal}
  {Phys. Rev. Lett.}\ }\textbf {\bibinfo {volume} {112}},\ \bibinfo {pages}
  {180404} (\bibinfo {year} {2014})}\BibitemShut {NoStop}%
\bibitem [{\citenamefont {Bu}\ \emph {et~al.}(2018)\citenamefont {Bu},
  \citenamefont {Anand},\ and\ \citenamefont {Singh}}]{WoAC}%
  \BibitemOpen
  \bibfield  {author} {\bibinfo {author} {\bibfnamefont {Kaifeng}\ \bibnamefont
  {Bu}}, \bibinfo {author} {\bibfnamefont {Namit}\ \bibnamefont {Anand}}, \
  and\ \bibinfo {author} {\bibfnamefont {Uttam}\ \bibnamefont {Singh}},\
  }\bibfield  {title} {\enquote {\bibinfo {title} {Asymmetry and coherence
  weight of quantum states},}\ }\href {\doibase 10.1103/PhysRevA.97.032342}
  {\bibfield  {journal} {\bibinfo  {journal} {Phys. Rev. A}\ }\textbf {\bibinfo
  {volume} {97}},\ \bibinfo {pages} {032342} (\bibinfo {year}
  {2018})}\BibitemShut {NoStop}%
\bibitem [{\citenamefont {Takagi}\ \emph {et~al.}(2019)\citenamefont {Takagi},
  \citenamefont {Regula}, \citenamefont {Bu}, \citenamefont {Liu},\ and\
  \citenamefont {Adesso}}]{TR1}%
  \BibitemOpen
  \bibfield  {author} {\bibinfo {author} {\bibfnamefont {Ryuji}\ \bibnamefont
  {Takagi}}, \bibinfo {author} {\bibfnamefont {Bartosz}\ \bibnamefont
  {Regula}}, \bibinfo {author} {\bibfnamefont {Kaifeng}\ \bibnamefont {Bu}},
  \bibinfo {author} {\bibfnamefont {Zi-Wen}\ \bibnamefont {Liu}}, \ and\
  \bibinfo {author} {\bibfnamefont {Gerardo}\ \bibnamefont {Adesso}},\
  }\bibfield  {title} {\enquote {\bibinfo {title} {Operational advantage of
  quantum resources in subchannel discrimination},}\ }\href {\doibase
  10.1103/PhysRevLett.122.140402} {\bibfield  {journal} {\bibinfo  {journal}
  {Phys. Rev. Lett.}\ }\textbf {\bibinfo {volume} {122}},\ \bibinfo {pages}
  {140402} (\bibinfo {year} {2019})}\BibitemShut {NoStop}%
\bibitem [{\citenamefont {Skrzypczyk}\ \emph {et~al.}(2019)\citenamefont
  {Skrzypczyk}, \citenamefont {\ifmmode \check{S}\else
  \v{S}\fi{}upi\ifmmode~\acute{c}\else \'{c}\fi{}},\ and\ \citenamefont
  {Cavalcanti}}]{RoI_task}%
  \BibitemOpen
  \bibfield  {author} {\bibinfo {author} {\bibfnamefont {Paul}\ \bibnamefont
  {Skrzypczyk}}, \bibinfo {author} {\bibfnamefont {Ivan}\ \bibnamefont
  {\ifmmode \check{S}\else \v{S}\fi{}upi\ifmmode~\acute{c}\else \'{c}\fi{}}}, \
  and\ \bibinfo {author} {\bibfnamefont {Daniel}\ \bibnamefont {Cavalcanti}},\
  }\bibfield  {title} {\enquote {\bibinfo {title} {All sets of incompatible
  measurements give an advantage in quantum state discrimination},}\ }\href
  {\doibase 10.1103/PhysRevLett.122.130403} {\bibfield  {journal} {\bibinfo
  {journal} {Phys. Rev. Lett.}\ }\textbf {\bibinfo {volume} {122}},\ \bibinfo
  {pages} {130403} (\bibinfo {year} {2019})}\BibitemShut {NoStop}%
\bibitem [{\citenamefont {Mori}(2019)}]{RoI_Channels}%
  \BibitemOpen
  \bibfield  {author} {\bibinfo {author} {\bibfnamefont {Junki}\ \bibnamefont
  {Mori}},\ }\href@noop {} {\enquote {\bibinfo {title} {Operational
  characterization of incompatibility of quantum channels with quantum state
  discrimination},}\ } (\bibinfo {year} {2019}),\ \Eprint
  {http://arxiv.org/abs/arXiv:1906.09859} {arXiv:1906.09859} \BibitemShut
  {NoStop}%
\bibitem [{\citenamefont {Takagi}\ and\ \citenamefont {Regula}(2019)}]{TR2}%
  \BibitemOpen
  \bibfield  {author} {\bibinfo {author} {\bibfnamefont {Ryuji}\ \bibnamefont
  {Takagi}}\ and\ \bibinfo {author} {\bibfnamefont {Bartosz}\ \bibnamefont
  {Regula}},\ }\bibfield  {title} {\enquote {\bibinfo {title} {General resource
  theories in quantum mechanics and beyond: Operational characterization via
  discrimination tasks},}\ }\href {\doibase 10.1103/PhysRevX.9.031053}
  {\bibfield  {journal} {\bibinfo  {journal} {Phys. Rev. X}\ }\textbf {\bibinfo
  {volume} {9}},\ \bibinfo {pages} {031053} (\bibinfo {year}
  {2019})}\BibitemShut {NoStop}%
\bibitem [{\citenamefont {Takagi}\ \emph {et~al.}(2020)\citenamefont {Takagi},
  \citenamefont {Wang},\ and\ \citenamefont {Hayashi}}]{RT1}%
  \BibitemOpen
  \bibfield  {author} {\bibinfo {author} {\bibfnamefont {Ryuji}\ \bibnamefont
  {Takagi}}, \bibinfo {author} {\bibfnamefont {Kun}\ \bibnamefont {Wang}}, \
  and\ \bibinfo {author} {\bibfnamefont {Masahito}\ \bibnamefont {Hayashi}},\
  }\bibfield  {title} {\enquote {\bibinfo {title} {Application of the resource
  theory of channels to communication scenarios},}\ }\href {\doibase
  10.1103/PhysRevLett.124.120502} {\bibfield  {journal} {\bibinfo  {journal}
  {Phys. Rev. Lett.}\ }\textbf {\bibinfo {volume} {124}},\ \bibinfo {pages}
  {120502} (\bibinfo {year} {2020})}\BibitemShut {NoStop}%
\bibitem [{\citenamefont {Ducuara}\ and\ \citenamefont
  {Skrzypczyk}(2020)}]{DS}%
  \BibitemOpen
  \bibfield  {author} {\bibinfo {author} {\bibfnamefont {Andr\'es~F.}\
  \bibnamefont {Ducuara}}\ and\ \bibinfo {author} {\bibfnamefont {Paul}\
  \bibnamefont {Skrzypczyk}},\ }\bibfield  {title} {\enquote {\bibinfo {title}
  {Operational interpretation of weight-based resource quantifiers in convex
  quantum resource theories},}\ }\href {\doibase
  10.1103/PhysRevLett.125.110401} {\bibfield  {journal} {\bibinfo  {journal}
  {Phys. Rev. Lett.}\ }\textbf {\bibinfo {volume} {125}},\ \bibinfo {pages}
  {110401} (\bibinfo {year} {2020})}\BibitemShut {NoStop}%
\bibitem [{\citenamefont {Uola}\ \emph {et~al.}(2020)\citenamefont {Uola},
  \citenamefont {Bullock}, \citenamefont {Kraft}, \citenamefont
  {Pellonp\"a\"a},\ and\ \citenamefont {Brunner}}]{uola2020}%
  \BibitemOpen
  \bibfield  {author} {\bibinfo {author} {\bibfnamefont {Roope}\ \bibnamefont
  {Uola}}, \bibinfo {author} {\bibfnamefont {Tom}\ \bibnamefont {Bullock}},
  \bibinfo {author} {\bibfnamefont {Tristan}\ \bibnamefont {Kraft}}, \bibinfo
  {author} {\bibfnamefont {Juha-Pekka}\ \bibnamefont {Pellonp\"a\"a}}, \ and\
  \bibinfo {author} {\bibfnamefont {Nicolas}\ \bibnamefont {Brunner}},\
  }\bibfield  {title} {\enquote {\bibinfo {title} {All quantum resources
  provide an advantage in exclusion tasks},}\ }\href {\doibase
  10.1103/PhysRevLett.125.110402} {\bibfield  {journal} {\bibinfo  {journal}
  {Phys. Rev. Lett.}\ }\textbf {\bibinfo {volume} {125}},\ \bibinfo {pages}
  {110402} (\bibinfo {year} {2020})}\BibitemShut {NoStop}%
\bibitem [{\citenamefont {Gonda}\ and\ \citenamefont {Spekkens}(2019)}]{TG}%
  \BibitemOpen
  \bibfield  {author} {\bibinfo {author} {\bibfnamefont {Tomáš}\ \bibnamefont
  {Gonda}}\ and\ \bibinfo {author} {\bibfnamefont {Robert~W.}\ \bibnamefont
  {Spekkens}},\ }\href@noop {} {\enquote {\bibinfo {title} {Monotones in
  general resource theories},}\ } (\bibinfo {year} {2019}),\ \Eprint
  {http://arxiv.org/abs/arXiv:1912.07085} {arXiv:1912.07085} \BibitemShut
  {NoStop}%
\bibitem [{\citenamefont {Kullback}\ and\ \citenamefont
  {Leibler}(1951)}]{KL1951}%
  \BibitemOpen
  \bibfield  {author} {\bibinfo {author} {\bibfnamefont {S.}~\bibnamefont
  {Kullback}}\ and\ \bibinfo {author} {\bibfnamefont {R.~A.}\ \bibnamefont
  {Leibler}},\ }\bibfield  {title} {\enquote {\bibinfo {title} {On information
  and sufficiency},}\ }\href {\doibase 10.1214/aoms/1177729694} {\bibfield
  {journal} {\bibinfo  {journal} {The Annals of Mathematical Statistics}\
  }\textbf {\bibinfo {volume} {22}},\ \bibinfo {pages} {79--86} (\bibinfo
  {year} {1951})}\BibitemShut {NoStop}%
\bibitem [{\citenamefont {Cover}\ and\ \citenamefont {Thomas}(2005)}]{CT}%
  \BibitemOpen
  \bibfield  {author} {\bibinfo {author} {\bibfnamefont {Thomas~M.}\
  \bibnamefont {Cover}}\ and\ \bibinfo {author} {\bibfnamefont {Joy~A.}\
  \bibnamefont {Thomas}},\ }\href {\doibase 10.1002/047174882x} {\emph
  {\bibinfo {title} {Elements of Information Theory}}}\ (\bibinfo  {publisher}
  {Wiley},\ \bibinfo {year} {2005})\BibitemShut {NoStop}%
\bibitem [{\citenamefont {Rényi}(1961)}]{renyi}%
  \BibitemOpen
  \bibfield  {author} {\bibinfo {author} {\bibfnamefont {Alfréd}\ \bibnamefont
  {Rényi}},\ }\bibfield  {title} {\enquote {\bibinfo {title} {On measures of
  entropy and information},}\ }in\ \href
  {https://projecteuclid.org/euclid.bsmsp/1200512181} {\emph {\bibinfo
  {booktitle} {Proceedings of the Fourth Berkeley Symposium on Mathematical
  Statistics and Probability, Volume 1: Contributions to the Theory of
  Statistics}}}\ (\bibinfo  {publisher} {University of California Press},\
  \bibinfo {address} {Berkeley, Calif.},\ \bibinfo {year} {1961})\ pp.\
  \bibinfo {pages} {547--561}\BibitemShut {NoStop}%
\bibitem [{\citenamefont {{van Erven}}\ and\ \citenamefont
  {{Harremos}}(2014{\natexlab{a}})}]{RD}%
  \BibitemOpen
  \bibfield  {author} {\bibinfo {author} {\bibfnamefont {T.}~\bibnamefont {{van
  Erven}}}\ and\ \bibinfo {author} {\bibfnamefont {P.}~\bibnamefont
  {{Harremos}}},\ }\bibfield  {title} {\enquote {\bibinfo {title} {Rényi
  divergence and kullback-leibler divergence},}\ }\href@noop {} {\bibfield
  {journal} {\bibinfo  {journal} {IEEE Transactions on Information Theory}\
  }\textbf {\bibinfo {volume} {60}},\ \bibinfo {pages} {3797--3820} (\bibinfo
  {year} {2014}{\natexlab{a}})}\BibitemShut {NoStop}%
\bibitem [{\citenamefont {{Fehr}}\ and\ \citenamefont
  {{Berens}}(2014)}]{review_RCE}%
  \BibitemOpen
  \bibfield  {author} {\bibinfo {author} {\bibfnamefont {S.}~\bibnamefont
  {{Fehr}}}\ and\ \bibinfo {author} {\bibfnamefont {S.}~\bibnamefont
  {{Berens}}},\ }\bibfield  {title} {\enquote {\bibinfo {title} {On the
  conditional rényi entropy},}\ }\href@noop {} {\bibfield  {journal} {\bibinfo
   {journal} {IEEE Transactions on Information Theory}\ }\textbf {\bibinfo
  {volume} {60}},\ \bibinfo {pages} {6801--6810} (\bibinfo {year}
  {2014})}\BibitemShut {NoStop}%
\bibitem [{\citenamefont {Bleuler}\ \emph {et~al.}(2020)\citenamefont
  {Bleuler}, \citenamefont {Lapidoth},\ and\ \citenamefont {Pfister}}]{BLP1}%
  \BibitemOpen
  \bibfield  {author} {\bibinfo {author} {\bibfnamefont {C{\'{e}}dric}\
  \bibnamefont {Bleuler}}, \bibinfo {author} {\bibfnamefont {Amos}\
  \bibnamefont {Lapidoth}}, \ and\ \bibinfo {author} {\bibfnamefont
  {Christoph}\ \bibnamefont {Pfister}},\ }\bibfield  {title} {\enquote
  {\bibinfo {title} {Conditional r{\'{e}}nyi divergences and horse betting},}\
  }\href {\doibase 10.3390/e22030316} {\bibfield  {journal} {\bibinfo
  {journal} {Entropy}\ }\textbf {\bibinfo {volume} {22}},\ \bibinfo {pages}
  {316} (\bibinfo {year} {2020})}\BibitemShut {NoStop}%
\bibitem [{\citenamefont {{Verdú}}(2015)}]{review_RMI}%
  \BibitemOpen
  \bibfield  {author} {\bibinfo {author} {\bibfnamefont {S.}~\bibnamefont
  {{Verdú}}},\ }\bibfield  {title} {\enquote {\bibinfo {title}
  {$\alpha$-mutual information},}\ }in\ \href@noop {} {\emph {\bibinfo
  {booktitle} {2015 Information Theory and Applications Workshop (ITA)}}}\
  (\bibinfo {year} {2015})\ pp.\ \bibinfo {pages} {1--6}\BibitemShut {NoStop}%
\bibitem [{\citenamefont {Sibson}(1969)}]{sibson}%
  \BibitemOpen
  \bibfield  {author} {\bibinfo {author} {\bibfnamefont {Robin}\ \bibnamefont
  {Sibson}},\ }\bibfield  {title} {\enquote {\bibinfo {title} {Information
  radius},}\ }\href {\doibase 10.1007/BF00537520} {\bibfield  {journal}
  {\bibinfo  {journal} {Zeitschrift f{\"u}r Wahrscheinlichkeitstheorie und
  Verwandte Gebiete}\ }\textbf {\bibinfo {volume} {14}},\ \bibinfo {pages}
  {149--160} (\bibinfo {year} {1969})}\BibitemShut {NoStop}%
\bibitem [{\citenamefont {Arimoto}(1977)}]{arimoto}%
  \BibitemOpen
  \bibfield  {author} {\bibinfo {author} {\bibfnamefont {S.}~\bibnamefont
  {Arimoto}},\ }\bibfield  {title} {\enquote {\bibinfo {title} {Information
  measures and capacity of order $\alpha$ for discrete memoryless channels},}\
  }\href {https://ci.nii.ac.jp/naid/10022581674/en/} {\bibfield  {journal}
  {\bibinfo  {journal} {Topics in Information Theory}\ } (\bibinfo {year}
  {1977})}\BibitemShut {NoStop}%
\bibitem [{\citenamefont {Csiszar}(1995)}]{csiszar}%
  \BibitemOpen
  \bibfield  {author} {\bibinfo {author} {\bibfnamefont {I.}~\bibnamefont
  {Csiszar}},\ }\bibfield  {title} {\enquote {\bibinfo {title} {Generalized
  cutoff rates and renyi{\textquotesingle}s information measures},}\ }\href
  {\doibase 10.1109/18.370121} {\bibfield  {journal} {\bibinfo  {journal}
  {{IEEE} Transactions on Information Theory}\ }\textbf {\bibinfo {volume}
  {41}},\ \bibinfo {pages} {26--34} (\bibinfo {year} {1995})}\BibitemShut
  {NoStop}%
\bibitem [{\citenamefont {Lapidoth}\ and\ \citenamefont {Pfister}(2019)}]{LP}%
  \BibitemOpen
  \bibfield  {author} {\bibinfo {author} {\bibfnamefont {Amos}\ \bibnamefont
  {Lapidoth}}\ and\ \bibinfo {author} {\bibfnamefont {Christoph}\ \bibnamefont
  {Pfister}},\ }\bibfield  {title} {\enquote {\bibinfo {title} {Two measures of
  dependence},}\ }\href {\doibase 10.3390/e21080778} {\bibfield  {journal}
  {\bibinfo  {journal} {Entropy}\ }\textbf {\bibinfo {volume} {21}},\ \bibinfo
  {pages} {778} (\bibinfo {year} {2019})}\BibitemShut {NoStop}%
\bibitem [{\citenamefont {{Tomamichel}}\ and\ \citenamefont
  {{Hayashi}}(2018)}]{TH}%
  \BibitemOpen
  \bibfield  {author} {\bibinfo {author} {\bibfnamefont {M.}~\bibnamefont
  {{Tomamichel}}}\ and\ \bibinfo {author} {\bibfnamefont {M.}~\bibnamefont
  {{Hayashi}}},\ }\bibfield  {title} {\enquote {\bibinfo {title} {Operational
  interpretation of rényi information measures via composite hypothesis
  testing against product and markov distributions},}\ }\href@noop {}
  {\bibfield  {journal} {\bibinfo  {journal} {IEEE Transactions on Information
  Theory}\ }\textbf {\bibinfo {volume} {64}},\ \bibinfo {pages} {1064--1082}
  (\bibinfo {year} {2018})}\BibitemShut {NoStop}%
\bibitem [{\citenamefont {Wilde}(2017)}]{Wilde_book}%
  \BibitemOpen
  \bibfield  {author} {\bibinfo {author} {\bibfnamefont {M.M.}\ \bibnamefont
  {Wilde}},\ }\href {https://books.google.co.uk/books?id=gYcHDgAAQBAJ} {\emph
  {\bibinfo {title} {Quantum Information Theory}}}\ (\bibinfo  {publisher}
  {Cambridge University Press},\ \bibinfo {year} {2017})\BibitemShut {NoStop}%
\bibitem [{\citenamefont {Ducuara}\ \emph {et~al.}(2020)\citenamefont
  {Ducuara}, \citenamefont {Lipka-Bartosik},\ and\ \citenamefont
  {Skrzypczyk}}]{MO}%
  \BibitemOpen
  \bibfield  {author} {\bibinfo {author} {\bibfnamefont {Andr\'es~F.}\
  \bibnamefont {Ducuara}}, \bibinfo {author} {\bibfnamefont {Patryk}\
  \bibnamefont {Lipka-Bartosik}}, \ and\ \bibinfo {author} {\bibfnamefont
  {Paul}\ \bibnamefont {Skrzypczyk}},\ }\bibfield  {title} {\enquote {\bibinfo
  {title} {Multiobject operational tasks for convex quantum resource theories
  of state-measurement pairs},}\ }\href {\doibase
  10.1103/PhysRevResearch.2.033374} {\bibfield  {journal} {\bibinfo  {journal}
  {Phys. Rev. Research}\ }\textbf {\bibinfo {volume} {2}},\ \bibinfo {pages}
  {033374} (\bibinfo {year} {2020})}\BibitemShut {NoStop}%
\bibitem [{\citenamefont {Datta}(2009)}]{datta}%
  \BibitemOpen
  \bibfield  {author} {\bibinfo {author} {\bibfnamefont {Nilanjana}\
  \bibnamefont {Datta}},\ }\bibfield  {title} {\enquote {\bibinfo {title}
  {{MAX}-{RELATIVE} {ENTROPY} {OF} {ENTANGLEMENT}, {ALIAS} {LOG}
  {ROBUSTNESS}},}\ }\href {\doibase 10.1142/s0219749909005298} {\bibfield
  {journal} {\bibinfo  {journal} {International Journal of Quantum
  Information}\ }\textbf {\bibinfo {volume} {07}},\ \bibinfo {pages} {475--491}
  (\bibinfo {year} {2009})}\BibitemShut {NoStop}%
\bibitem [{\citenamefont {von Neumann}\ and\ \citenamefont
  {Morgenstern}(2007)}]{risk_vNM}%
  \BibitemOpen
  \bibfield  {author} {\bibinfo {author} {\bibfnamefont {John}\ \bibnamefont
  {von Neumann}}\ and\ \bibinfo {author} {\bibfnamefont {Oskar}\ \bibnamefont
  {Morgenstern}},\ }\href {\doibase 10.1515/9781400829460} {\emph {\bibinfo
  {title} {Theory of Games and Economic Behavior (60th Anniversary
  Commemorative Edition)}}}\ (\bibinfo  {publisher} {Princeton University
  Press},\ \bibinfo {year} {2007})\BibitemShut {NoStop}%
\bibitem [{\citenamefont {B{\c{a}}k}(2020)}]{nobel}%
  \BibitemOpen
  \bibfield  {author} {\bibinfo {author} {\bibfnamefont {Sylwia}\ \bibnamefont
  {B{\c{a}}k}},\ }\bibfield  {title} {\enquote {\bibinfo {title} {The problem
  of uncertainty and risk as a subject of research of the nobel prize laureates
  in economic sciences},}\ }\href {\doibase 10.22367/jem.2020.39.02} {\bibfield
   {journal} {\bibinfo  {journal} {Journal of Economics and Management}\
  }\textbf {\bibinfo {volume} {39}},\ \bibinfo {pages} {21--40} (\bibinfo
  {year} {2020})}\BibitemShut {NoStop}%
\bibitem [{\citenamefont {Eeckhoudt}\ \emph {et~al.}(2011)\citenamefont
  {Eeckhoudt}, \citenamefont {Gollier},\ and\ \citenamefont
  {Schlesinger}}]{risk_EGS}%
  \BibitemOpen
  \bibfield  {author} {\bibinfo {author} {\bibfnamefont {Louis}\ \bibnamefont
  {Eeckhoudt}}, \bibinfo {author} {\bibfnamefont {Christian}\ \bibnamefont
  {Gollier}}, \ and\ \bibinfo {author} {\bibfnamefont {Harris}\ \bibnamefont
  {Schlesinger}},\ }\href {\doibase 10.2307/j.ctvcm4j15} {\emph {\bibinfo
  {title} {Economic and Financial Decisions under Risk}}}\ (\bibinfo
  {publisher} {Princeton University Press},\ \bibinfo {year}
  {2011})\BibitemShut {NoStop}%
\bibitem [{\citenamefont {Robson}(1996)}]{risk_biology1}%
  \BibitemOpen
  \bibfield  {author} {\bibinfo {author} {\bibfnamefont {Arthur~J.}\
  \bibnamefont {Robson}},\ }\bibfield  {title} {\enquote {\bibinfo {title} {A
  biological basis for expected and non-expected utility},}\ }\href {\doibase
  10.1006/jeth.1996.0023} {\bibfield  {journal} {\bibinfo  {journal} {Journal
  of Economic Theory}\ }\textbf {\bibinfo {volume} {68}},\ \bibinfo {pages}
  {397--424} (\bibinfo {year} {1996})}\BibitemShut {NoStop}%
\bibitem [{\citenamefont {Zhang}\ \emph {et~al.}(2014)\citenamefont {Zhang},
  \citenamefont {Brennan},\ and\ \citenamefont {Lo}}]{risk_biology2}%
  \BibitemOpen
  \bibfield  {author} {\bibinfo {author} {\bibfnamefont {Ruixun}\ \bibnamefont
  {Zhang}}, \bibinfo {author} {\bibfnamefont {Thomas~J.}\ \bibnamefont
  {Brennan}}, \ and\ \bibinfo {author} {\bibfnamefont {Andrew~W.}\ \bibnamefont
  {Lo}},\ }\bibfield  {title} {\enquote {\bibinfo {title} {The origin of risk
  aversion},}\ }\href {\doibase 10.1073/pnas.1406755111} {\bibfield  {journal}
  {\bibinfo  {journal} {Proceedings of the National Academy of Sciences}\
  }\textbf {\bibinfo {volume} {111}},\ \bibinfo {pages} {17777--17782}
  (\bibinfo {year} {2014})}\BibitemShut {NoStop}%
\bibitem [{\citenamefont {Knoch}\ \emph {et~al.}(2006)\citenamefont {Knoch},
  \citenamefont {Gianotti}, \citenamefont {Pascual-Leone}, \citenamefont
  {Treyer}, \citenamefont {Regard}, \citenamefont {Hohmann},\ and\
  \citenamefont {Brugger}}]{NS1}%
  \BibitemOpen
  \bibfield  {author} {\bibinfo {author} {\bibfnamefont {D.}~\bibnamefont
  {Knoch}}, \bibinfo {author} {\bibfnamefont {L.~R.~R.}\ \bibnamefont
  {Gianotti}}, \bibinfo {author} {\bibfnamefont {A.}~\bibnamefont
  {Pascual-Leone}}, \bibinfo {author} {\bibfnamefont {V.}~\bibnamefont
  {Treyer}}, \bibinfo {author} {\bibfnamefont {M.}~\bibnamefont {Regard}},
  \bibinfo {author} {\bibfnamefont {M.}~\bibnamefont {Hohmann}}, \ and\
  \bibinfo {author} {\bibfnamefont {P.}~\bibnamefont {Brugger}},\ }\bibfield
  {title} {\enquote {\bibinfo {title} {Disruption of right prefrontal cortex by
  low-frequency repetitive transcranial magnetic stimulation induces
  risk-taking behavior},}\ }\href {\doibase 10.1523/jneurosci.0804-06.2006}
  {\bibfield  {journal} {\bibinfo  {journal} {Journal of Neuroscience}\
  }\textbf {\bibinfo {volume} {26}},\ \bibinfo {pages} {6469--6472} (\bibinfo
  {year} {2006})}\BibitemShut {NoStop}%
\bibitem [{\citenamefont {Fecteau}\ \emph {et~al.}(2007)\citenamefont
  {Fecteau}, \citenamefont {Pascual-Leone}, \citenamefont {Zald}, \citenamefont
  {Liguori}, \citenamefont {Theoret}, \citenamefont {Boggio},\ and\
  \citenamefont {Fregni}}]{NS2}%
  \BibitemOpen
  \bibfield  {author} {\bibinfo {author} {\bibfnamefont {S.}~\bibnamefont
  {Fecteau}}, \bibinfo {author} {\bibfnamefont {A.}~\bibnamefont
  {Pascual-Leone}}, \bibinfo {author} {\bibfnamefont {D.~H.}\ \bibnamefont
  {Zald}}, \bibinfo {author} {\bibfnamefont {P.}~\bibnamefont {Liguori}},
  \bibinfo {author} {\bibfnamefont {H.}~\bibnamefont {Theoret}}, \bibinfo
  {author} {\bibfnamefont {P.~S.}\ \bibnamefont {Boggio}}, \ and\ \bibinfo
  {author} {\bibfnamefont {F.}~\bibnamefont {Fregni}},\ }\bibfield  {title}
  {\enquote {\bibinfo {title} {Activation of prefrontal cortex by transcranial
  direct current stimulation reduces appetite for risk during ambiguous
  decision making},}\ }\href {\doibase 10.1523/jneurosci.0314-07.2007}
  {\bibfield  {journal} {\bibinfo  {journal} {Journal of Neuroscience}\
  }\textbf {\bibinfo {volume} {27}},\ \bibinfo {pages} {6212--6218} (\bibinfo
  {year} {2007})}\BibitemShut {NoStop}%
\bibitem [{\citenamefont {Tom}\ \emph {et~al.}(2007)\citenamefont {Tom},
  \citenamefont {Fox}, \citenamefont {Trepel},\ and\ \citenamefont
  {Poldrack}}]{NS3}%
  \BibitemOpen
  \bibfield  {author} {\bibinfo {author} {\bibfnamefont {S.~M.}\ \bibnamefont
  {Tom}}, \bibinfo {author} {\bibfnamefont {C.~R.}\ \bibnamefont {Fox}},
  \bibinfo {author} {\bibfnamefont {C.}~\bibnamefont {Trepel}}, \ and\ \bibinfo
  {author} {\bibfnamefont {R.~A.}\ \bibnamefont {Poldrack}},\ }\bibfield
  {title} {\enquote {\bibinfo {title} {The neural basis of loss aversion in
  decision-making under risk},}\ }\href {\doibase 10.1126/science.1134239}
  {\bibfield  {journal} {\bibinfo  {journal} {Science}\ }\textbf {\bibinfo
  {volume} {315}},\ \bibinfo {pages} {515--518} (\bibinfo {year}
  {2007})}\BibitemShut {NoStop}%
\bibitem [{\citenamefont {Bernoulli}(1954)}]{risk_bernoulli}%
  \BibitemOpen
  \bibfield  {author} {\bibinfo {author} {\bibfnamefont {Daniel}\ \bibnamefont
  {Bernoulli}},\ }\bibfield  {title} {\enquote {\bibinfo {title} {Exposition of
  a new theory on the measurement of risk},}\ }\href {\doibase 10.2307/1909829}
  {\bibfield  {journal} {\bibinfo  {journal} {Econometrica}\ }\textbf {\bibinfo
  {volume} {22}},\ \bibinfo {pages} {23} (\bibinfo {year} {1954})}\BibitemShut
  {NoStop}%
\bibitem [{\citenamefont {Arrow}(1965)}]{risk_arrow}%
  \BibitemOpen
  \bibfield  {author} {\bibinfo {author} {\bibfnamefont {K.J.}\ \bibnamefont
  {Arrow}},\ }\href {https://books.google.com.co/books?id=hnNEAAAAIAAJ} {\emph
  {\bibinfo {title} {Aspects of the theory of risk-bearing}}}\ (\bibinfo
  {publisher} {Yrj{\"o} Jahnssonin S{\"a}{\"a}ti{\"o}},\ \bibinfo {year}
  {1965})\BibitemShut {NoStop}%
\bibitem [{\citenamefont {Pratt}(1964)}]{risk_pratt}%
  \BibitemOpen
  \bibfield  {author} {\bibinfo {author} {\bibfnamefont {John~W.}\ \bibnamefont
  {Pratt}},\ }\bibfield  {title} {\enquote {\bibinfo {title} {Risk aversion in
  the small and in the large},}\ }\href {\doibase 10.2307/1913738} {\bibfield
  {journal} {\bibinfo  {journal} {Econometrica}\ }\textbf {\bibinfo {volume}
  {32}},\ \bibinfo {pages} {122} (\bibinfo {year} {1964})}\BibitemShut
  {NoStop}%
\bibitem [{\citenamefont {de~Finetti}(1952)}]{risk_finetti}%
  \BibitemOpen
  \bibfield  {author} {\bibinfo {author} {\bibfnamefont {Bruno}\ \bibnamefont
  {de~Finetti}},\ }\bibfield  {title} {\enquote {\bibinfo {title} {{SULLA}
  {PREFERIBILITÀ}},}\ }\href {https://www.jstor.org/stable/23236169}
  {\bibfield  {journal} {\bibinfo  {journal} {Giornale degli Economisti e
  Annali di Economia}\ }\textbf {\bibinfo {volume} {11}},\ \bibinfo {pages}
  {685--709} (\bibinfo {year} {1952})}\BibitemShut {NoStop}%
\bibitem [{\citenamefont {{van Erven}}\ and\ \citenamefont
  {{Harremos}}(2014{\natexlab{b}})}]{NV1}%
  \BibitemOpen
  \bibfield  {author} {\bibinfo {author} {\bibfnamefont {T.}~\bibnamefont {{van
  Erven}}}\ and\ \bibinfo {author} {\bibfnamefont {P.}~\bibnamefont
  {{Harremos}}},\ }\bibfield  {title} {\enquote {\bibinfo {title} {Rényi
  divergence and kullback-leibler divergence},}\ }\href@noop {} {\bibfield
  {journal} {\bibinfo  {journal} {IEEE Transactions on Information Theory}\
  }\textbf {\bibinfo {volume} {60}},\ \bibinfo {pages} {3797--3820} (\bibinfo
  {year} {2014}{\natexlab{b}})}\BibitemShut {NoStop}%
\bibitem [{\citenamefont {{Sason}}\ and\ \citenamefont {{Verdú}}(2018)}]{NV2}%
  \BibitemOpen
  \bibfield  {author} {\bibinfo {author} {\bibfnamefont {I.}~\bibnamefont
  {{Sason}}}\ and\ \bibinfo {author} {\bibfnamefont {S.}~\bibnamefont
  {{Verdú}}},\ }\bibfield  {title} {\enquote {\bibinfo {title} {Improved
  bounds on guessing moments via rényi measures},}\ }in\ \href@noop {} {\emph
  {\bibinfo {booktitle} {2018 IEEE International Symposium on Information
  Theory (ISIT)}}}\ (\bibinfo {year} {2018})\ pp.\ \bibinfo {pages}
  {566--570}\BibitemShut {NoStop}%
\bibitem [{\citenamefont {Valverde-Albacete}\ and\ \citenamefont
  {Pel{\'{a}}ez-Moreno}(2019{\natexlab{a}})}]{SR1}%
  \BibitemOpen
  \bibfield  {author} {\bibinfo {author} {\bibfnamefont {Francisco}\
  \bibnamefont {Valverde-Albacete}}\ and\ \bibinfo {author} {\bibfnamefont
  {Carmen}\ \bibnamefont {Pel{\'{a}}ez-Moreno}},\ }\bibfield  {title} {\enquote
  {\bibinfo {title} {The case for shifting the renyi entropy},}\ }\href
  {\doibase 10.3390/e21010046} {\bibfield  {journal} {\bibinfo  {journal}
  {Entropy}\ }\textbf {\bibinfo {volume} {21}},\ \bibinfo {pages} {46}
  (\bibinfo {year} {2019}{\natexlab{a}})}\BibitemShut {NoStop}%
\bibitem [{\citenamefont {Valverde-Albacete}\ and\ \citenamefont
  {Pel{\'{a}}ez-Moreno}(2019{\natexlab{b}})}]{SR2}%
  \BibitemOpen
  \bibfield  {author} {\bibinfo {author} {\bibfnamefont {Francisco~J.}\
  \bibnamefont {Valverde-Albacete}}\ and\ \bibinfo {author} {\bibfnamefont
  {Carmen}\ \bibnamefont {Pel{\'{a}}ez-Moreno}},\ }\bibfield  {title} {\enquote
  {\bibinfo {title} {The r{\'{e}}nyi entropies operate in positive
  semifields},}\ }\href {\doibase 10.3390/e21080780} {\bibfield  {journal}
  {\bibinfo  {journal} {Entropy}\ }\textbf {\bibinfo {volume} {21}},\ \bibinfo
  {pages} {780} (\bibinfo {year} {2019}{\natexlab{b}})}\BibitemShut {NoStop}%
\bibitem [{\citenamefont {{Aishwarya}}\ and\ \citenamefont
  {{Madiman}}(2019)}]{remarks}%
  \BibitemOpen
  \bibfield  {author} {\bibinfo {author} {\bibfnamefont {G.}~\bibnamefont
  {{Aishwarya}}}\ and\ \bibinfo {author} {\bibfnamefont {M.}~\bibnamefont
  {{Madiman}}},\ }\bibfield  {title} {\enquote {\bibinfo {title} {Remarks on
  rényi versions of conditional entropy and mutual information},}\ }in\
  \href@noop {} {\emph {\bibinfo {booktitle} {2019 IEEE International Symposium
  on Information Theory (ISIT)}}}\ (\bibinfo {year} {2019})\ pp.\ \bibinfo
  {pages} {1117--1121}\BibitemShut {NoStop}%
\bibitem [{\citenamefont {Nakiboğlu}(2019)}]{Nakiboglu}%
  \BibitemOpen
  \bibfield  {author} {\bibinfo {author} {\bibfnamefont {Barış}\ \bibnamefont
  {Nakiboğlu}},\ }\bibfield  {title} {\enquote {\bibinfo {title} {The rényi
  capacity and center},}\ }\href {\doibase 10.1109/TIT.2018.2861002} {\bibfield
   {journal} {\bibinfo  {journal} {IEEE Transactions on Information Theory}\
  }\textbf {\bibinfo {volume} {65}},\ \bibinfo {pages} {841--860} (\bibinfo
  {year} {2019})}\BibitemShut {NoStop}%
\bibitem [{\citenamefont {Guerini}\ \emph {et~al.}(2017)\citenamefont
  {Guerini}, \citenamefont {Bavaresco}, \citenamefont {Cunha},\ and\
  \citenamefont {Ac{\'{\i}}n}}]{simulability}%
  \BibitemOpen
  \bibfield  {author} {\bibinfo {author} {\bibfnamefont {Leonardo}\
  \bibnamefont {Guerini}}, \bibinfo {author} {\bibfnamefont {Jessica}\
  \bibnamefont {Bavaresco}}, \bibinfo {author} {\bibfnamefont {Marcelo~Terra}\
  \bibnamefont {Cunha}}, \ and\ \bibinfo {author} {\bibfnamefont {Antonio}\
  \bibnamefont {Ac{\'{\i}}n}},\ }\bibfield  {title} {\enquote {\bibinfo {title}
  {Operational framework for quantum measurement simulability},}\ }\href
  {\doibase 10.1063/1.4994303} {\bibfield  {journal} {\bibinfo  {journal}
  {Journal of Mathematical Physics}\ }\textbf {\bibinfo {volume} {58}},\
  \bibinfo {pages} {092102} (\bibinfo {year} {2017})}\BibitemShut {NoStop}%
\bibitem [{\citenamefont {{Kelly}}(1956)}]{kelly}%
  \BibitemOpen
  \bibfield  {author} {\bibinfo {author} {\bibfnamefont {J.~L.}\ \bibnamefont
  {{Kelly}}},\ }\bibfield  {title} {\enquote {\bibinfo {title} {A new
  interpretation of information rate},}\ }\href@noop {} {\bibfield  {journal}
  {\bibinfo  {journal} {The Bell System Technical Journal}\ }\textbf {\bibinfo
  {volume} {35}},\ \bibinfo {pages} {917--926} (\bibinfo {year}
  {1956})}\BibitemShut {NoStop}%
\bibitem [{\citenamefont {Moser}(2020)}]{LN_moser}%
  \BibitemOpen
  \bibfield  {author} {\bibinfo {author} {\bibfnamefont {S.~M.}\ \bibnamefont
  {Moser}},\ }\href {https://moser-isi.ethz.ch/scripts.html} {\enquote
  {\bibinfo {title} {Information theory (lecture notes)},}\ } (\bibinfo {year}
  {2020})\BibitemShut {NoStop}%
\bibitem [{\citenamefont {Tomamichel}(2015)}]{review_qrd}%
  \BibitemOpen
  \bibfield  {author} {\bibinfo {author} {\bibfnamefont {M.}~\bibnamefont
  {Tomamichel}},\ }\href {https://books.google.com.co/books?id=643DCgAAQBAJ}
  {\emph {\bibinfo {title} {Quantum Information Processing with Finite
  Resources: Mathematical Foundations}}},\ SpringerBriefs in Mathematical
  Physics\ (\bibinfo  {publisher} {Springer International Publishing},\
  \bibinfo {year} {2015})\BibitemShut {NoStop}%
\bibitem [{\citenamefont {Petz}(1986)}]{petz-renyi}%
  \BibitemOpen
  \bibfield  {author} {\bibinfo {author} {\bibfnamefont {D{\'{e}}nes}\
  \bibnamefont {Petz}},\ }\bibfield  {title} {\enquote {\bibinfo {title}
  {Quasi-entropies for finite quantum systems},}\ }\href {\doibase
  10.1016/0034-4877(86)90067-4} {\bibfield  {journal} {\bibinfo  {journal}
  {Reports on Mathematical Physics}\ }\textbf {\bibinfo {volume} {23}},\
  \bibinfo {pages} {57--65} (\bibinfo {year} {1986})}\BibitemShut {NoStop}%
\bibitem [{\citenamefont {M\"{u}ller-Lennert}\ \emph
  {et~al.}(2013)\citenamefont {M\"{u}ller-Lennert}, \citenamefont {Dupuis},
  \citenamefont {Szehr}, \citenamefont {Fehr},\ and\ \citenamefont
  {Tomamichel}}]{sandwiched1}%
  \BibitemOpen
  \bibfield  {author} {\bibinfo {author} {\bibfnamefont {Martin}\ \bibnamefont
  {M\"{u}ller-Lennert}}, \bibinfo {author} {\bibfnamefont
  {Fr{\'{e}}d{\'{e}}ric}\ \bibnamefont {Dupuis}}, \bibinfo {author}
  {\bibfnamefont {Oleg}\ \bibnamefont {Szehr}}, \bibinfo {author}
  {\bibfnamefont {Serge}\ \bibnamefont {Fehr}}, \ and\ \bibinfo {author}
  {\bibfnamefont {Marco}\ \bibnamefont {Tomamichel}},\ }\bibfield  {title}
  {\enquote {\bibinfo {title} {On quantum r{\'{e}}nyi entropies: A new
  generalization and some properties},}\ }\href {\doibase 10.1063/1.4838856}
  {\bibfield  {journal} {\bibinfo  {journal} {Journal of Mathematical Physics}\
  }\textbf {\bibinfo {volume} {54}},\ \bibinfo {pages} {122203} (\bibinfo
  {year} {2013})}\BibitemShut {NoStop}%
\bibitem [{\citenamefont {Wilde}\ \emph {et~al.}(2014)\citenamefont {Wilde},
  \citenamefont {Winter},\ and\ \citenamefont {Yang}}]{sandwiched2}%
  \BibitemOpen
  \bibfield  {author} {\bibinfo {author} {\bibfnamefont {Mark~M.}\ \bibnamefont
  {Wilde}}, \bibinfo {author} {\bibfnamefont {Andreas}\ \bibnamefont {Winter}},
  \ and\ \bibinfo {author} {\bibfnamefont {Dong}\ \bibnamefont {Yang}},\
  }\bibfield  {title} {\enquote {\bibinfo {title} {Strong converse for the
  classical capacity of entanglement-breaking and hadamard channels via a
  sandwiched r{\'{e}}nyi relative entropy},}\ }\href {\doibase
  10.1007/s00220-014-2122-x} {\bibfield  {journal} {\bibinfo  {journal}
  {Communications in Mathematical Physics}\ }\textbf {\bibinfo {volume}
  {331}},\ \bibinfo {pages} {593--622} (\bibinfo {year} {2014})}\BibitemShut
  {NoStop}%
\bibitem [{\citenamefont {Matsumoto}(2018)}]{geometric}%
  \BibitemOpen
  \bibfield  {author} {\bibinfo {author} {\bibfnamefont {Keiji}\ \bibnamefont
  {Matsumoto}},\ }\bibfield  {title} {\enquote {\bibinfo {title} {A new quantum
  version of f-divergence},}\ }in\ \href {\doibase
  10.1007/978-981-13-2487-1_10} {\emph {\bibinfo {booktitle} {Springer
  Proceedings in Mathematics {\&} Statistics}}}\ (\bibinfo  {publisher}
  {Springer Singapore},\ \bibinfo {year} {2018})\ pp.\ \bibinfo {pages}
  {229--273}\BibitemShut {NoStop}%
\bibitem [{\citenamefont {Fawzi}\ and\ \citenamefont {Fawzi}(2020)}]{sharp}%
  \BibitemOpen
  \bibfield  {author} {\bibinfo {author} {\bibfnamefont {Hamza}\ \bibnamefont
  {Fawzi}}\ and\ \bibinfo {author} {\bibfnamefont {Omar}\ \bibnamefont
  {Fawzi}},\ }\href@noop {} {\enquote {\bibinfo {title} {Defining quantum
  divergences via convex optimization},}\ } (\bibinfo {year} {2020}),\ \Eprint
  {http://arxiv.org/abs/arXiv:2007.12576} {arXiv:2007.12576} \BibitemShut
  {NoStop}%
\bibitem [{\citenamefont {Donald}(1986)}]{measured1}%
  \BibitemOpen
  \bibfield  {author} {\bibinfo {author} {\bibfnamefont {Matthew~J.}\
  \bibnamefont {Donald}},\ }\bibfield  {title} {\enquote {\bibinfo {title} {On
  the relative entropy},}\ }\href {\doibase 10.1007/bf01212339} {\bibfield
  {journal} {\bibinfo  {journal} {Communications in Mathematical Physics}\
  }\textbf {\bibinfo {volume} {105}},\ \bibinfo {pages} {13--34} (\bibinfo
  {year} {1986})}\BibitemShut {NoStop}%
\bibitem [{\citenamefont {Cooney}\ \emph {et~al.}(2016)\citenamefont {Cooney},
  \citenamefont {Mosonyi},\ and\ \citenamefont {Wilde}}]{qrd_channels1}%
  \BibitemOpen
  \bibfield  {author} {\bibinfo {author} {\bibfnamefont {Tom}\ \bibnamefont
  {Cooney}}, \bibinfo {author} {\bibfnamefont {Mil{\'{a}}n}\ \bibnamefont
  {Mosonyi}}, \ and\ \bibinfo {author} {\bibfnamefont {Mark~M.}\ \bibnamefont
  {Wilde}},\ }\bibfield  {title} {\enquote {\bibinfo {title} {Strong converse
  exponents for a quantum channel discrimination problem and
  quantum-feedback-assisted communication},}\ }\href {\doibase
  10.1007/s00220-016-2645-4} {\bibfield  {journal} {\bibinfo  {journal}
  {Communications in Mathematical Physics}\ }\textbf {\bibinfo {volume}
  {344}},\ \bibinfo {pages} {797--829} (\bibinfo {year} {2016})}\BibitemShut
  {NoStop}%
\bibitem [{\citenamefont {Leditzky}\ \emph {et~al.}(2018)\citenamefont
  {Leditzky}, \citenamefont {Kaur}, \citenamefont {Datta},\ and\ \citenamefont
  {Wilde}}]{qrd_channels2}%
  \BibitemOpen
  \bibfield  {author} {\bibinfo {author} {\bibfnamefont {Felix}\ \bibnamefont
  {Leditzky}}, \bibinfo {author} {\bibfnamefont {Eneet}\ \bibnamefont {Kaur}},
  \bibinfo {author} {\bibfnamefont {Nilanjana}\ \bibnamefont {Datta}}, \ and\
  \bibinfo {author} {\bibfnamefont {Mark~M.}\ \bibnamefont {Wilde}},\
  }\bibfield  {title} {\enquote {\bibinfo {title} {Approaches for approximate
  additivity of the holevo information of quantum channels},}\ }\href {\doibase
  10.1103/PhysRevA.97.012332} {\bibfield  {journal} {\bibinfo  {journal} {Phys.
  Rev. A}\ }\textbf {\bibinfo {volume} {97}},\ \bibinfo {pages} {012332}
  (\bibinfo {year} {2018})}\BibitemShut {NoStop}%
\bibitem [{\citenamefont {Gour}\ and\ \citenamefont
  {Wilde}(2018)}]{channel_entropy}%
  \BibitemOpen
  \bibfield  {author} {\bibinfo {author} {\bibfnamefont {Gilad}\ \bibnamefont
  {Gour}}\ and\ \bibinfo {author} {\bibfnamefont {Mark~M.}\ \bibnamefont
  {Wilde}},\ }\href@noop {} {\enquote {\bibinfo {title} {Entropy of a quantum
  channel},}\ } (\bibinfo {year} {2018}),\ \Eprint
  {http://arxiv.org/abs/arXiv:1808.06980} {arXiv:1808.06980} \BibitemShut
  {NoStop}%
\bibitem [{\citenamefont {Berta}\ \emph {et~al.}(2017)\citenamefont {Berta},
  \citenamefont {Fawzi},\ and\ \citenamefont {Tomamichel}}]{measured3}%
  \BibitemOpen
  \bibfield  {author} {\bibinfo {author} {\bibfnamefont {Mario}\ \bibnamefont
  {Berta}}, \bibinfo {author} {\bibfnamefont {Omar}\ \bibnamefont {Fawzi}}, \
  and\ \bibinfo {author} {\bibfnamefont {Marco}\ \bibnamefont {Tomamichel}},\
  }\bibfield  {title} {\enquote {\bibinfo {title} {On variational expressions
  for quantum relative entropies},}\ }\href {\doibase
  10.1007/s11005-017-0990-7} {\bibfield  {journal} {\bibinfo  {journal}
  {Letters in Mathematical Physics}\ }\textbf {\bibinfo {volume} {107}},\
  \bibinfo {pages} {2239--2265} (\bibinfo {year} {2017})}\BibitemShut {NoStop}%
\bibitem [{\citenamefont {Hiai}\ and\ \citenamefont {Petz}(1991)}]{measured2}%
  \BibitemOpen
  \bibfield  {author} {\bibinfo {author} {\bibfnamefont {Fumio}\ \bibnamefont
  {Hiai}}\ and\ \bibinfo {author} {\bibfnamefont {D{\'{e}}nes}\ \bibnamefont
  {Petz}},\ }\bibfield  {title} {\enquote {\bibinfo {title} {The proper formula
  for relative entropy and its asymptotics in quantum probability},}\ }\href
  {\doibase 10.1007/bf02100287} {\bibfield  {journal} {\bibinfo  {journal}
  {Communications in Mathematical Physics}\ }\textbf {\bibinfo {volume}
  {143}},\ \bibinfo {pages} {99--114} (\bibinfo {year} {1991})}\BibitemShut
  {NoStop}%
\bibitem [{\citenamefont {Brand\~ao}(2005)}]{WoE_Brandao}%
  \BibitemOpen
  \bibfield  {author} {\bibinfo {author} {\bibfnamefont {Fernando G. S.~L.}\
  \bibnamefont {Brand\~ao}},\ }\bibfield  {title} {\enquote {\bibinfo {title}
  {Quantifying entanglement with witness operators},}\ }\href {\doibase
  10.1103/PhysRevA.72.022310} {\bibfield  {journal} {\bibinfo  {journal} {Phys.
  Rev. A}\ }\textbf {\bibinfo {volume} {72}},\ \bibinfo {pages} {022310}
  (\bibinfo {year} {2005})}\BibitemShut {NoStop}%
\bibitem [{\citenamefont {Tirone}\ \emph {et~al.}(2020)\citenamefont {Tirone},
  \citenamefont {Ghio}, \citenamefont {Livieri}, \citenamefont {Giovannetti},\
  and\ \citenamefont {Marmi}}]{salvatore}%
  \BibitemOpen
  \bibfield  {author} {\bibinfo {author} {\bibfnamefont {Salvatore}\
  \bibnamefont {Tirone}}, \bibinfo {author} {\bibfnamefont {Maddalena}\
  \bibnamefont {Ghio}}, \bibinfo {author} {\bibfnamefont {Giulia}\ \bibnamefont
  {Livieri}}, \bibinfo {author} {\bibfnamefont {Vittorio}\ \bibnamefont
  {Giovannetti}}, \ and\ \bibinfo {author} {\bibfnamefont {Stefano}\
  \bibnamefont {Marmi}},\ }\href@noop {} {\enquote {\bibinfo {title} {Kelly
  betting with quantum payoff: a continuous variable approach},}\ } (\bibinfo
  {year} {2020}),\ \Eprint {http://arxiv.org/abs/arXiv:2001.11395}
  {arXiv:2001.11395} \BibitemShut {NoStop}%
\bibitem [{\citenamefont {Pfister}(2019)}]{thesis_CP}%
  \BibitemOpen
  \bibfield  {author} {\bibinfo {author} {\bibfnamefont {Christoph}\
  \bibnamefont {Pfister}},\ }\emph {\bibinfo {title} {On Renyi Information
  Measures and Their Applications}},\ \href {\doibase 10.3929/ETHZ-B-000393481}
  {Ph.D. thesis} (\bibinfo {year} {2019})\BibitemShut {NoStop}%
\bibitem [{\citenamefont {Jensen}(1906)}]{jensen}%
  \BibitemOpen
  \bibfield  {author} {\bibinfo {author} {\bibfnamefont {J.~L. W.~V.}\
  \bibnamefont {Jensen}},\ }\bibfield  {title} {\enquote {\bibinfo {title} {Sur
  les fonctions convexes et les in{\'{e}}galit{\'{e}}s entre les valeurs
  moyennes},}\ }\href {\doibase 10.1007/bf02418571} {\bibfield  {journal}
  {\bibinfo  {journal} {Acta Mathematica}\ }\textbf {\bibinfo {volume} {30}},\
  \bibinfo {pages} {175--193} (\bibinfo {year} {1906})}\BibitemShut {NoStop}%
\bibitem [{\citenamefont {Sion}(1958)}]{sion1}%
  \BibitemOpen
  \bibfield  {author} {\bibinfo {author} {\bibfnamefont {Maurice}\ \bibnamefont
  {Sion}},\ }\bibfield  {title} {\enquote {\bibinfo {title} {On general minimax
  theorems},}\ }\href {\doibase 10.2140/pjm.1958.8.171} {\bibfield  {journal}
  {\bibinfo  {journal} {Pacific Journal of Mathematics}\ }\textbf {\bibinfo
  {volume} {8}},\ \bibinfo {pages} {171--176} (\bibinfo {year}
  {1958})}\BibitemShut {NoStop}%
\bibitem [{\citenamefont {Komiya}(1988)}]{sion2}%
  \BibitemOpen
  \bibfield  {author} {\bibinfo {author} {\bibfnamefont {Hidetoshi}\
  \bibnamefont {Komiya}},\ }\bibfield  {title} {\enquote {\bibinfo {title}
  {Elementary proof for sion{\textquotesingle}s minimax theorem},}\ }\href
  {\doibase 10.2996/kmj/1138038812} {\bibfield  {journal} {\bibinfo  {journal}
  {Kodai Mathematical Journal}\ }\textbf {\bibinfo {volume} {11}} (\bibinfo
  {year} {1988}),\ 10.2996/kmj/1138038812}\BibitemShut {NoStop}%
\end{thebibliography}%

\end{document}